\documentclass[11pt]{article} 
\usepackage[margin=1in]{geometry} 
\usepackage{amsmath, amssymb, amsthm} 
\usepackage{mathrsfs}
\usepackage{bbm} 
\usepackage{hyperref} 
\usepackage[dvipsnames]{xcolor} 
\usepackage{dsfont} 
\usepackage{natbib}
\usepackage{tikz} 
\usetikzlibrary{shapes.geometric, arrows.meta, calc}
\usepackage{enumitem}
\usepackage{cleveref}
\usepackage{comment}
\usepackage{thmtools, thm-restate}
\usepackage{subcaption} 
\usepackage[bb=libus]{mathalpha}
\usepackage{booktabs}  

\usepackage{float}

\newcommand{\one}{\mathbf{1}}

\usepackage[textwidth=3cm]{todonotes}
\newcommand{\PMmx}{\ensuremath{{\mathbb{P}\left(M=m \mid x\right)}}}
\newcommand{\PMmxm}{\ensuremath{{\mathbb{P}\left(M=m \mid  x^{(m)}\right)}}}
\newcommand{\PMzerox}{\ensuremath{{\mathbb{P}\left(M=\mathbf{0} \mid x\right)}}}

\newcommand{\E}{\ensuremath{{\mathbb E}}}
\newcommand{\Var}{\ensuremath{{\mathbb V}}}
\newcommand{\R}{\ensuremath{{\mathbb R}}}
\newcommand{\Prob}{\ensuremath{{\mathbb P}}}

\newcommand{\fulldata}{\mathcal{S}_n}
\newcommand{\data}{\tilde{\mathcal{S}}_n}

\newcommand{\Pjointn}{\ensuremath{\mathbb{P}_0^n}{}}
\newcommand{\Pjoint}{\ensuremath{\mathbb{P}_0}}
\newcommand{\Pthetajoint}{\ensuremath{\mathbb{P}}}
\newcommand{\Gauss}[2]{\mathcal{N}\left(#1,#2\right)}

\newcommand{\claim}[1]{\par\medskip\noindent\textit{#1}\ }

\newcommand{\KLMAR}{\ensuremath{\widetilde{\textnormal{KL}}}}

\newcommand{\KL}{\ensuremath{{\textnormal{KL}}}}
\newcommand{\He}{\ensuremath{\textnormal{H}}}
\newcommand{\HeMAR}{\ensuremath{\widetilde{\textnormal{H}}}}

\newcommand{\PMAR}{\Prob^{\scriptscriptstyle\mathrm{MAR}}}
\newcommand{\PMNAR}{\Prob^{\scriptscriptstyle\mathrm{MNAR}}}

\newcommand{\logm}{\ensuremath{\log_{-}}}
\newcommand{\Pnsieve}{\mathcal{P}_n}
\newcommand{\Pn}{\ensuremath{\mathbb{P}_n}}

\newcommand{\Pobsloc}[1]{{\widetilde{\mathcal{P}}_{n,#1}}}
\newcommand{\Pnloc}[1]{{\mathcal{P}_{n,#1}}}

\newcommand{\order}{\lesssim}
\DeclareMathOperator*{\argmax}{arg\,max}
\newcommand{\Jt}{J}
\newcommand{\Vt}{\widetilde{V}}

\newcommand{\densityspace}{\mathcal{D}}
\newcommand{\densityspacemodel}{\mathcal{P}}
\newcommand{\Modelspace}{\mathscr{P}}
\newcommand{\Solutionspace}{\mathcal{P}_{\star}}

\newcommand{\m}[2]{\ell_{#1, #2}}
\newcommand{\tildem}[2]{\tilde{\ell}_{#1, #2}}
\newcommand{\mm}[2]{\ell^{(m)}_{#1, #2}}

\newcommand{\score}[1]{\Psi_{#1}}

\newcommand{\dir}{\mathrm{d}}

\usepackage{algorithm}
\usepackage{algpseudocode}

\newcommand{\CIV}{C_4}

\newcommand{\N}{N}
\newcommand{\NI}{N_1}
\newcommand{\NII}{N_2}

\title{Asymptotics of Nonparametric Estimation under General Non-monotone MAR Missingness: A Nonparametric Maximum Likelihood Approach}

\author{Yating Zou$^1$, Huimin Hu$^2$, Gitte Kremling$^3$, Jeffrey Näf$^2$, \vspace{0.2cm} \\
$ ^1$ University of North Carolina at Chapel Hill \\
        $ ^2$Research Institute for Statistics and Information Science,\\
        University of Geneva \vspace{0.1cm}\\
        $ ^3$  Department of Mathematics,\\
    University of Hamburg\\
}
\date{}

\newtheorem{thm}{Theorem}[section] 
\newtheorem{asm}{Assumption}[section] 
\newtheorem{prop}[thm]{Proposition} 
\newtheorem{lemma}[thm]{Lemma} 
\newtheorem{cor}[thm]{Corollary} 
\newtheorem{dfn}{Definition}[section]
\newtheorem{exm}{Example}[section] 

\theoremstyle{remark}
\newtheorem{remark}{Remark}

\AtBeginEnvironment{remark}{\pushQED{\qed}}
\AtEndEnvironment{remark}{\popQED}

\begin{document}

\maketitle

\begin{abstract}
Missing data are pervasive in empirical research, and the missing at random (MAR) condition has been central to their analysis since its introduction fifty years ago. Yet theoretical guarantees remain limited, particularly when missingness is non-monotone.
Where they exist, they are confined to simplified settings such as monotone or block-wise missingness, or rest on restrictive assumptions about the missingness mechanism. 
In this paper, we establish a general rate of convergence result for an ignoring sieve MLE that holds without these restrictions. It requires only a natural positivity condition and a condition strictly weaker than (possibly non-monotone) MAR, which we refer to as ``{MAR-like MNAR}''. We moreover show that, even if this condition is violated, the resulting bias may still be bounded by a measure of deviation from MAR-like MNAR, which is of order $\varepsilon$ under Huber-type $\varepsilon$-MNAR-contamination of a MAR mechanism. 
Our result can be used broadly to design new algorithms for data generation and imputation. 
We illustrate this through a density estimator for the H\"{o}lder class and show that it achieves the complete-data minimax rate for any prescribed smoothness level, up to a logarithmic factor. The estimator is approximated in practice by a simple expectation-maximization (EM) algorithm and performs comparably to the kernel density estimator supplied with the complete data. We illustrate the utility of the new estimator in several applications, including clustering under MAR missingness.
\end{abstract}

\section{Introduction}\label{sec_Intro}

In any data analysis, some values might not be recorded. Such missing values represent a pervasive challenge across empirical disciplines. Far from being a mere nuisance, they can introduce systematic biases that severely compromise the validity of statistical inference. The formal statistical analysis of missing values arguably started with the inception of missing at random (MAR) in 1976 by Rubin in the seminal paper \cite{Rubin_Inferenceandmissing}. MAR allows missingness to depend on the data only through the recorded values, neither as restrictive as missing completely at random (MCAR), which requires missingness to be unrelated to the data altogether, nor as permissive as missing not at random (MNAR), which prohibits identification in general \citep{ourresult}. However, there has been some confusion about the MAR condition, what it implies and what it does not. As such, despite a wealth of research and despite the development of methods with theoretical guarantees in special instances of MAR, such as when the probability of missingness only depends on a subset of fully-observed variables (\cite{EmpiricalLik1, EmpiricalLik2, Quantile2012, Quantile2014, Quantile2015, EmpiricalLik3, han2019quantile, Responsequantileimputation,imputationpoweredinference, pervin2026calibrationframeworkinferencepartially}), developing (nonparametric) estimation procedures with statistical guarantees under general MAR is challenging. This is partly due to the fact that MAR does not restrict the possible patterns, so that missingness can be non-monotone, as illustrated in Figure \ref{fig:illustrationfocond}. In this case, the seemingly simple MAR condition presents several (estimation) issues, as indicated by the lack of theoretical-grounded methods and the search for alternatives to MAR (\cite{Malinsky2022, chen2022pattern, MNARcontamination, MMDmissingnesspaper}).

Despite this inherent difficulty, MAR rose to prominence due to the ``ignorability result'' in \cite{Rubin_Inferenceandmissing}; using maximum likelihood estimation (MLE) and Bayesian analysis, the missing data mechanism (MDM) can be ``ignored'' under MAR and an additional condition on the parameter space, leading to the so-called ``ignoring MLE'', a version of MLE where only the observed marginal distributions are modeled. Though this ignorability result itself does not imply anything about the statistical properties of the ignoring MLE, it points towards a special relationship of MAR with likelihood modeling. Indeed, although the asymptotic validity of the ignoring MLE was widely assumed and used, a formal proof of consistency and asymptotic normality under explicit conditions only appeared in \cite{Takai2013}. However, these results do not extend to nonparametric estimation and, in particular, not to general M-estimation \citep{Mestimatormissingvalues}. Direct estimation strategies to remedy this, such as inverse probability weighting (IPW), have strong theoretical guarantees, but appear infeasible under general MAR without restrictive assumptions on the form of the propensity scores (see, e.g., the discussion in \cite{tsiatis2006semiparametric, MARinverseweighting, chen2022pattern, Bayes}, and Section \ref{Sec_Motivating}). The same hold true for alternative approaches such as the empirical likelihood weighting of \cite{EmpiricalLik4}. On the other hand, in practice, data is often ``imputed'', filling the missing values in a way that should (implicitly) recover the distribution of the unobserved complete-data \cite{OTimputation, näf2024goodimputationmarmissingness}, to be able to use the resulting sample for any desired analysis in a second step. Instead of imputation, recent work also has studied methods to generate a complete sample from a sample with missing values \citet{misgan, missdiff, FLOWGEM}. Such samples can essentially be seen as Monte-Carlo sampling to obtain estimators, as already described in \cite{RichardDensityEstimation} in the case of complete data. Compared to (A)IPW approaches, imputation and generation can be easily used even for non-monotone missingness pattern, and do not require an estimate of the propensity scores. However, there appears to be no (nonparametric) imputation or generation method with theoretical guarantees under MAR, and this may lead to methods with significant bias, as discussed, e.g., in \cite{näf2024goodimputationmarmissingness, OneBenchmarktorulethemall}. Naturally, these issues are all exacerbated when instead of MAR, the missingness process is MNAR. However, anecdotal evidence suggests that the performance of an empirically successful MAR (imputation) method is often remarkably robust to changing to an MNAR mechanism, and recent literature studies the robustness of estimators to MNAR deviations from the \emph{MCAR} assumption \citep{MNARcontamination, MMDmissingnesspaper}. 

Consequently, questions such as ``are there imputation/generation methods that can consistently (and nonparametrically) recover the complete-data under MAR?', ``How does one obtain a consistent M-estimator under non-monotone MAR, without parametric assumptions on the propensity scores?" and ``How could one quantify deviations from MAR in a way that makes it possible to obtain robustness guarantees depending on the extend of this deviation'' have not yet been answered in a frequentist setting, to the best of our knowledge. The only general theoretical results we are aware of were recently presented in \cite{Bayes}, providing an asymptotic contraction rate result for nonparametric Bayesian estimation under MAR. These results show in particular that, despite non-monotone MAR missingness, the density of a set of independently and identically distributed data can be consistently estimated, in the sense that the posterior contracts around the true density. 

In this paper, we start by deriving an analogous result in a frequentist setting. Using the principle of the ``sieve'' MLE, we derive general nonparametric rate results under MAR. Just as in \cite{Bayes} we achieve the same rates as in the complete-data case, only increased by a constant multiple related to a natural positivity assumption. We then apply this general result to a new ignoring density estimator that is able to take missing data and return a consistent density estimate. While this density estimate is of separate interest, it can also be used to generate a sample without missing values, which is what we referred to as ``generation'' above. Given the difficulty of direct estimation, being able to recover the complete-data distribution is particularly important under general non-monotone MAR. We also show that our approach can be used to impute data under MAR to obtain an imputed sample whose distribution is consistent in terms of Hellinger distance. This appears to be the first result to show distributional convergence for nonparametric imputation under non-monotone MAR. Since imputation retains the originally observed entries together with their missing patterns, an imputed sample is a sample of the data, missing pattern pair rather than of the data alone. As such, intuitively, imputation should be more difficult than generation. We show that this is indeed the case, by proving that our density estimator remains consistent under a strictly weaker assumption than MAR, which we will refer to as ``MAR-like MNAR''. In contrast, for successful imputation with our method, MAR is necessary. Moreover, even if the missing mechanism is not MAR-like MNAR, we show that the (irreducible) error of the estimation under MNAR can be bounded by a natural new measure of deviation from MAR-like MNAR. In doing so, we generalize the parametric results in \cite{MMDmissingnesspaper}, who study the irreducible error under deviations from MCAR, to the case of nonparametric estimation under deviations from MAR. To develop these results, we use the concept of a ``MAR representation'' for the observed-data distribution, generated from a possibly MNAR process, connecting our work to \cite{gill1997coarsening, ourresult}.

In sum, our nonparametric ignoring approach makes it possible to obtain consistency under minimal assumption on the missingness mechanism; all we need is MAR-like MNAR and a natural positivity condition. In particular, we can obtain results under general non-monotone MAR; we do not need to make assumptions about dependencies between variables (such as in the graphical modeling literature, see e.g., \cite{nabi2025define} and the references therein), or relationship of patterns (such as in the general methodology introduced in \cite{chen2022pattern}); and we do not need to adapt our estimation procedure on identifying assumptions (such as in \cite{nonparametricpmm, chen2022pattern}) or to rely on potentially unstable inverse probability weighting. Even if MAR or its weaker counterpart do not hold, we can provide bounds on the irreducible error, depending on a natural measure of deviation from MAR. Finally, compared to \cite{FLOWGEM}, who combine the ignoring approach with Wasserstein gradient flows, without providing consistency results, we do not need to assume any smoothness conditions on the missingness mechanism.\footnote{Though not directly spelled out in \cite{FLOWGEM}, these smoothness conditions arise naturally in the approximation of the gradient.} Figure \ref{fig:overview} illustrates the main results in our paper.

Aside from the assumptions on missingness, our sieve MLE result relies on rather standard and mild assumptions. In particular, compared to the standard result in \cite[Theorem 3.4.12]{vandervaart2023weak}, we are able to utilize the recent result in \cite{kaji2026hellinger} to avoid the stringent assumption of bounded density ratios. The theorem allows one to show rate consistency of a large class of (new) ignoring estimators. Further, our density estimator result may be of independent interest, even for complete-data. We use a sieve of finite Gaussian mixtures inspired by \cite[Chapter 9]{Fundamentals}, with a covariance matrix whose eigenvalues are allowed to diverge with $n$. Although in practice this estimator is simply an estimator of a finite Gaussian mixture, which has been studied for complete-data \cite{EM0, EM1, EM2, MixtureofGaussians}, we show that it is able to estimate a density satisfying a Hölder condition with the minimax rate of estimation for any level of smoothness without the assumption of bounded density ratios. Although ignoring estimators in parametric MLE are known to increase computational complexity, compared to the computationally heavy estimator introduced in \cite{Bayes}, the estimator presented here allows for closed-form iterations and is quick to calculate. Moreover, a simple BIC criterion, that does not change the consistency properties of the algorithm, appears to be quite effective in choosing the number of components in practice. It is also easily possible to parallelize the computation of our estimator over the number of components, rendering the calculations highly scalable. As such, apart from providing the frequentist version of the general contraction rate results of \cite{Bayes}, this also opens the door to develop efficient and well-performing new estimators.


The paper proceeds as follows. Section \ref{sec_Background} introduces background and notation, while Section \ref{sec_literature_lontributions} discusses a motivating example, as well as our contributions. Section \ref{sec_contractionresults} then introduces the first main theorem, a general convergence result for the sieve MLE under MAR, while Section \ref{sec_densityestimation} derives the convergence result for density estimation and discusses implementation. Section \ref{sec_MNAR} then extend this to MAR-like MNAR and finally to a bound on the bias for general MNAR. Section \ref{sec_relatedliterature} dives deeper into related literature, complementing Section \ref{sec_Intro}. Finally, Sections \ref{sec_empirical} and \ref{sec_Application} consider simulations and several applications, while Section \ref{sec_conclusion} concludes. Figure \ref{fig:overview} illustrates the main convergence results we obtain for the ignoring sieve MLE. Code for the density estimation method as well as the experiments can be found in \url{https://github.com/huimin-hu-stat/MAR-v2}.

\begin{figure}
     \scalebox{0.85}{
    \begin{tikzpicture}[
  big/.style   = {ellipse, draw, thick, align=center, minimum width=6cm, minimum height=1.8cm},
  small/.style = {ellipse, draw, thick, align=center, inner sep=1pt,
                minimum width=3cm, minimum height=1.1cm},
  edge/.style  = {thick},
  bullet/.style= {circle, draw, fill=white, inner sep=1.3pt},
  item/.style  = {anchor=west, font=\small, inner sep=2pt},
]
    \node[big]   (root) at (0,-0.5)   {Ignoring Sieve MLE};
  \node[small] (L)    at (-5,-3) {\textbf{Generation}\\ (Recovering complete-data distribution $P_0$)};
  \node[small] (R)    at ( 5,-3) {\textbf{Imputation}\\ (Recovering joint distribution $\Pjoint$)};

  \draw[edge] (root.260) to[bend left=15]  ([xshift=0.5mm, yshift=0mm]L.north);
  \draw[edge] (root.280) to[bend right=15] ([xshift=-0.5mm,yshift=0mm]R.north);

  \foreach \n/\items in {
      L/{Consistency under MAR (\Cref{Thm:DensityResult}), Consistency under MAR-like MNAR (\Cref{cor:MARlike}),
      Error bound under MNAR (\Cref{cor:beyondW1})},
      R/{Consistency under MAR (\Cref{Cor:Imputationresult})}%
  }{
    \coordinate (s) at (\n.185);
    \foreach \t [count=\i] in \items {
      \draw[thick] ($(s)+(0,-0.8*\i+0.8)$) |- ++(0.5,-0.8)
        node[bullet] {} node[item, xshift=3pt] {\t};
    }
  }
\end{tikzpicture}
}
    \caption{Overview over the main results obtained with the ignoring sieve MLE approach.}
    \label{fig:overview}
\end{figure}

\section{Background and Notation}\label{sec_Background}

Let $X_1, \dots, X_n$ be random vectors taking values in a measurable space $(\mathcal{X}, \mathcal{A})$, $\mathcal{X} \subset \mathbb{R}^d$. The $X_i$'s are assumed to be independently and identically distributed (i.i.d.) from a distribution $P_0$ belonging to a statistical model $\Modelspace$. We assume that each $P \in \Modelspace$ admits a density $p$ with respect to Lebesgue's measure on $\mathcal{X}$. Instead of observing the $X_i$'s directly, for each $i \in \{1, \dots, n\}$, we only observe a subset of its $d$ components. Let $M_i \in \{0,1\}^d$ be a random binary mask that indicates which coordinates of $X_i$ are observed:
\[
(M_i)_j = 
\begin{cases}
0 & \text{if the } j\text{-th component of } X_i \text{ is observed}, \\
1 & \text{if the } j\text{-th component of } X_i \text{ is missing}.
\end{cases}
\]
We assume that the pairs $(X_i, M_i)$ are i.i.d.\ with a joint distribution $\mathbb{P}_{(X,M)}$ on $\mathcal{X} \times \{0,1\}^d$ that is assumed to be absolutely continuous with respect to the product of the Lebesgue and counting measures. We denote by $X_i^{(M_i)}$ the restriction of $X_i$ to its observed components, that is, the subvector containing only the entries for which $(M_i)_j = 0$. The available data thus consist of the partially observed vectors along with their positions 
$$
\data = \left\{ \left(X_1^{(M_1)},M_1\right),  \dots, \left(X_n^{(M_n)},M_n\right) \right\} \, .
$$
For convenience, we also define the corresponding (unobserved) complete data samples as 
$$
\fulldata = \left\{ \left(X_1,M_1\right),  \dots, \left(X_n,M_n\right) \right\} \, .
$$

We denote by \( \Prob \) the joint distribution of \( (X,M) \) when \( X \sim P \) and \( M \) follows the unknown conditional missingness mechanism. In particular, $(X,M) \sim \Pjoint$ and $\fulldata \sim \Pjointn$.

\subsection{Missingness at Random and Ignorability}
We focus in this paper on a specific family of missingness mechanisms usually referred to as \emph{Missing at Random}, which is encoded in the following assumption:

\begin{asm}[MAR]
\label{asm_true_MDM}
    The true conditional distribution of $M$ given $X$ is \emph{Missing at Random} (MAR), meaning that for almost every $x\in\mathcal{X}$, the probability mass function of $M$ given $X=x$ only depends on its observed components: for almost every $x\in\mathcal{X}$, for any $m\in\{0,1\}^d$, 
    $$
   \PMmx = \PMmxm .
    $$
\end{asm}
Assumption \ref{asm_true_MDM} might be seen as a ``population version'' of the original MAR version of \cite{Rubin_Inferenceandmissing}, or ``everywhere MAR''. For a discussion we refer to \cite{whatismeant3}. Crucially, it is the weakest MAR assumption, aside from the original definition in \cite{Rubin_Inferenceandmissing}. For instance, it is often assumed that $\Prob\left(M=m\mid X\right)$ only depends on a set of fully observed variables, a much stronger assumption (see e.g., the discussion in \cite{näf2024goodimputationmarmissingness}). In turn we will further weaken the MAR assumption in Section \ref{sec_MNAR}.

\vspace{0.2cm}
We do not assume anything on the structure of patterns in $M$, other than that the completely empty pattern $M=\mathbbm{1}$ (where $\mathbbm{1}=(1,\ldots,1)$) has probability zero:

\begin{asm}[Disregarding the empty pattern]
\label{asm_no_empty_MDM}
    We assume $\Prob(M=\mathbbm{1})=0$.
\end{asm}

In particular, missingness can be non-monotone, as illustrated in Figure \ref{fig:illustrationfocond}. The first example exhibits monotone missingness, but the second already violates this property, with $X_2$ missing in the second pattern and $X_3$ missing in the third. If the missingness mechanism $\Prob\left(M=m\mid X\right)$ depends only on $X_3$, learning $\Prob\left(M=m\mid X\right)$ and applying weighting estimators remains relatively straightforward. Under MAR, however, missingness may instead depend on $X_2$ in the second pattern and on $X_1$ in the third, which complicates matters considerably. This example is discussed further in Sections \ref{Sec_Motivating} and \ref{sec_empirical}.


\vspace{0.2cm}
We also introduce a technical assumption, standard in the literature, requiring the so-called \emph{propensity score}, defined as the function $x \mapsto \PMzerox$ giving the conditional probability of being fully observed, to be bounded away from zero:
\begin{asm}[Positivity of the propensity score]
\label{asm_positive_MDM}
    The density $p_0$ is positive Lebesgue-a.e.\ on $\mathcal{X}$, and the propensity score $x\mapsto \PMzerox$ is bounded away from zero: there exists some constant $c_0 > 0$ such that for almost every $x \in \mathcal{X}$, we have $\PMzerox \geq c_0$.
\end{asm}

This is a rather standard assumption for theoretical analysis of missing values (see, e.g., \cite{MARinverseweighting, Malinsky2022}). 

\begin{figure*}
    \centering
\begin{tikzpicture}


\node at (0,0) {$ \begin{pmatrix}
x_{1,1} & x_{1,2} & x_{1,3} \\
\textrm{NA} & x_{2,2} & x_{2,3} \\
\textrm{NA} & \textrm{NA} & x_{3,3}
\end{pmatrix}$};

\node at (4,0) {$\begin{pmatrix}
x_{1,1} & x_{1,2} & x_{1,3} \\
 \textrm{NA} & x_{2,2} & x_{2,3} \\
x_{3,1} & \textrm{NA} & x_{3,3}
\end{pmatrix}$};


\node at (8,0) {$\begin{pmatrix}
x_{1,1} & x_{1,2} & x_{1,3} \\
 x_{2,1} &\textrm{NA} & x_{2,3} \\
\textrm{NA} & x_{3,2} &\textrm{NA} 
\end{pmatrix}$};

\end{tikzpicture}
    \caption{Three data matrices with missing values, each with three different patterns. Each contains the fully observed pattern $M=\mathbf{0}$ and does not contain the completely unobserved pattern ($M\neq \mathbbm{1}$).}
    \label{fig:illustrationfocond}
\end{figure*}



Let $\Theta \subset \R^d$ be a set of parameters for a model $\Modelspace=\{P_{\theta}\}_\theta$. Since the pair $(X,M)$ is random, the central idea is to model its joint distribution by combining the complete-data model $\Modelspace$ for $X$ with a model for the missingness mechanism parameterized by some $\phi$. Given the partially observed dataset $\data$, the maximum likelihood estimator is obtained by maximizing the likelihood of $\data$ with respect to $(\theta,\phi)$. A fundamental result in the missing-data literature, known as the \emph{ignorability} property of the MLE and established by \cite{Rubin_Inferenceandmissing}, shows that inference based on the estimator
\begin{align}\label{eq_ignoring}
    \widetilde{\theta}_n = \arg\max_{\theta\in\Theta} \sum_{i=1}^{n} \log p_{\theta}^{(M_i)}\!\left(X_i^{(M_i)}\right) \, ,
\end{align}
where $p_{\theta}^{(m)}(x)$ denotes the marginal density of $x$ corresponding to the components observed under mask $m\in\{0,1\}^d$, is equivalent to the maximum likelihood estimator computed from the full observed dataset $\data$ under the MAR (Assumption~\ref{asm_true_MDM}), provided that $\theta$ and $\phi$ are distinct, meaning that they do not share any parameters. This ignoring MLE optimizes over $\theta$ while ignoring the missingness mechanism. It is particularly convenient in practice, as it avoids the often difficult task of specifying a model for the missingness mechanism. We now extend this principle into a fully nonparametric approach.

\subsection{Notation}\label{sec_notation}

We now introduce and summarize the notation used throughout the paper. We write ``almost all $x \in \mathcal{X}$'' or simple ``a.e.'' if a property holds for all $x \in \mathcal{X}$ in a set with Lebesgue measure one. Throughout, we assume that $p_0(x) > 0$ for almost all $x \in \mathcal{X}$.

\begin{itemize}
\item We assume to observe masked i.i.d. data along with their positions, $$\data = \left\{ \left(X_1^{(M_1)},M_1\right),  \dots, \left(X_n^{(M_n)},M_n\right) \right\}.$$

   \item $\Pjoint$ is the joint distribution of $(X,M)$, with $\fulldata \sim \Pjointn$. We denote the conditional distribution of $M \mid X$ as $\PMmx$ for $m \in \{0,1\}^d$ and almost all $x \in \mathcal{X}$, which is well defined on $\mathcal{X} \subset \R^d$. We assume that $\Pjoint$ is absolutely continuous with respect to the product of Lebesgue's and counting measures, with joint density 
   $$\mathbb{p}_0(x,m)=  p_0(x) \PMmx.$$
   
   \item We assume a model class $\Modelspace$ for the distribution of $X \sim P_0$, with $P_0 \in \Modelspace$ (i.e. the well-specified case). We assume each $P \in \Modelspace$ is absolutely continuous with respect to the Lebesgue measure with density $p$.

   \item We denote by $\densityspace$ the set of all densities on $\mathcal{X}$ with respect to the Lebesgue measure and by $\densityspacemodel \subset \densityspace$ the space of all densities associated with $P \in \Modelspace$.
   
  \item  In addition, $\Modelspace$ also induces a model class $\{\Pthetajoint: P \in \Modelspace\} $ of $(X,M)$, whereby $\Pthetajoint$ is the distribution induced by the density $(x,m) \mapsto  p(x) \PMmx$.
   \item For $m \in \{0,1\}^d$, $X^{(m)}$ is the subvector of $X$ corresponding to the variables such that $m_j=0$, while $X^{(\bar{m})}$ is the subvector of $X$ corresponding to the variables such that $m_j=1$. We denote the corresponding marginal densities as $p^{(m)}$ and $p^{(\bar{m})}$.
   \item We work under the convention $0/0 := 0$ and $c/0 := \infty$ for $c > 0$.
    \item For complete data, we denote by $\KL(P_1 \|P_2)$, the KL divergence between $P_1,P_2$,
    \begin{align*}
        \KL(P_1 \| P_2)=\E_{ P_1}\left[ \log \left( \frac{p_1(X)}{p_2(X)} \right) \right].
    \end{align*}
    We note that this expectation is always defined and nonnegative, though $\KL(P_1 \| P_2)=\infty$, if $P_1(\{p_2(X)=0\}) > 0$.
    \item For complete data, the Hellinger distance between $P_1, P_2$ is given as:
\[
\He^2(p_1, p_2) = \int \left( \sqrt{p_1}(x)- \sqrt{p_2}(x)\right)^2 \mathrm{d}x.
\]
\item  For a class $\mathcal{F}$ and an appropriate pseudometric $d$, write $N_{[\,]}(\varepsilon,\mathcal{F},d)$ for the
$\varepsilon$-bracketing number and
\[
  J_{[\,]}(\delta,\mathcal{F},d):=\int_0^\delta\sqrt{1+\log N_{[\,]}(\varepsilon,\mathcal{F},d)}\,
  d\varepsilon
\]
for the (nonstandardized) bracketing integral. It is increasing and concave in $\delta$, so
$J_{[\,]}(c\delta,\mathcal{F},d)\le c\,J_{[\,]}(\delta,\mathcal{F},d)$ for $c\ge1$.
\end{itemize}

Given any collection
$(f^{(m)})_{m\in\{0,1\}^d}$ of functions, each $f^{(m)}$ a measurable function of the observed
block $x^{(m)}$, define on the joint space
$\mathcal{X}\times\{0,1\}^d$ the measurable function
\begin{equation}\label{eq:glue}
  \tilde f(x,m)\;:=\;\sum_{m'} f^{(m')}(x^{(m')})\,\one\{m=m'\}.
\end{equation}
For $(X,M)\sim\Pjoint$ and the i.i.d.\ sample, this definition allows us to express
\begin{equation}\label{eq:ep}
  \E_{\Pjoint}\bigl[f^{(M)}(X^{(M)})\bigr]=\Pjoint\tilde f,
  \qquad
  \frac1n\sum_{i=1}^{n} f^{(M_i)}\bigl(X_i^{(M_i)}\bigr)=\Pn\tilde f,
\end{equation}
mirroring empirical-process notation.

\section{Motivating Example and Contributions}\label{sec_literature_lontributions}

Let $\Pnsieve \subset \densityspacemodel$ be a sieve for the complete-data density. The goal of this paper is to demonstrate that, under MAR and positivity (Assumptions \ref{asm_true_MDM}--\ref{asm_positive_MDM}) and further rather standard assumptions, a (approximate) maximizer $\hat{p}_n$ of the log-likelihood over the sieve $\Pnsieve$,
\begin{equation}
\hat{p}_n = \argmax_{p \in \Pnsieve} \sum_{i=1}^{n} \log p^{(M_i)}\!\left(X_i^{(M_i)}\right),
\label{eq:ourmle}
\end{equation}
has $d(\hat{p}_n,p_0)$ going to zero with a rate $\delta_n$. Thus, despite the MAR missingness, the nonparametric MLE is consistent. Moreover, the rate $\delta_n$ will be the same as in the complete data case up to the constant $1/\sqrt{c_0}$, where $c_0$ is the lower bound in Assumption \ref{asm_positive_MDM}.

We now discuss a motivating example to fix our ideas.

\subsection{Motivating Example: M(N)AR Challenge}\label{Sec_Motivating}

For $d=2$, consider $X \sim \Gauss{0}{\Sigma}$, where $\Gauss{\mu}{\Sigma}$ denotes the Gaussian distribution with mean $\mu$ and covariance matrix $\Sigma$. We then define the following missingness mechanisms: For $(m_1,m_2,m_3) = ((0,0), (0,1), (1,0))$, let
\[
\begin{gathered}[t]
    \text{(MAR)} \\
    \begin{aligned}
        &\Prob(M=m_1\mid x)=(G(x_1)+G(x_2))/3, \\
        &\Prob(M=m_2\mid x)=(2-G(x_1))/3, \\
        &\Prob(M=m_3\mid x)=(1-G(x_2))/3,
    \end{aligned}
\end{gathered}
\qquad\qquad
\begin{gathered}[t]
    \text{(MNAR)} \\
    \begin{aligned}
        &\Prob(M=m_1\mid x)=(G(x_1)+G(x_2))/3, \\
        &\Prob(M=m_2\mid x)=(2-G(x_2))/3, \\
        &\Prob(M=m_3\mid x)=(1-G(x_1))/3,
    \end{aligned}
\end{gathered}
\]
\[
\begin{gathered}[t]
    \text{(MAR-like MNAR)} \\
    \begin{aligned}
     &\Prob(M=m_1\mid x) = 1/2 \nonumber \\
     &     \Prob(M=m_2\mid x) = 1/4 - 1/2(\Phi(x_1) - 1/2)(\Phi(x_2) - 1/2)\nonumber \\ 
     &\Prob(M=m_3\mid x) = 1/4 + 1/2(\Phi(x_1) - 1/2)(\Phi(x_2) - 1/2) 
\label{eq_MARlikemissing}
    \end{aligned}
\end{gathered}
\]
with
    \begin{align}\label{eq_Gx}
            G(x)=\max(\Phi(x), 2c_0),
    \end{align}
and $\Phi$ the standard Gaussian distribution function. In the (MAR) and (MNAR) cases, $\Sigma$ can be any symmetric positive definite matrix, while (MAR-like MNAR) restricts $\Sigma$ to be diagonal. The (MAR) and (MNAR) mechanisms are adaptations of the original MAR challenge formulated in \cite{näf2024goodimputationmarmissingness, practical}. The mechanism in (MAR) meets the MAR condition~\eqref{asm_true_MDM}, while the mechanism in (MNAR) is an MNAR censoring mechanism. Finally, the (MAR-like MNAR) meets the condition weaker than MAR found in \Cref{def_MARlike}. In all cases, the conditional probability of sampling the fully observed pattern is strictly bounded away from zero ($\Prob(M=m_1\mid X=x) \geq c_0 > 0$). We do not claim that these examples are particularly realistic. Rather, their simplicity illustrates how easily non-monotone missingness becomes challenging. Even (MAR), the only mechanism that is not MNAR, is a rather complex non-monotone mechanism, designed to be difficult for imputation and estimation. 

The goal of the MAR challenge was to estimate a specific M-estimator, namely the $0.1$ quantile of $X_1$. Since there is no fully observed set of variables on which $\Prob(M=m \mid x)$ depends, quantile estimation methods such as in \cite{Quantile2012, Quantile2014, Quantile2015, han2019quantile} cannot be used. For the same reason, estimating $\Prob(M=m_1 \mid x)$ for IPW approaches is challenging. For instance, it is one of the MAR examples discussed in \cite{chen2022pattern}, which cannot be written in terms of regular pattern graphs. One of the very few methods one could use is that of \cite{MARinverseweighting}, but this comes at the price of a parametric assumption that is not met here. 
As such, while commonly used IPW estimators allow for asymptotic uncertainty quantification, they are biased in the MAR example above and, with MAR-based weights, in both MNAR examples. Moreover, IPW is primarily designed for finite-dimensional M-estimators. Extensions to infinite-dimensional targets such as densities \citep{Dubnicka2009} assume that missingness depends only on fully observed variables. 
On the other hand, while there already are (nonparametric) imputation and generation methods that perform well in the (MAR) example, such as a subset of MICE methods \citep{practical} and the recent method in \cite{FLOWGEM}, none of them provides (frequentist) guarantees. Moreover, given the heuristic nature of nonparametric imputation methods, many widely used and recent methods will exhibit strong biases in this example. We refer to \cite[Figure 2]{FLOWGEM} and \cite[Figure 5]{practical} for an overview. 

Inspired by imputation/generation methods, we first estimate the density using the method in Section \ref{sec_densityestimation} and then sample a new set of observations from this density on which the quantile can be (consistently) estimated, similar to the ideas in \cite{RichardDensityEstimation} with complete data. In contrast to heuristic methods, our ignoring sieve MLE method produces a consistent density estimator that can also be used for nonparametric imputation. The first graph in Figure \ref{fig:MotivatingExample} shows the result of the quantile estimation for (MAR), using generation and imputation, respectively. Interestingly, for consistent imputation, we need MAR, while for generation, MAR-like MNAR is enough. This is demonstrated in the second picture for the (MAR-like MNAR) mechanism, where the quantile estimate produced by generation is still consistent. Finally, we might ask what happens if we deviate from this MAR-like MNAR mechanism to general MNAR. Using the contamination framework of \cite{MNARcontamination}, where the MAR mechanism gets contaminated by a fraction $\varepsilon$ of the MNAR mechanism in (MNAR), we show in \Cref{Cor:Imputationresult} that the irreducible error that we incur can be upper bounded by a function that increases smoothly in $\varepsilon$. This is demonstrated in the third graph in Figure \ref{fig:MotivatingExample}. Here we increase the contamination of (MAR) by (MNAR) with $\varepsilon$ increasing from 0 to 0.9.



\begin{figure}
    \centering
    \includegraphics[width=0.95\linewidth]{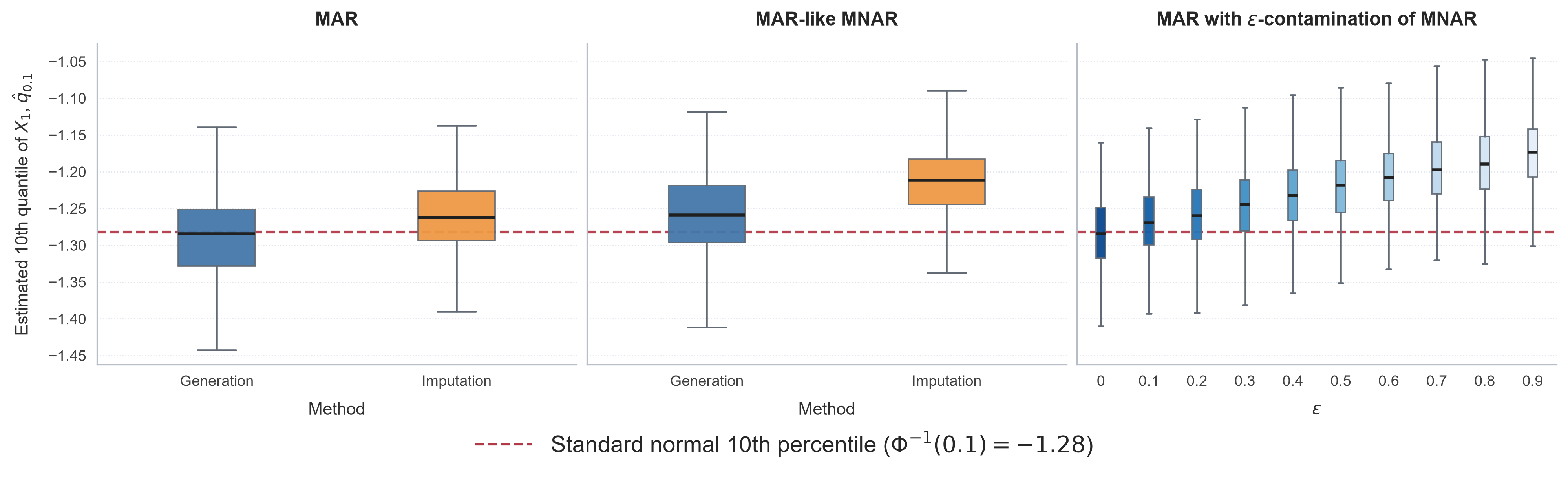}
        \caption{Estimation of the 0.1 quantile of $X_1$ , for $B=500$, with $n=2000$ observations each. Left: (MAR), Middle: (MAR-like MNAR), Right: (MAR) with increasing contamination of (MNAR). The dashed red line corresponds to the correct value.}
    \label{fig:MotivatingExample}
\end{figure}

\subsection{Contributions}

\begin{enumerate}
    \item We derive a general nonparametric sieve MLE estimation rate result under non-monotone MAR, showing that essentially the same rate can be attained as in the complete datasetting, up to the constant $1/\sqrt{c_0}$ related to the positivity assumption. We use recent results in \cite{kaji2026hellinger} to remove the stringent assumption of bounded density ratios.
    \item We apply this result to density estimation, developing a new (ignoring) estimator that achieves minimax optimal rates up to logarithmic factors in a Hölder class, inspired by the Bayesian estimator discussed in \cite{Fundamentals, Bayes}. Compared with the KDE on complete data, this estimator achieves this near-minimax rate for any smoothness parameter. Moreover, this estimator allows to draw new data points from the estimated complete-data distribution.
    \item We use this new estimator in the context of imputation, to provide a (nonparametric) rate-consistent imputation method under non-monotone MAR. This allows us to obtain a sample with a distribution that converges to the joint distribution $\Pjoint$.
    \item We go beyond MAR in two crucial ways. First, we show that our estimation results for $P_0$ still hold under a strictly weaker condition than MAR, MAR-like MNAR, proving that consistent estimation of $P_0$ is possible beyond MAR. Second, we use the contamination framework of \cite{MNARcontamination} to show that even when this weaker condition is violated, the ignoring estimation methodology can be robust to MNAR deviations.
\end{enumerate}

Point 4. provides theoretical support for some of the anecdotal evidence that a good MAR method appears to perform robustly even under MNAR. It is also interesting to note that consistently estimating $P_0$ itself can be done under a condition strictly weaker than MAR, while for imputation, where the goal is to consistently recover $\Pjoint$, MAR is needed. Aside from our theoretical contributions,  we implement the new ignoring density estimator using an Expectation Maximization (EM) algorithm that optimizes the parameters of the mixture distribution. This EM algorithm operates on the missing data directly and is not related to the usual parametric EM algorithms that use implicit imputations of the missing values, such as in \cite{EM0, EM2}.


\section{Consistency of the Sieve MLE under MAR}\label{sec_contractionresults}


We introduce the main result in this section, beginning with the adapted 
quantities it requires. First, the population value of a likelihood 
criterion is a Kullback--Leibler divergence. Since our MLE maximizes a criterion 
built from the marginals $p^{(m)}$, the relevant divergence is an adapted one, 
as in \cite{Bayes}:
\begin{dfn}[KL relative to $\Pjoint$]\label{definition_KL0}
We define for any $P \in \Modelspace$,
\[
\KLMAR(P_0 \| P):=\E_{\Pjoint}\left[ \log \left( \frac{p_0^{(M)}(X^{(M)})}{p^{(M)}(X^{(M)})} \right) \right].
\]
$\KLMAR$ is referred to as the Kullback-Leibler (KL) divergence relative to $\Pjoint$.
\end{dfn}



Minimizing the adapted divergence recovers $P_0$, the distribution of the complete data $X$, even though $\KLMAR$ is computed from the observed marginals alone. This is what makes $\eqref{eq:ourmle}$ a legitimate target:
\begin{restatable}[Proposition 4.2 \cite{Bayes}]{prop}{KLloss}\label{informationloss_KL}
Under Assumptions \ref{asm_true_MDM}, \ref{asm_no_empty_MDM}, for any $P \in \Modelspace$, it holds that
    \begin{align*}
     0 \leq   \KLMAR(P_0\| P) \leq  \KL(P_0\| P).
    \end{align*}
    Moreover, $\KL(P_0\| P) > 0$ implies $\KLMAR(P_0\| P)>0$ under the additional Assumption \ref{asm_positive_MDM}.
    In this case, we thus have
\[
\widetilde{\textnormal{KL}}(P_0 \| P) = 0 \quad \text{if and only if} \quad P = P_0 .
\]
\end{restatable}


For the second part, only the pointwise consequence $\PMzerox > 0$ of Assumption~\ref{asm_positive_MDM} is used; the uniform floor $c_0$ enters later, in the rate of Theorem~\ref{thm:8mar}. The pointwise condition is sharp; if $\PMzerox = 0$ on a set of positive $P_0$-measure, it is possible to construct a $P_1$ with $\KL(P_0\|P_1) > 0$ but $\KLMAR(P_0\|P_1) = 0$, so that $\KLMAR$ no longer identifies $P_0$.

The rate analysis requires a quantity that behaves as a distance and admits bracketing entropy bounds. The adapted Hellinger functional supplies one, generalizing the adapted Hellinger distance of \cite{Bayes} to arbitrary nonnegative integrable functions.

\begin{dfn}[Adapted Hellinger functional]\label{def:Ht}
For nonnegative integrable functions $f,g$ on $\mathcal X$, we define the \emph{adapted Hellinger functional} as
\begin{align}\label{eq:defHt}
  \HeMAR^2(f,g)
  :=\sum_{m}\int\Bigl(\sqrt{f^{(m)}(x^{(m)})}-\sqrt{g^{(m)}(x^{(m)})}\Bigr)^2
  \,\Prob(M=m\mid x^{(m)})\,dx^{(m)}, \;\text{ where }
\end{align}
\[
  f^{(m)}(x^{(m)}):=\int f\bigl(x^{(m)},x^{(\bar m)}\bigr)\,dx^{(\bar m)},
  \qquad
  g^{(m)}(x^{(m)}):=\int g\bigl(x^{(m)},x^{(\bar m)}\bigr)\,dx^{(\bar m)}
\]
are their observed-block marginals, obtained by integrating out the missing block
$x^{(\bar m)}$ (cf.\ \eqref{eq:marg}).
\end{dfn}


$\HeMAR$ is a pseudometric and crucially compatible with the pointwise partial order, allowing us to define bracketing numbers with respect to $\HeMAR$.

\begin{restatable}[Adapted Hellinger property]{prop}{propppseudometric}\label{prop:pseudometric}
For all nonnegative integrable functions $f,g,h,\ell,u$ on $\mathcal X$,
\[
  \HeMAR(f,f)=0,
  \qquad
  \HeMAR(f,g)=\HeMAR(g,f)\ge0,
  \qquad
  \HeMAR(f,h)\le\HeMAR(f,g)+\HeMAR(g,h), \text{ and }
\]
\[
\ell\le f,g\le u\ \Longrightarrow\ \HeMAR(f,g)\le\HeMAR(\ell,u).
\]
In particular $\HeMAR$ is a pseudometric compatible with the pointwise partial
order. Under Assumption~\ref{asm_positive_MDM} it is further a metric with $\HeMAR(f,g)=0$ implying $f=g$ almost everywhere.
 \end{restatable}


$\HeMAR$ relates to full-data Hellinger distance $\He$ in the following way. 
\begin{prop}[\citet{Bayes}, Proposition B.1]\label{prop:sandwich}
        Assume Assumptions \ref{asm_true_MDM}, \ref{asm_positive_MDM} hold true with $\Prob(M=\mathbf{0} \mid x) \geq c_0 > 0$. Then for any $p_1, p_2 \in \densityspacemodel$
    \begin{align*}
    c_0 \He^2(p_1, p_2)\leq \HeMAR^2(p_1, p_2)\leq \He^2(p_1, p_2).
\end{align*}
\end{prop}


For a radius
$r>0$ we write the localized sieves
\[
  \Pnloc{r}:=\bigl\{p\in\Pnsieve:\He(p_0,p)\le r\bigr\},
  \qquad
  \Pobsloc{r}:=\bigl\{p\in\Pnsieve:\HeMAR(p_0,p)\le r\bigr\},
\]
the complete-data $\He$-ball and the observed-data $\HeMAR$-ball, respectively. We recall the notation of Section \ref{sec_notation}, in particular, for $(X,M)\sim\Pjoint$,
\begin{equation*}
  \E_{\Pjoint}\bigl[f^{(M)}(X^{(M)})\bigr]=\Pjoint\tilde f,
  \qquad
  \frac1n\sum_{i=1}^{n} f^{(M_i)}\bigl(X_i^{(M_i)}\bigr)=\Pn\tilde f.
\end{equation*}
 We now apply this construction to likelihood
ratios. For a density $p$ on $\mathcal{X}$ and a fixed pattern
$m\in\{0,1\}^d$, write
\begin{equation}\label{eq:marg}
  p^{(m)}(x^{(m)})\;:=\;\int p\bigl(x^{(m)},x^{(\bar m)}\bigr)\,
  dx^{(\bar m)}
\end{equation}
for the marginal density of the observed block. The full-data and observed-data likelihood ratios are written as
\begin{equation}\label{eq:rho}
  \rho_{p_0,p}(x)\;:=\;\frac{p_0(x)}{p(x)},
  \qquad
  \rho^{(m)}_{p_0,p}(x^{(m)})
  \;:=\;\frac{p^{(m)}_0(x^{(m)})}{p^{(m)}(x^{(m)})}.
\end{equation}
For a Borel function $g:\R\to\R$, the collection
$(g(\rho^{(m)}_{p_0,p}))_m$ yields the ``glued'' form
$g(\tilde\rho_{p_0,p})$; likewise $(p^{(m)})_m$ yields
$\tilde p$. We can accordingly write
$\Pn\,g(\tilde\rho_{p_0,p})$,
$\Pjoint\,g(\tilde\rho_{p_0,p})$,
$\Pn\,g(\tilde p)$,
$\Pjoint\,g(\tilde p)$.



\subsection{General Rate Results under MAR}

The above preliminaries lead to our first main result.

\begin{restatable}[Rate of convergence of the observed-data sieve MLE under MAR]{thm}{mainthmMAR}\label{thm:8mar}
  Let $(M_1,X_1^{(M_1)}),$ $(M_2,X_2^{(M_2)})$ $,\dots$ be i.i.d.\ from the observed-data law
induced by $\Pjoint$ under Assumptions~\textup{\ref{asm_true_MDM}--\ref{asm_positive_MDM}}.
Let $\Pnsieve$ be a sequence of sieves and $\Pnloc{\delta}$ the $\He$-ball of
radius $\delta$ about $p_0$.
Suppose there exist $p_n \in \Pnsieve$, a sequence $\delta_n\ge0$, and a
constant $M_{\mathrm{MAR}}\in[0,\infty)$ such that:
\begin{align}
  \HeMAR(p_0,p_n) &\order \delta_n,
  \tag{4$'$}\label{eq:approx}\\[2pt]
    \E_{\Pjoint}\!\left[\Bigl(\log\rho_{ p_0, p_n}^{(M)}(X^{(M)})\Bigr)^2\,\one\{\rho_{ p_0, p_n}^{(M)}(X^{(M)})>4\}\right]
    &\le M_{\mathrm{MAR}}\,\delta_n^2,
  \tag{5$'$}\label{eq:tail}
\end{align}
and, for some constants $K<\infty$ and $\alpha_0<2$ such that for n large enough
\begin{equation}
  \Jt_{[\,]}\bigl(\delta,\Pnloc{\delta/\sqrt{c_0}},\He\bigr)\;\le\;K\sqrt n\,\delta_n^{2}\Bigl(\frac{\delta}{\delta_n}\Bigr)^{\alpha_0}
  \qquad\text{for all }\delta\in[\delta_n,\delta_n\vee\sqrt{2c_0}].
  \tag{6$'$}\label{eq:entropy}
\end{equation}
Then any approximate maximizer $\hat p_n\in\Pnsieve$ of the observed-data
likelihood $p\mapsto\prod_{i=1}^n p^{(M_i)}(X_i^{(M_i)})$, in the sense that
\begin{equation}
  \frac{1}{n}\sum_{i=1}^{n}\log \hat p_n^{(M_i)}\bigl(X_i^{(M_i)}\bigr)
  \;\ge\;
  \frac{1}{n}\sum_{i=1}^{n}\log p_n^{(M_i)}\bigl(X_i^{(M_i)}\bigr)-O_P(\delta_n^2),
  \tag{7$'$}\label{eq:approxmax}
\end{equation}
satisfies
\[
  \He\bigl(p_0,\hat p_n\bigr)
  \;=\; O_P\;\!\Bigl(\tfrac{1}{\sqrt{c_0}}\bigl(\delta_n\vee n^{-1/2}\bigr)\Bigr).
\]
 \end{restatable}

 \bigskip

\begin{remark}[What MAR changes, and what it does not]\label{rem:main}
\
\begin{enumerate}[label=\textup{(\roman*)},leftmargin=2.4em]

\item \textit{Conditions: Only \eqref{eq:tail} must be verified through $p^{(M)}$.}
Condition \eqref{eq:approx} follows from its complete-data counterpart
$\He(p_0,p_n)\order\delta_n$, since $\HeMAR\le\He$ by Proposition~\ref{prop:sandwich}. 
Condition \eqref{eq:entropy} is the usual bracketing entropy condition, taken in $\He$ rather than in $\HeMAR$; the global condition 
$\Jt_{[\,]}(\delta_n,\Pnsieve,\He)\order\delta_n^{2}\sqrt n$
is stronger but sufficient. 
Condition \eqref{eq:tail} constrains the observed-block ratio
$\rho^{(m)}_{p_0,p_n}=p_0^{(m)}/p_n^{(m)}$, 
which is a conditional average of the complete-data ratio,
\[
  \rho^{(m)}_{p_0,p_n}\bigl(x^{(m)}\bigr)
  =\int \rho_{p_0,p_n}\bigl(x^{(m)},x^{(\bar m)}\bigr)\,
    p_n\bigl(x^{(\bar m)}\mid x^{(m)}\bigr)\,dx^{(\bar m)} .
\]
A uniformly bounded complete-data ratio, as required by \citet[Theorem~3.4.12]{vandervaart2023weak}, together with \eqref{eq:approx} suffices. However, uniform boundedness can be restrictive. 
Consider the Gaussian location family $\{p_\theta = \phi(\cdot-\theta) : \theta \in \R\}$ on $\R$, where $\phi$ is the standard normal density and $\theta$ indexes the location, with $\theta_0$ the true value. The ratio $p_{\theta_0}/p_{\theta}$ equals
$\exp\{(\theta_0-\theta)x+(\theta^{2}-\theta_0^{2})/2\}$, unbounded whenever
$\theta\neq\theta_0$. Alternatively, \eqref{eq:tail} can be verified through a pointwise
majorant of $\rho^{(m)}$, as in the proof of Theorem~\ref{Thm:DensityResult}.

\item \textit{Conclusion: The rate is of complete-data order, and MAR enters only as a constant.}
This arises from the empirical-process step, where the criterion class $\widetilde{\mathcal M}_{n,\delta}$ (built from a transform of $\widetilde{\mathcal P}_{n,\delta}$) is built from a glued family of pattern marginals rather than from a single density. Its entropy is measured in the Bernstein norm, which is dominated by $\HeMAR$
(Proposition~\ref{prop:bn}), and $\HeMAR$-entropy in turn by $\He$-entropy
(Lemma~\ref{lem:entropyinherit}). Assumption~\ref{asm_positive_MDM} turns a rate in $\HeMAR$ into a rate for the full-data density; the factor $1/\sqrt{c_0}$ is the price of that promotion.

\item \textit{Scope: MAR does not narrow the scope of the nonparametric sieve MLE theory.}
Nothing additional is required of the sieve, and nothing beyond the ordinary likelihood construction is required of the estimator. No propensity model has to be fitted. 
On the missingness side, we require MAR itself, together with two minimal conditions, that we never observe nothing (Assumption~\eqref{asm_no_empty_MDM}), and that every unit has at least a fixed chance of being seen in full (Assumption~\eqref{asm_positive_MDM}). No further identification restriction is imposed. Missing patterns may be non-monotone, and no coordinate need be observed throughout. 
This permits application to a wide range of sieves, including the Gaussian mixture studied in Section~\ref{sec_densityestimation}.

\end{enumerate}
\end{remark}

Having established the general rate result, we now demonstrate how this can be applied to nonparametric density estimation under a Hölder assumption.

\subsection{Application: A new (ignoring) Density Estimator}\label{sec_densityestimation}


    

In the following, we apply the general rate result to derive a density estimation procedure under MAR based on Gaussian mixtures. Let $\phi(\cdot, \Sigma)$ be the density of the $\Gauss{0}{\Sigma}$-distribution where $\Sigma$ is a positive definite $d \times d$ matrix. Let moreover $I_d$ be the identity matrix in $d$ dimensions and
\[
p_{F, \Sigma}(x)=\int \phi(x-z, \Sigma) d F(z).
\]
We study the following sieve space adapted from \cite[Chapter 9.4]{Fundamentals}:
\begin{align*}
    \Pnsieve=\{p_{F, \Sigma}: F \in \mathcal{F}_{\N, a}, \Sigma \in \mathcal{D}_{\sigma} \}
\end{align*}
with $\delta_n$, $a=a(n)$, $\sigma=\sigma(n)$ positive and $\N=\N(n)$ integer,
\begin{align*}
\mathcal{F}_{\N, a}&=\left\{ \sum_{j=1}^{\N} w_j \delta_{z_j}: z_1, \ldots, z_{\N} \in [-a,a]^d, w \in \Delta_{\N} \right\}\\
    \mathcal{D}_{\sigma}&=\{\Sigma: \sigma^2 \leq \mathrm{eig}_{1}(\Sigma) \leq \mathrm{eig}_{d}(\Sigma) < \sigma^{-2}\},
\end{align*}
where $\delta_z$ denotes the Dirac measure at $z$, $\Delta_{\N}:=\bigl\{w\in[0,1]^{\N}:\sum_{j=1}^{\N}w_j=1\bigr\}$ is the $(\N-1)$-dimensional unit simplex, and $\mathrm{eig}_{1}(\Sigma), \ldots, \mathrm{eig}_{d}(\Sigma)$ denote the eigenvalues of a matrix $\Sigma$ ordered by size. In the following, $\sigma, \delta_n$ will decrease with $n$ at rates specified below, while $a$, $\N$ are intended to increase. We apply Theorem \ref{thm:8mar} with the sieve $\Pnsieve$. Let $t$ be a number with 
\begin{equation}\label{eq_tcondition}
    t > \frac{\beta d + \beta d/\tau + d + \beta}{2 \beta + d},
\end{equation}
$\CIV$ be a large enough constant, and
\begin{align}
\delta_n := C_4 n^{-\beta/(2\beta + d)} (\log(n))^t, \qquad    \N := \frac{n\delta_n^2}{\log(n\delta_n^2)},\qquad
    a := n \delta_n^2, \qquad
    \sigma^{-2} := n\delta_n^2.
    \label{eq:paramrates}
\end{align}


We moreover assume that the true density $p_0$ follows the same Hölder class as in \cite[Chapter 9]{Fundamentals}.

\begin{asm}[Hölder Class]\label{asm_Holder}
Let $k=(k_1, \ldots, k_d) \in \mathbb{N}^d$ be a multi-index of nonnegative integers, write $|k|=\sum_{j=1}^{d} k_j$, and define the derivative operator
\[
D^k :=\frac{\partial ^{|k|} }{\partial x_1^{k_1} \cdots \partial x_d^{k_d}}.
\]
Then $p_0$ has mixed partial derivatives $D^k p_0$ of order up to $|k| \leq \underline{\beta} := \lceil \beta - 1 \rceil$, satisfying for a function $L: \mathbb{R}^d \to [0, \infty)$,
\begin{equation*}
|D^k p_0(x_1 + x_2) - D^k p_0(x_1)| \leq L(x_1) e^{\varkappa \|x_2\|^2} \|x_2\|^{\beta - \underline{\beta}}, \quad |k| = \underline{\beta}, \; x_1, x_2 \in \mathbb{R}^d,
\end{equation*}
\begin{equation*}
P_0 \left[ \left(\frac{L}{p_0}\right)^2 + \left(\frac{|D^k p_0|}{p_0}\right)^{2\beta/|k|} \right] < \infty, \quad 1 \leq |k| \leq \underline{\beta}.
\end{equation*}
Here $\varkappa > 0$ is a fixed constant. Furthermore, $p_0(x) \leq c e^{-b\|x\|^{\tau}}$, for every $\|x\| > R_0$, for positive constants $R_0, b, c, \tau$.
\end{asm}

\noindent

Each $p_{F,\Sigma}\in\densityspacemodel_n$ is indexed by
$\theta=(w,z,\Sigma)\in\Theta_n
:=\Delta_{\N}\times[-a,a]^{d \N}\times\mathcal D_{\sigma}$. We
metrize it by
\begin{equation}\label{eq:paramnorm}
\varrho(\theta,\theta')
   :=\lVert w-w' \rVert_{1}
     +\max_{1\le j\le \N}\lVert z_{j}-z_{j}' \rVert_{\infty}
     +\lVert \Sigma-\Sigma' \rVert_{\mathrm{op}},
\end{equation}
where $\max_{1\le j\le \N}\lVert z_{j}-z_{j}'\rVert_{\infty} = \max_{j,k}|z_{j,k} - z^{'}_{j,k}|$. 

This upper bound on the bracketing number then leads to the following density estimation result:

\begin{restatable}[Density Result]{thm}{densityresult}\label{Thm:DensityResult}
Assume that $d \geq 2$ and that Assumptions \ref{asm_true_MDM} -- \ref{asm_positive_MDM} and Assumption \ref{asm_Holder}, with parameters $\beta, \tau$, hold. Then there exists $p_n$ such that \eqref{eq:approx} and \eqref{eq:tail} hold true and any approximate maximizer $\hat p_n\in\Pnsieve$ of the observed-data
likelihood $p\mapsto\prod_{i=1}^n p^{(M_i)}(X_i^{(M_i)})$, in the sense that
\begin{equation}
  \frac{1}{n} \sum_{i=1}^{n} \log(\hat p_n^{(M_i)}(X_i^{(M_i)}))
  \;\ge\;
  \Pn\log \tilde{p}_n - O_P(n^{-2\beta/(2\beta + d)} (\log(n))^{2t}),
\end{equation}
satisfies
\begin{align}
\He(p_0, \hat{p}_n)=O_P\left(\frac{\CIV}{\sqrt{c_0}} n^{-\beta/(2\beta + d)} (\log(n))^t\right). 
\end{align}
\end{restatable}

\bigskip

\begin{remark}[Minimax optimality and the role of the sieve under MAR]\label{rem:minimax}
\
\begin{enumerate}[label=\textup{(\roman*)},leftmargin=2.4em]

\item \textit{Minimax optimality.}
The rate $n^{-\beta/(2\beta+d)}$ is the minimax optimal rate for estimating a $\beta$-smooth density on $\R^d$ in Hellinger distance \citep{Tsybakov2009}. Since complete data is a special case of MAR (take $\PMzerox=1$, i.e.\ $c_0=1$), the complete-data minimax lower bound applies unchanged to the MAR setting. Our upper bound therefore matches the minimax rate up to the logarithmic factor $(\log n)^t$ with additional constant $\CIV/\sqrt{c_0}$.

\item \textit{The logarithmic factor is not an artifact of missingness.}
The factor $(\log n)^t$ in Theorem~\ref{Thm:DensityResult} originates from the
metric entropy of the Gaussian mixture sieve, specifically the
$\N\log(a/\sigma)$ term in the bracketing number bound of
Lemma~\ref{lem:cov2br}, which does not involve the missingness mechanism.
The same $(\log n)^t$ factor appears in the complete-data Bayesian posterior contraction rate for Dirichlet mixtures of normals \citep[Theorem~9.9]{Fundamentals}. 
More generally, sieve MLEs are not known to attain the rate $\varepsilon_n$,
measured in Hellinger distance and defined implicitly by the entropy fixed point
$\log N(\varepsilon_n,\mathcal{F},\He)\;\asymp\;n\varepsilon_n^{2},$
at which the metric entropy of the target class balances the information in $n$
observations. For complete data, \citet[Theorem~6]{wong1995probability} attain
$\varepsilon_n(\log(1/\varepsilon_n))^{1/2}$ instead.


\item \textit{Supersmooth densities.}
If instead of Assumption~\ref{asm_Holder} the true density is itself a normal
location mixture, $p_0(x)=\int\phi(x-z;\Sigma_0)\,dF_0(z)$, with $1-F_0([-z,z]^d)\lesssim e^{-b_0 z^{r_0}}$ for constants $b_0>0$ and $r_0\ge2$, then the same sieve and the same application of Theorem~\ref{thm:8mar} yield the nearly parametric rate
\[
  \He(p_0,\hat p_n)
  \;=\;O_P\!\Bigl(\tfrac{1}{\sqrt{c_0}}\,n^{-1/2}\,(\log n)^{(d+1+1/r_0)/2}\Bigr),
\]
following the same proof structure: moment-matching discretization of $F_0$
\citep[Lemma~9.12]{Fundamentals} replaces the $T_{\beta,\sigma}$ bias-correction
construction, and the covariance $\Sigma$ stays near $\Sigma_0$ rather than
collapsing to zero. As in the $\beta$-smooth case, the only cost of MAR
missingness is the factor $1/\sqrt{c_0}$.

\item \textit{Comparison with kernel density estimation.}
On complete data, our Gaussian mixture sieve achieves almost the minimax rate for all smoothness level $\beta > 0$, whereas the KDE with a Gaussian kernel is limited to $\beta \leq 2$. Thus our frequentist estimator offers a similar rate advantage than the Bayesian Dirichlet mixture of normal prior \citep[Section~9.4]{Fundamentals}. Under MAR the limitation of KDE is more severe.
The multivariate KDE requires evaluating $K_h(x - X_i)$ at the full vector $X_i$, which is unavailable when only $X_i^{(M_i)}$ is observed. Using KDE in this context requires learning conditional densities from other patterns, such as $\{X_i : M_i = 0\}$, necessitating complex identification restrictions outside of MCAR. More fundamentally, KDE, as a plug-in estimator, falls outside the likelihood maximizer framework and offers no way to exploit the ignorability property under MAR.
\end{enumerate}
\end{remark}

In classical maximum likelihood estimation, the ignoring estimator is known to increase the complexity of estimation \citep{EM0}. For instance, even if the data is Gaussian, the ignoring estimator does not allow for closed-form expressions of the mean and covariance matrix, and instead EM algorithms are used that impute the data in each step. To avoid this additional numerical complexity, we constrain the hypothesis class to $\widetilde{\Pnsieve} = \{p_{F,\Sigma} \in \Pnsieve : \Sigma \text{ is diagonal}\}$ and try to find an approximate maximizer $\widetilde{p}_n \in \widetilde{\Pnsieve}$ of the observed-data likelihood $p \mapsto \prod_{i=1}^n p^{(M_i)} (X_i^{(M_i)})$. $\widetilde{p_n}$ cannot be better than the approximate likelihood maximizer $\hat p_n \in \Pnsieve$ since $\widetilde{\Pnsieve} \subset \Pnsieve$. However, this simplification allows for an EM algorithm that targets the ignoring likelihood directly with closed-form solutions at each step, making the approach very efficient. Moreover, the non-diagonal $\Sigma$ is not essential for the proof of Theorem \ref{Thm:DensityResult}, and the same asymptotic result could be shown for $\Pnsieve$ replaced by $\widetilde{\Pnsieve}$.

Numerically, we take two steps to find $\widetilde{p_n}$. First, we consider a fixed number of components $N$, and find a Gaussian Mixture that approximately maximizes the observed data likelihood, using the MAR-EM algorithm in Algorithm \eqref{algo:base}. Optimization is performed for multiple random initializations, reducing the likelihood of converging to a local optimum. We add a small regularization to the covariances during M-step to ensure positivity. Second, we run the first step over a range of $N$s and pick the model that minimizes the BIC criterion as in Algorithm~\eqref{algo:modelselection}. The corresponding Gaussian Mixture density is $\widetilde{p_n}$ and we refer to this approach as ignoring GMM or simply GMM. 

\subsection{Consistent Nonparametric Imputation under MAR}\label{sec_imputation}

We might not just be interested in recovering $P_0$, or generate new samples from it, but instead in imputing the missing entries in the original sample points. Since $X^{(M)}$ is observed (and thus its distribution can be directly estimated), successful imputation amounts to recovering $\Pjoint$. Crucially, MAR in Assumption \ref{asm_true_MDM} is equivalent to the following condition: for all $m$ and a.e. $x$,
\begin{equation}\label{eq_imputationcondition}
    \mathbb{p}_0(x^{(\bar{m})} \mid x^{(m)}, M=m)=p_0(x^{(\bar{m})} \mid x^{(m)}).
\end{equation}
With the estimator of Section \ref{sec_densityestimation}, it is simple to obtain a consistent estimate of $p_0(x^{(\bar{m})} \mid x^{(m)})$. We will now show that for all $m$,
\[
\hat{\mathbb{p}}_n(x , m)= \hat{p}_n(x^{(\bar{m})} \mid x^{(m)}) p_0(x^{(m)}, m),
\]
is a consistent estimator of $\mathbb{p}_0(x,m)$. That is, if we consider the imputed sample $\tilde{X}=(X^{(\bar{M})}, X^{(M)}, M)$, with $X^{(\bar{M})}$ a draw from the estimated conditional distribution $X^{(\bar{M})} \mid X^{(M)}$. Then the density of this $\tilde{X}$ will converge to $\mathbb{p}_0$ for $n \to \infty$, in the sense that $\He (\mathbb{p}_0, \hat{\mathbb{p}}_n) \to 0$, with
$$\He (\mathbb{p}_0, \hat{\mathbb{p}}_n) := \left(\sum_m \int \big(\sqrt{\mathbb{p}_0(x, m)} - \sqrt{\hat{\mathbb{p}}_n(x, m)}\big)^2 \dir x\right)^{1/2}.$$


\begin{restatable}[Imputation Result]{cor}{imputationresult}\label{Cor:Imputationresult}
Assume that $\He(p_0, \hat{p}_n)=O_P(r_n)$, for a sequence $r_n \to 0$. Then it holds that,
\begin{align}
\He(\mathbb{p}_0, \hat{\mathbb{p}}_n)=O_P\left(r_n\right). 
\end{align}
\end{restatable}

Theorem \ref{Thm:DensityResult} and Corollary \ref{Cor:Imputationresult} together provide a convergence result for nonparametric imputation under non-monotone MAR.

We implement the imputation method by first estimating the ignoring density estimator as described in Section \ref{sec_densityestimation} and then sampling from the corresponding conditional distributions, as detailed in Algorithm \ref{algo:impute}.

\section{Estimation Beyond MAR}\label{sec_MNAR}

In this section, we examine how our results change under non-monotone MNAR and provide a bound on the asymptotic error incurred. While this bound will be vacuous in general, we combine this with an extension of the (realisable) contamination framework of \cite{MNARcontamination} to show that the nonparametric MLE is indeed robust against MNAR deviation, in the sense that the error that increases smoothly with the extent of the contamination. Crucially, we also recover a weaker condition than MAR under which convergence to $P_0$ still holds, that we call $\emph{MAR-like MNAR}$:
\begin{dfn}[MAR-like MNAR]\label{def_MARlike}
        The true conditional distribution of $M$ given $X$ is called MAR-like MNAR, if \begin{align}\label{eq_MARlike}
    W(x) = \sum_{m} \PMmxm =1, \quad a.e.
\end{align}
\end{dfn}

We will show that the asymptotic error under general MNAR can be bounded by $\Var_{P_0}(W(X))=\E_{P_0}[(W(X)-1)^2]$, the expected squared deviation of $W(X)$ from one, or simply the variance of $W(X)$. In the following, let
\begin{align}\label{eq:mathbbp}
  \mathbb{p}_0\bigl(x^{(m)},m\bigr) \;&:=\; \int p_0(x) \mathbb{P}(M=m \mid X = x) dx^{(\bar{m})}
\end{align} 
denotes the observed data density under $\mathbb{P}_0$, and define the set
\begin{equation}
  \Solutionspace =\arg\max_{p \in \densityspacemodel} L(p),
\label{eq:minimizer_set}
\end{equation}
with
\[
  L(p):=\E_{\Pjoint}\bigl[\log p^{(M)}(X^{(M)})\bigr]
       =\sum_m\int \log p^{(m)}\bigl(x^{(m)}\bigr)\,\mathbb{p}_0\bigl(x^{(m)},m\bigr)\,dx^{(m)}.
\]

\subsection{Existence and Uniqueness of the Maximizer}

We begin by characterizing the full solution set $\Solutionspace$ without any assumption on the missing mechanism.
Recall that $\densityspace$ the space of any density with respect to Lebesgue measure, so that $\densityspacemodel\subseteq\densityspace$.
The central object is the \emph{score functional}
\begin{align}\label{eq:Phi}
  \score{p}(x)\;&:=\;\sum_{m\in \mathcal{M}}\frac{\mathbb{p}_0\bigl(x^{(m)},m\bigr)}{p^{(m)}\bigl(x^{(m)}\bigr)},
  \qquad \text{for any  } p\in\densityspace, 
\end{align} 
with $p^{(m)}(x^{(m)}) = \int p(x) dx^{(\bar{m})}$, so $\score{p}$ takes values in $[0,\infty]$.


\begin{restatable}{lemma}{propofPhi}\label{lem:propofPhi}
The following properties follow directly from \eqref{eq:Phi}.
\begin{enumerate}[label=\textup{(\roman*)}, itemsep=0pt, leftmargin=2.4em]
  \item $\score{p_0}(x) = \sum_{m} \PMmxm$ for a.e.\ $x$.
   \item  For every $p\in\densityspace$ such that $p(x)>0$ a.e., $\E_P[\score{p}] = 1$.
\end{enumerate}
\end{restatable}

The score functional $\Psi_p$ allows us to characterize the solution space $\Solutionspace$, and to show that if it is nonempty, the solution must be unique.

\begin{restatable}[Population target of the ignoring MLE under arbitrary missingness]{thm}{solutioncharacterization}\label{thm:new}
Assume Assumptions~\ref{asm_no_empty_MDM} and~\ref{asm_positive_MDM}, with the model class $\Modelspace/\densityspacemodel$ convex and $P_0\in\Modelspace$. Then:
\begin{enumerate}[label=\textup{(\roman*)}, itemsep=0pt, leftmargin=2.4em]
\item \emph{(Concavity)} $L$ is concave on $\densityspacemodel$.
\item \emph{(Characterization)} Let $P\in\Modelspace$ have density $p \in \densityspacemodel$.
      \begin{enumerate}[label=\textup{(ii-\alph*)}, ref=\textup{(ii-\alph*)}, itemsep=0pt, leftmargin=3.0em]
      \item\label{ii-a} \emph{Sufficiency.} If
      \begin{equation}\label{eq:stat}
        \score{p}(x) = 1 \quad \text{a.e.},
      \end{equation}
      then $p\in\Solutionspace$.
      \item\label{ii-b} \emph{Necessity.} If $p\in\Solutionspace$ then $\int q(x)\,\score{p}(x)dx\le1$ for every $q\in\densityspacemodel$. If moreover $\densityspacemodel=\densityspace$, then \eqref{eq:stat} holds.
      \end{enumerate}
    \item \emph{(Uniqueness)} $\Solutionspace$ contains at most one element, denoted $p_\star$ when it exists. We denote by $P_{\star}$ the distribution associated to $p_\star$.
    \item \emph{(Truth)} Let $W(x)=\sum_m\PMmxm$, so that $\score{p_0}=W$. If $\densityspacemodel=\densityspace$, then $\Solutionspace=\{p_0\}$ if and only if $W(x) =1$ a.e.
    \item \emph{(Curvature)} 
       Let $\kappa:=\operatorname*{ess\,inf}\,p_0(x)/p_\star(x)$. If $\kappa>0$, then for all $p\in\mathscr{P}$,
      \[
        L(p_\star)-L(p)\;\ge\;\kappa c_0\, \He^2(p,p_\star) 
      \]
    \item \emph{(Bias)}  If $\kappa>0$, then
\begin{equation}\label{eq:varbound}
  \He^2(p_0,p_\star)\ \le\ \frac{(1+\kappa^{-1/2})^2}{(c_0\kappa)^2}\,\Var_{P_0}\bigl(W(X)\bigr),
\end{equation}
with $\Var_{P_0}\bigl(W(X)\bigr)=\E_{P_0}[(W(X)-1)^2]$. 
\end{enumerate}
\end{restatable}

\begin{remark}
    The bound derived in (vi) will be of particular relevance. For $\mathcal{M}=\{m: \Prob(M=m) > 0\}$, we note that since $c_0 \leq W(x) \leq |\mathcal{M}|$, Popoviciu's inequality \citep{bhatia2000variance}, gives the crude bound,
    \begin{align*}
        \Var_{P_0}(W(X))\leq \frac{(|\mathcal{M}| - c_0)^2}{4}.
    \end{align*}
\end{remark}

Theorem~\ref{thm:new} reduces the population behavior of the ignoring MLE to the score functional $p\mapsto\score{p}$.
When $\mathscr{P}=\densityspace$, parts~(ii) and~(iii) give $\Solutionspace=\{p\in\densityspace:\ \score{p}=1\text{ a.e.}\}$, a set with at most one element up to null sets. Part~(iv) shows that this element is $p_0$ exactly when $W(x) =1$ a.e. Among all features of the missingness mechanism, function $W(x)$ determines whether the ignoring MLE identifies $p_0$. When $W\neq1$ on a set of positive measure, parts~(v) and~(vi) quantify the resulting cost.

Interestingly, the condition $\score{p}=1$ also has a direct probabilistic meaning. a density $p$ satisfies it exactly when there exists some MAR mechanism that, when combined with $p$, reproduces the observed-data law $\mathbb{p}_0\bigl(x^{(m)},m\bigr)$.
The following notion goes back to \cite{gill1997coarsening,ourresult}.


\begin{dfn}[MAR Representation]\label{def:marrep}
    A density $p \in \densityspace$ \emph{admits a MAR representation of} $\mathbb{p}_0\bigl(x^{(m)}, m\bigr)$ if there exists a collection $\nu = (\nu_m)_{m \in \mathcal{M}}$ of measurable functions
    $\nu_m : \mathcal{X}^{(m)} \to [0,1]$ with $\sum_{m} \nu_m\bigl(x^{(m)}\bigr)$ $ = 1$ for
    a.e.\ $x$, such that
    \[
      \nu_m\bigl(x^{(m)}\bigr)\, p^{(m)}\bigl(x^{(m)}\bigr) = \mathbb{p}_0\bigl(x^{(m)}, m\bigr)
    \]
    for every $m$ and a.e.\ $x^{(m)}$. In such cases, we call the pair $(p, \nu)$ a MAR representation.
\end{dfn}

\begin{restatable}{prop}{MARRepresentationCharacterization}\label{prop:MARRepresentationCharacterization}
Assume Assumption~\ref{asm_positive_MDM}. A density $p\in\densityspace$ admits a MAR
representation if and only if \eqref{eq:stat} holds. In that case $\nu$ is unique up to null subsets of $\mathcal X$, and
\begin{equation}\label{eq:nuforced}
  \nu_m\bigl(x^{(m)}\bigr)=
  \begin{cases}
    \mathbb p_0\bigl(x^{(m)},m\bigr)\big/p^{(m)}\bigl(x^{(m)}\bigr), & p^{(m)}\bigl(x^{(m)}\bigr)>0,\\[2pt]
    0, & \text{otherwise}.
  \end{cases}
\end{equation}
\end{restatable}




\begin{remark}[Connection to prior literature]
The connection of the MAR representation with the KL divergence minimizer (maximizer of $L(p)$) was made much earlier in the context of random variables with finite support and CAR in \cite{gill1997coarsening}, see also \cite[Theorem 2.1]{robins1997nonresponse}. However, aside from the extension to the continuous case and the new Condition \eqref{eq:stat}, the crucial new part of Theorem \ref{thm:new} is the bound in (vi), allowing to derive natural error bounds of the nonparametric MLE under general non-monotone MNAR. We refer to Section \ref{sec_relatedliterature} for a more detailed discussion on related literature.
\end{remark}

A first consequence of Proposition~\ref{prop:MARRepresentationCharacterization} is that the ignoring MLE targets $p_0$ exactly when $(p_0,\nu_{p_0})$ is itself a MAR representation, that is, when \eqref{eq_MARlike} holds. This brings us to the a strengthening of Theorem~\ref{thm:8mar} and Theorem~\ref{Thm:DensityResult} beyond MAR.

\subsection{Beyond MAR: MAR-like MNAR}

We first note that MAR (Assumption \ref{asm_true_MDM}) clearly implies $W(x) = 1$ a.e. On the other hand, the (MAR-like MNAR) example of Section \ref{Sec_Motivating} demonstrates that $W(x) = 1$ is strictly weaker than MAR. We repeat it here for completeness.

\begin{exm}\label{exm:unif}
Let
\[
X_1, \, X_2 \sim \Gauss{0}{1} \qquad \text{and} \qquad X_1 \perp X_2.
\]
with the following missingness mechanism:
\begin{align*}
    &\Prob(M=(0,0)\mid x) = 1/2 \nonumber \\
    &\Prob(M=(1,0)\mid x) = 1/4 + 1/2(\Phi(x_1) - 1/2)(\Phi(x_2) - 1/2) \nonumber \\ 
    &\Prob(M=(0,1)\mid x) = 1/4 - 1/2(\Phi(x_1) - 1/2)(\Phi(x_2) - 1/2). 
\end{align*}
Clearly this is MNAR (not MAR), but it can be checked that $\sum_m \PMmxm = 1$.
\end{exm}

The characterization of $\Solutionspace$ in the previous subsection allows us to extend our main results beyond MAR to missingness mechanisms satisfying \eqref{eq_MARlike}.

\begin{restatable}{cor}{Marlikeresult}\label{cor:MARlike}
Theorems~\ref{thm:8mar} and~\ref{Thm:DensityResult} remain valid, with the same target $p_0$ and the same constants, when the MAR Assumption~\ref{asm_true_MDM} is replaced by \eqref{eq_MARlike}, that is,
\[
\He\bigl(p_0,\hat p_n\bigr)
  \;=\; O_P\;\!\Bigl(\tfrac{1}{\sqrt{c_0}}\bigl(\delta_n\vee n^{-1/2}\bigr)\Bigr).
\]
Similarly for the application to Holder class density estimation (Theorem~\ref{Thm:DensityResult}).
\end{restatable}


This allows us to quantify $\He(p_\star, \hat{p}_n)$ using the same technique we used to quantify $\He(p_0, \hat{p}_n)$. Importantly, Condition~\eqref{eq_MARlike} cannot be further weakened without additional structure to identify $p_0$. In contrast, MAR is necessary to our imputation result in Section \ref{sec_densityestimation}, as otherwise \eqref{eq_imputationcondition} might not hold, making it impossible to recover $\mathbb{p}_{0}$. This formalizes the intuition that estimating $P_0$ by itself is easier than estimating $\Pjoint$.

\subsection{Robustness Beyond MAR-like MNAR}

When \eqref{eq_MARlike} ($W(x) =1$) fails, Theorem~\ref{thm:new}\,(iv) shows that the ignoring MLE no longer targets $p_0$. If $\Solutionspace$ is nonempty, it targets instead the density $p_\star$, which still admit a set of distributions $\mathbb{P}(M=m \mid X)$ that reproduce the observed distribution $\mathbb{p}(x^{(m)},m)$. The total error in Hellinger distance now splits into two parts
\[
  \He(p_0,\hat p_n)\;\le\;{\He(p_\star,\hat p_n)}
  \;+\;{\He(p_\star,p_0)},
\]
where the estimation term is controlled by the technique of Corollary~\ref{cor:MARlike}, applied to the MAR representation $(p_\star,\nu_{p_\star})$, and the identification bias term by how far $W$ is from $1$. Specifically, by Theorem~\ref{thm:new}\,(vi), $\He^2(p_0,p_\star)$ can be bounded by the variance of $W(X)$ over $X \sim P_0$, which is nothing else than the expected squared deviation from $\E_{P_0}[W(X)]=1$. 

Nonetheless, since $\He^2\le2$, the bounds above are informative only when $W$ is close to $1$ relative to $c_0\kappa$. 
A practically relevant question is therefore how much a small, unverifiable departure from MAR can cost. We hence adopt a Huber-type $\varepsilon$-contamination model introduced in \cite{MNARcontamination}. In particular, if we assume further 
\begin{align}\label{eq_Contamination}
    \Prob(M=m \mid X)=(1-\varepsilon) \PMAR(M=m \mid X) + \varepsilon \PMNAR(M=m \mid X) \ \ a.e.
\end{align}
where $\PMAR$ is a MAR mechanism and $\PMNAR$ an arbitrary one, define 
\[
  W^{\scriptscriptstyle\mathrm{MNAR}}(x):=\sum_m\int\PMNAR(M=m\mid x)\,
  p_0\bigl(x^{(\bar m)}\mid x^{(m)}\bigr)\,dx^{(\bar m)} ,
\]
the analogue of $W$ for the contaminating mechanism, we immediately obtain that $W - 1 = \epsilon(W^{\scriptscriptstyle\mathrm{MNAR}} - 1)$ and a bound that is linear in $\epsilon$. 

For a density $p$, write $\nu_{p,m}$ for \eqref{eq:nuforced} and $\HeMAR_p$ for the adapted Hellinger functional of Definition~\ref{def:Ht} with weights
$\nu_{p,m}$. 
Unlike in the previous subsection, where $\nu_{p_0,m}(x^{(m)})=\Prob(M=m\mid x^{(m)})$ and hence $\HeMAR_{p_0}=\HeMAR$, here $$\nu_{p_\star,m}(x^{(m)})=\Prob(M=m\mid x^{(m)})\, \frac{p_0^{(m)}(x^{(m)})}{p_\star^{(m)}(x^{(m)})}\neq\Prob(M=m\mid x^{(m)})$$
in general, so the extra subscript is needed for distinction. Finally, let $\kappa$ be defined as in Theorem \ref{thm:new},
 \begin{equation}\label{eq:kappa}
   \kappa=\operatorname*{ess\,inf}_{x\in\mathcal X}\frac{p_0(x)}{p_\star(x)},
 \end{equation}
The result is summarized as below:



\begin{restatable}[The ignoring MLE beyond $W=1$]{cor}{beyondW}\label{cor:beyondW1}
Assume Assumption~\ref{asm_no_empty_MDM}, ~\ref{asm_positive_MDM}, and that $p_\star$ satisfies $\Psi_{p_\star}(x) = 1$ for almost all $x$. Moreover, assume $\kappa > 0$ and suppose that the Conditions \eqref{eq:approx}--\eqref{eq:approxmax} of Theorem~\ref{thm:8mar} hold with $p_0$ replaced by $p_\star$, $\HeMAR$ by $\HeMAR_{p_\star}$, and $c_0$ by $c_0\kappa$. Then
\begin{enumerate}[label=\textup{(\roman*)},leftmargin=2em,itemsep=3pt]
\item \emph{(General mechanism)}
\[
  \He\bigl(p_0,\hat p_n\bigr)
  \;=\;O_P\Bigl(\tfrac{1}{\sqrt{c_0\kappa}}\bigl(\delta_n\vee n^{-1/2}\bigr)\Bigr)
  \;+\frac{(1+\kappa^{-1/2})}{(c_0\kappa)}\;\sqrt{\Var_{P_0}\bigl(W(X)\bigr)  }.
\]

\item \emph{(Contamination)} If moreover \eqref{eq_Contamination} holds, then $W-1=\varepsilon\,(W^{\scriptscriptstyle\mathrm{MNAR}}-1)$ a.e., so that
\[
\He\bigl(p_0,\hat p_n\bigr)
  \;=\;O_P\Bigl(\tfrac{1}{\sqrt{c_0\kappa}}\bigl(\delta_n\vee n^{-1/2}\bigr)\Bigr)
  \;+\;\frac{(1+\kappa^{-1/2})}{(c_0\kappa)^2}\varepsilon\sqrt{\Var_{P_0}\bigl(W^{\scriptscriptstyle\mathrm{MNAR}}(X)\bigr)  }.
\]
\item \emph{(Positivity of the MAR component only)} If in (ii)
Assumption~\ref{asm_positive_MDM} is replaced by
$\PMAR(M=\mathbf 0\mid X=x)\ge c_0$ for a.e.\ $x$, the same holds with $c_0$
replaced by $(1-\varepsilon)c_0$.
\end{enumerate}
Similarly for H\"older class density estimation, if Assumption~\ref{asm_Holder}
holds for $p_\star$, the first term may be replaced by the rate of
Theorem~\ref{Thm:DensityResult}, with $c_0$ replaced by $c_0\kappa$.
\end{restatable}

\begin{remark}[What departing from $W=1$ cost]\label{rem:beyondW1}
\
\begin{enumerate}[label=\textup{(\roman*)},leftmargin=2.4em]
\item \textit{Reduction under $W=1$ or MAR.}
Corollary~\ref{cor:beyondW1} covers every mechanism under which some $p_\star$ solves $\score{p_\star}(x)=1$ (that is, whenever a solution exist). This $p_\star$ is then the population maximizer of $\mathbb E_{\Prob_0}[\log p^{(M)}(X^{(M)})]$. Targeting $p_\star$ instead of $p_0$ cost introduces parameter $\kappa := \min_m\operatorname*{ess\,inf}\,p_0^{(m)}/p_\star^{(m)} =\operatorname{ess\,inf}_x p_0(x)/p_\star(x)\in(0,1]$, which measures how much more mass $p_\star$ can put on a region than $p_0$ does. 
If we have the stronger assumption $W=1$ (or MAR), then $p_\star=p_0$, $\kappa=1$ and the bias term vanishes, so this corollary reduces to Corollary~\ref{cor:MARlike} exactly.

\item \textit{Generalization of \cite{MMDmissingnesspaper}} In \cite{MMDmissingnesspaper}, deviations from MCAR were studied in a parametric setting. In particular, they studied bounds on the error incurred by parametric MLE for $d=1$, where MCAR and MAR are equivalent. In this case the departure from MCAR can be simply measured by $\Var_{P_0}(\Prob(M=m \mid X))$. Under MAR however, this variance will not be zero and this approach cannot trivially be extended. Our results show that $\Var_{P_0}(W(X))=\E_{P_0}[(W(X)-1)^2]$ presents a natural measure of deviation from MAR that can be used bound the error even in nonparametric settings instead. As such, Corollary \ref{cor:beyondW1} might be seen as the nonparametric, MAR generalization of the univariate result in \cite[Theorem 4.1]{MMDmissingnesspaper}.

\item \textit{$W$ measures the non-cancelling part of the departure from MAR.}
With $\Pi_m$ the $L^2(P_0)$ projection onto functions of $x^{(m)}$, we have $\nu_{p_0,m}=\Pi_m\pi_m$ for $\pi_m(x):=\Prob(M=m\mid X=x)$, so $W-1=-\sum_mR_m$, where $R_m:=(I-\Pi_m)\pi_m$ is the part of $\pi_m$ that varies with $x^{(\bar m)}$ given $x^{(m)}$. 
MAR requires every $R_m$ to vanish; $W=1$ only requires them to
cancel, as in Example~\ref{exm:unif}. Under contamination, the identification bias is bounded by a term of order $\varepsilon^2 \Var_{P_0}\bigl(W^{\scriptscriptstyle\mathrm{MNAR}}(X)\bigr) $, which is small if either the contamination fraction $\varepsilon$ or the departure $(W^{\scriptscriptstyle\mathrm{MNAR}}(X)-1)^2$ of the contaminating mechanism is small on average.

\end{enumerate}
\end{remark}

To further illustrate the result in this section:
\begin{exm}\label{exm:gaussianW}
    Let $(X_1, X_2)$ follow a bivariate Gaussian distribution with mean zero and correlation $0.7$, with $m_1=(0,0)$ and $m_2=(1,0)$. Consider the mechanism
    \begin{align*}
        \Prob(M=m_1\mid x) = G(x_1), \ \ \Prob(M=m_2\mid x) = 1-G(x_1),
    \end{align*}
    with $G$ defined as in Section \ref{Sec_Motivating}. This is a self-censoring mechanism, which is quite difficult to handle. It can be shown that in this example
    \begin{align*}
        p_{\star}(x)=p_{0}(x_2) p_0(x_1 \mid x_2)\frac{\Prob(M=m_1 \mid x)}{\Prob(M=m_1 \mid x_2)}=\frac{p_{0}(x)G(x_1)}{\tilde{G}(x_2)},
    \end{align*}
    where $\tilde{G}(x_2)=\int G(x_1) \phi(x_1; 0.7 x_2, 1-0.7^2)d x_1=\E[G(X_1) \mid x_2]$. This is rather intuitive and follows the reasoning of \cite{ourresult}: Assuming MAR is equivalent to wrongly assuming that the conditional distribution of $X_1 \mid X_2$ is the same in both patterns. Thus $p_{\star}$ is the correct marginal of $x_2$ mixed with the conditional distribution of $X_1 \mid X_2$ in pattern $m_1=(0,0)$. In this case we have, under Assumption \ref{asm_positive_MDM},
    \[
    \kappa= \operatorname*{ess\,inf}_{x_1,x_2} \frac{\tilde{G}(x_2)}{G(x_1)} \geq c_0, \text{ and } W(X)= G(X_1) + 1-\E[G(X_1) \mid X_2],
    \]
    and using properties of the conditional expectation, and Popoviciu's inequality (\cite{bhatia2000variance}),
    \begin{align*}
        \Var_{P_0}(W(X))
        &=\E_{P_0}[\Var(G(X_1)\mid X_2 )]\leq (1-c_0)^2/4.
    \end{align*}
    giving a rather uninformative bound of $(1+c_0^{-1/2})(1-c_0)/(2c_0^2)$. However, if assume that a MAR mechanism is instead contaminated by the mechanism above, i.e. for $\varepsilon \in (0,1)$,
    \begin{align*}
        \Prob(M=m_1\mid x) = (1-\varepsilon)G(x_2) + \varepsilon G(x_1), \ \ \Prob(M=m_2\mid x) = 1-((1-\varepsilon)G(x_2) + \varepsilon G(x_1)),
    \end{align*}
    the bound becomes $\varepsilon (1+c_0^{-1/2})(1-c_0)/(2c_0^2)$, which scales linearly with $\varepsilon$.
\end{exm}

\section{Further Related Literature}\label{sec_relatedliterature}

The literature on missing values is vast and rapidly growing, so giving an appropriate account of the related literature can be difficult. We aim to achieve this in this section.

The term MAR has been a frequent cause of confusion in the literature. As such, it is often claimed that \textit{``Procedure XX is valid under MAR''}. However, to the best of our knowledge, frequentist validity in the sense of consistency has not been formally established in general models, and has only relatively recently been addressed for regular parametric complete-data models in \cite{Takai2013}. The difficulty surrounding the MAR condition has led to a split in the literature. On the one hand, an ever increasing number of empirical (imputation) methods are proposed, often claiming to ``work under MAR'' without theoretical guarantees, such as \cite{stekhoven_missoforest, CARTpaper0, GAIN, directcompetitor1, MIWAE,VAE1, misgan, missdiff, MIRI} among others. We refer to \cite{OneBenchmarktorulethemall, näf2024goodimputationmarmissingness} for a discussion of the pitfalls of some of these methods. On the other hand, theoretical results under MAR typically focus on special cases, for instance, by assuming that there is a set of fully observed variables that completely determines the missingness (e.g., \cite{EmpiricalLik1, EmpiricalLik2, EmpiricalLik3, Responsequantileimputation, imputationpoweredinference}). One interesting recent example of this is \cite{EmpiricalLik4}, who use empirical likelihood methods as an alternative to IPWs. Since the empirical likelihood approach is a version of nonparametric MLE, it shares some connection to our work. However, the approach is again only applicable in the restricted MAR settings that we are looking to overcome. Finally, one of the most general sets of sufficiency conditions for IPW and imputation-based estimators are the assumptions encoded in \cite{chen2022pattern}. In particular, this paper generalizes several prior identification restrictions, such as donor-based restrictions of \cite{nonparametricpmm}, into a single framework. In this approach, the so-called ``pattern graphs'' encode the missingness assumptions, which can even be missing not at random (MNAR). ``Regularity'' of this pattern graph is sufficient to construct IPW and even imputation-based (M-) estimators, which in turn depend on the graph. Although this is another viable alternative to the ignoring approach we pursue, it comes at the price of assuming a specific graph, and, as discussed in \cite{chen2022pattern}, MAR in general does not define a regular pattern graph. In particular, the example in Section \ref{Sec_Motivating} cannot be written as such a regular graph.

This difficulty of dealing with MAR may be one of the reasons why the MAR condition has fallen somewhat out of favor and instead robust MCAR version are considered \cite{MNARcontamination, MMDmissingnesspaper}. However, the work in \cite{Rubin_Inferenceandmissing, Takai2013, Bayes} and others has shown that MAR naturally aligns with likelihood maximization, or equivalently, with KL divergence minimization. This is the angle we considered to derive our convergence results. While some papers understood the importance of this connection between likelihood optimization and MAR, such as \cite{MIWAE, MIRI}, they appear to present only limited theoretical considerations. One crucial exception is the paper \cite{FLOWGEM} which builds on \cite{Bayes} and makes use of the result that a (population) minimizer of the Kullback-Leibler (KL) divergence corresponds to the true distribution, even without parametric restrictions. A crucial condition for this, which was not mentioned in the previous papers, is positivity, i.e. Assumption \ref{asm_positive_MDM}. However, they then only provide a result on their gradient approximation, which is an important first step, but leaves crucial theoretical questions open. Moreover, inspecting their assumptions in more detail reveals that even this first step needs a relatively strong smoothness condition on $\Prob(M=m \mid x)$ for all $m$ with $\Prob(M=m) > 0$, namely that the second derivative exists and is uniformly bounded. 

In contrast, we provided a fairly general convergence result with minimal assumptions on $\Prob(M=m \mid x)$ that could give rise to a wide range of algorithms. To illustrate, we derived a new nonparametric ignoring density estimator able to estimate the density under MAR-like MNAR, which is even weaker than general non-monotone MAR, and positivity, without the need for further identification conditions or propensity score estimates. To underscore the significance of this density estimation result, it is instructive to consider the recent paper \cite{nonparametricpmm}. In this paper, kernel density estimators (KDEs) per pattern are used to establish rates of estimation for both the density and CDF of the data. This is close in spirit to the density estimator we obtain as a natural application of our theory. However, in \cite{nonparametricpmm} a density per pattern $m$ is estimated, making it necessary to use complex identification conditions, to be able to estimate the correct conditional densities from the remaining patterns. In contrast, the ignoring sieve MLE approach uses all available data to estimate the density and attains the minimax rate of estimation, up to logarithmic factors. In addition, this new estimator has advantages over the KDE even on complete data, as our sieve construction allows to attain this minimax rate for any smoothness level of the true density. 

We also discussed the idea of robustness against deviations from the MAR missingness mechanism, as studied under MCAR in \cite{MNARcontamination, MMDmissingnesspaper}. In a parametric setting, \cite{Golden2019} showed that, under the assumption of existence of the corresponding quantities, the ignoring MLE may still converge to a (misspecified) limit, though they did not derive any error bounds. \cite{MMDmissingnesspaper} then showed that in one dimension, where MCAR and MAR are the same, the total variation (TV) distance between the MLE distribution and $P_{0}$ can be bounded by a term that signifies the deviation from \emph{MCAR}. Concretely, their Theorem 4.1 shows that 
\[
\He^2(p_{\theta^*}, p_0) \leq 2\frac{\Var_{P_0}(\Prob(M=\mathbf{0} \mid X))}{\Prob(M=\mathbf{0})^2},
\]
where $\theta^*$ is the maximizer of the expected log-likelihood, which was assumed to exist. 
We generalized prior results in three different directions. First, our results are nonparametric, making model misspecification a somewhat less pressing issue than in the parametric setting. Second, connecting to results going back to \cite{gill1997coarsening} in the discrete case, we characterized the best approximation under MAR as exactly the MAR representation of $\Pjoint$. Third, as a consequence of this general theory, we were able to show that consistent estimation of $P_0$ alone (``Generation'') is possible under MAR-like MNAR, while for imputation, MAR is needed in general.

Finally, we clarify the distinction in our question scope to The classical literature on corsening at random (CAR), a generalization of MAR that also covers censored, rounded, and grouped data. This line of literature including \cite{gill1997coarsening, robins1997nonresponse, Cator2004, Jaeger2005, GillGrunwald2008} has primarily focused on characterizing CAR mechanisms and the class of observed-data distributions they can generate. 
The latter is a question of testability. Namely, whether the map sending a full-data density and a MAR mechanism to the observed-data density they generate, is onto, or, in the weaker statistical sense of \cite{Cator2004}, has a dense image. We study the opposite direction. Given the observed law, we ask which full-data law the ignoring MLE identifies, whether it is unique, how far it is from $p_0$, and how fast $\hat p_n$ converges to it.

\section{Simulation Study}\label{sec_empirical}

In this section, we consider densities $p_0$ satisfying Assumption~\ref{asm_Holder}, and conduct simulations to investigate the finite-sample performance of the proposed density estimator. For the M-estimation part, we compare our results to the IPW estimator of \cite{MARinverseweighting}, which we will refer to as ``NMMIPW estimator''. We consider three main settings: (I) a Gaussian distribution, (II) a mixture of Gaussians, and (III) a logistic distribution combined with a Gaussian copula. The missingness in Case (I) is generated as in the (MAR) example in Section \ref{Sec_Motivating}. For Cases (II) and (III), we follow \cite{FLOWGEM} and first randomly generate $n_{pat}=5$ patterns, ensuring that the fully observed pattern $\mathbf{0}$ is included, and then use the ``ampute'' function of \cite{ampute} with these pre-specified patterns. We note that the low number of patterns is mostly due to the computation time of the NMMIPW estimator.\footnote{We used an LLM to transform the \textsf{R} function from the original paper into a faster Python function, but computation time still increases rapidly with the number of patterns.} Table \ref{tab:summary} summarizes these three settings.

We start by considering estimating the density itself and its use in generation, by comparing the generated data from a sample with missing values with a complete sample from $P_0$ in Section \ref{Sec_Simulation_DensityEstimation}.  We also include increasing contamination of the MAR mechanisms with an MNAR mechanism and study how well the complete data density can be recovered over varying levels of contamination. Next, we study nonparametric estimation of the propensity score $\Prob(M=\mathbf{0} \mid x)$ using our estimator $\hat{p}$, and then consider finite-dimensional M-Estimators in Section \ref{Sec_Simulation_MEstimation}. 

\begin{table}[ht]
\caption{Considered Scenarios. In Case (II), the covariance matrix of each of the 20 components is diagonal, but randomly generated. In all cases, we take $n=2000$.}
\label{tab:summary}
\centering
\begin{tabular}{llllr}
\toprule
Case & $P_0$ & Mean & Covariance   & $d$ \\
\midrule
(I) & Gaussian  &  $\mathbf{0}$ & $(\Sigma)_{i,j}=0.7 \cdot\one\{i\ne j\} + \one\{i=j\}$ 
& 2   \\
(II) & Mixture of 20 Gaussians &  random & random    & 20  \\
(III) & Logistic with Gaussian copula & $\mathbf{0}$ & $(\Sigma)_{i,j}=0.7^{|i-j|}$     & 20  \\
\bottomrule
\end{tabular}
\end{table}

Overall, this section shows that our estimator is quite powerful and is able to recover the missing data well, with performance comparable to or better than KDE using complete data when measured in terms of energy distance \citep{EnergyDistance}. However, M-estimation reveals a certain struggle of our estimator to model (linear) dependence. Since we consider mixtures of independent Gaussians this is not surprising, and further research will explore ways to improve this. As such, we note that we do not expect that our method has the best empirical performance compared to the currently most successful imputation methods. Indeed, the number of imputation methods is vast and a handful of methods stand out as consistently strong; see, e.g. the recent comparison of over 70 imputation methods in \cite{OneBenchmarktorulethemall}. While our method would likely perform well for numerical datasets, methods like ``mice$\_$rf'' and ``mice$\_$cart'' might still outperform. Nonetheless, we believe that the experiments here, together with the application in Section \ref{sec_Application} serve to show the versatility of our method. Moreover, Theorem \ref{thm:8mar} may open the way to more powerful methods, as also discussed in Section \ref{sec_conclusion}.

\subsection{Density Estimation}\label{Sec_Simulation_DensityEstimation}

In this section, we study the finite-sample performance of the proposed density estimator under M(N)AR missingness, focusing on cases (II) and (III) for $d=20$ and $n=2'000$. We evaluate the capability of the density estimator to recover the complete-data density using the energy distance \citep{EnergyDistance} against the true data generating mechanism. To check the robustness of the density estimator, we contaminate the MAR with an MNAR model and compare how well the complete data density could be recovered over varying levels of contamination $\varepsilon$. 

Figure~\ref{fig:densityestimation} shows the result together with a KDE used on \emph{complete data} as a comparison: For the Gaussian mixture (left), the KDE surprisingly displays a slightly better performance than our ignoring GMM, although that might be simply the consequence of the missing values. As we increase $\varepsilon$, the error increases without exploding, as one would expect from our discussion in Section~\ref{sec_MNAR}. For the Logistic distribution (right) our estimator is slightly better than the KDE under MAR and the error again increases with $\varepsilon$, though this time there appears to be a lot of variation masking the increase. As such, even with $\varepsilon=0.9$ our estimator performs comparably to the KDE on complete data.


  \begin{figure}[htbp]
    \centering
    \begin{minipage}{0.48\textwidth}
      \centering
      \includegraphics[width=\linewidth]{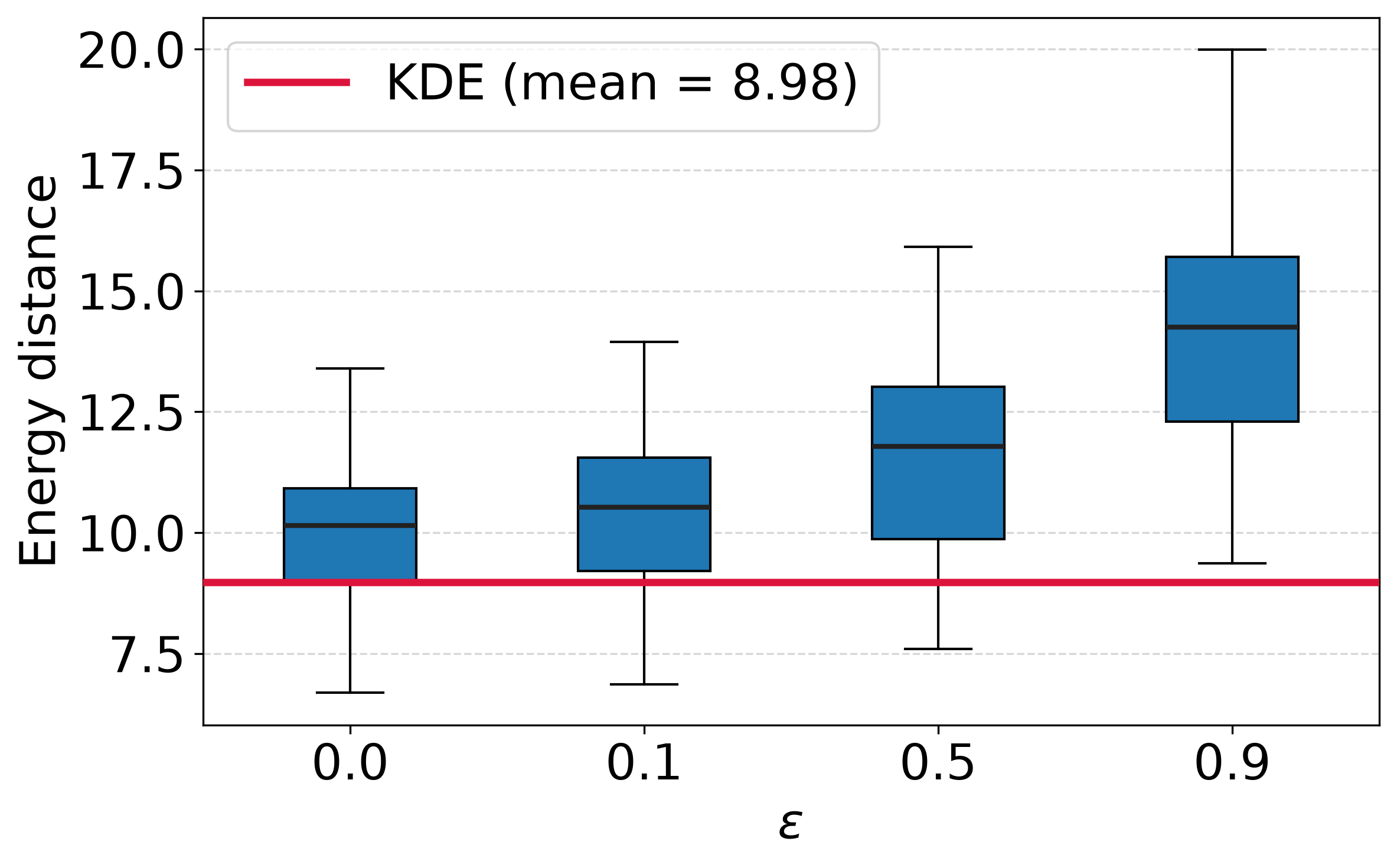}
    \end{minipage}\hfill
    \begin{minipage}{0.48\textwidth}
      \centering
      \includegraphics[width=\linewidth]{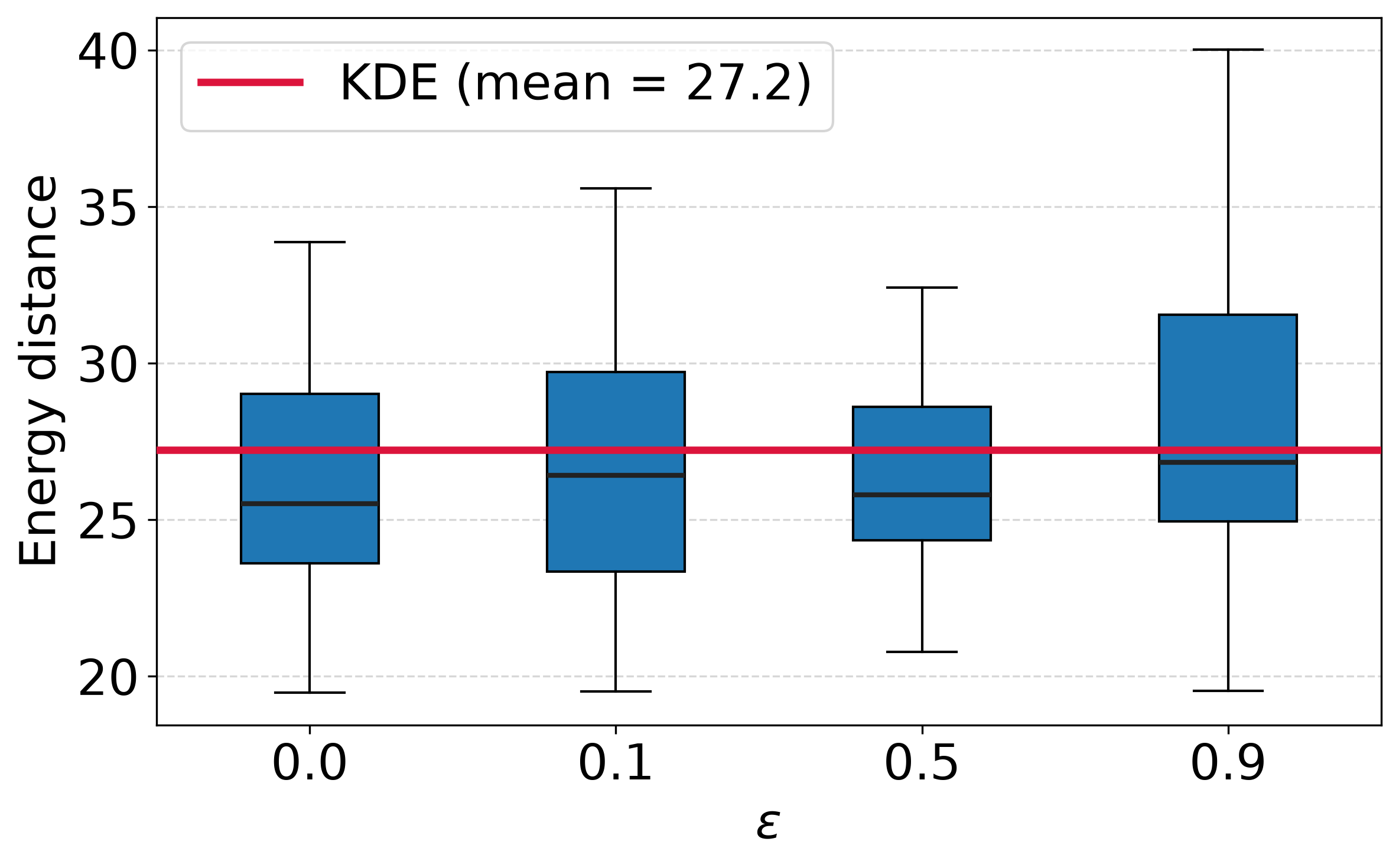}
    \end{minipage}
    \caption{Energy distance for Case (II) (Gaussian Mixture, Left) and Case (III) (Logistic, Right), with $\varepsilon \in \{0, 0.1,0.5,0.9\}$.}
    \label{fig:densityestimation}
  \end{figure}


\subsection{Further Estimation Examples}\label{Sec_Simulation_MEstimation}

\paragraph{Estimating the propensity score $\Prob(M= \mathbf 0 \mid x)$.} We now use our estimator to obtain a nonparametric estimator of the propensity score $\Prob(M= \mathbf 0 \mid x)$, as this can be used for IPW or might be of independent interest. We fit the estimated density using the ignoring GMM to obtain $\hat{p}(x)$ and then use the same algorithm on all observations $i$ with $M_i=\mathbf 0$, i.e. the fully observed pattern, to obtain $\hat{\mathbb{p}}(x \mid M= \mathbf 0)$. Finally, we calculate the estimate
\[
\frac{\hat{\mathbb{p}}(x \mid M= \mathbf 0)}{\hat{p}(x)} \hat{\pi},
\]
where $\hat{\pi}=1/n \sum_{i=1}^{n} \one\{M_i=\mathbf 0\}$. This is a very naive nonparametric estimator that can likely be improved, but we will consider it in our analysis to present an additional IPW estimator. Results for the MAR motivating example in Section \ref{Sec_Motivating} are shown in Figures \ref{fig:propensityscore1} and \ref{fig:propensityscore2}, compared to the parametric NMMIPW estimator. Overall, the new estimator is better in this example, as expected, although the increase in variance in terms of RMSE is already visible in Figure \ref{fig:propensityscore1}. 

\begin{figure}[H]
    \centering
    \includegraphics[width=0.8\linewidth]{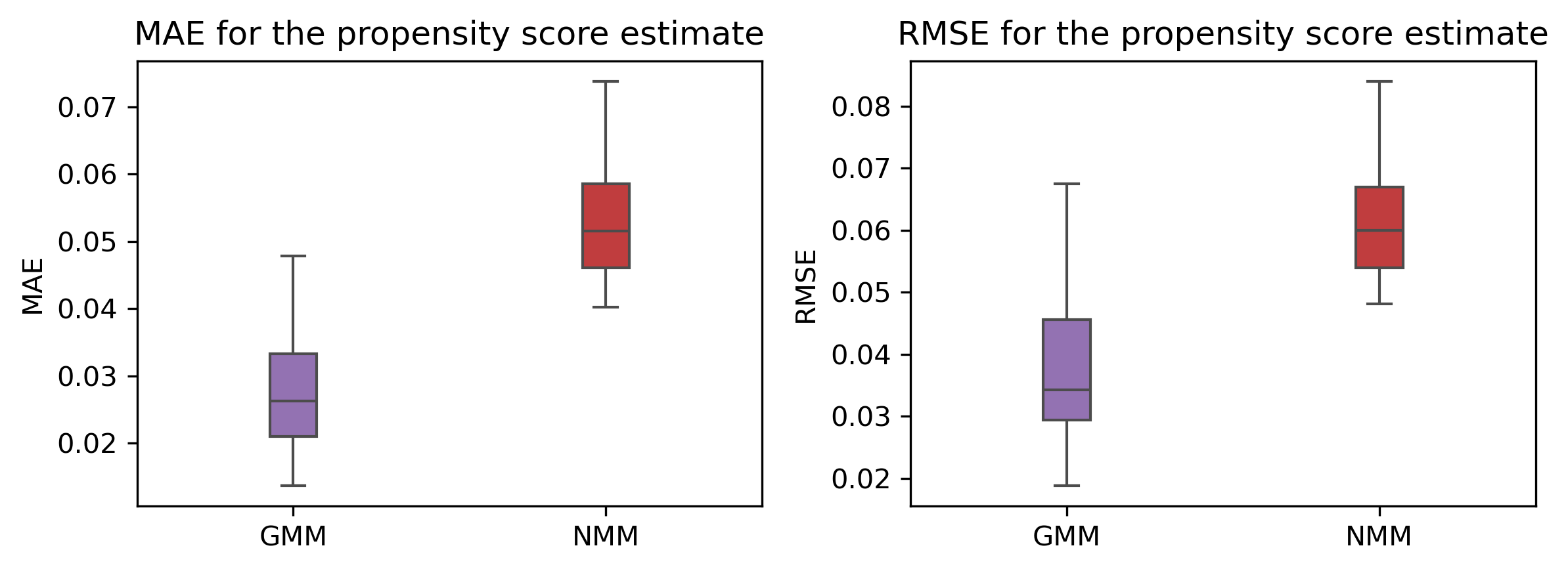}
    \caption{MAE (left) and RMSE (right) of the two estimators of $\Prob(M= \mathbf 0 \mid x)$.}
    \label{fig:propensityscore1}
\end{figure}

\begin{figure}[H]
    \centering
    \includegraphics[width=0.95\linewidth]{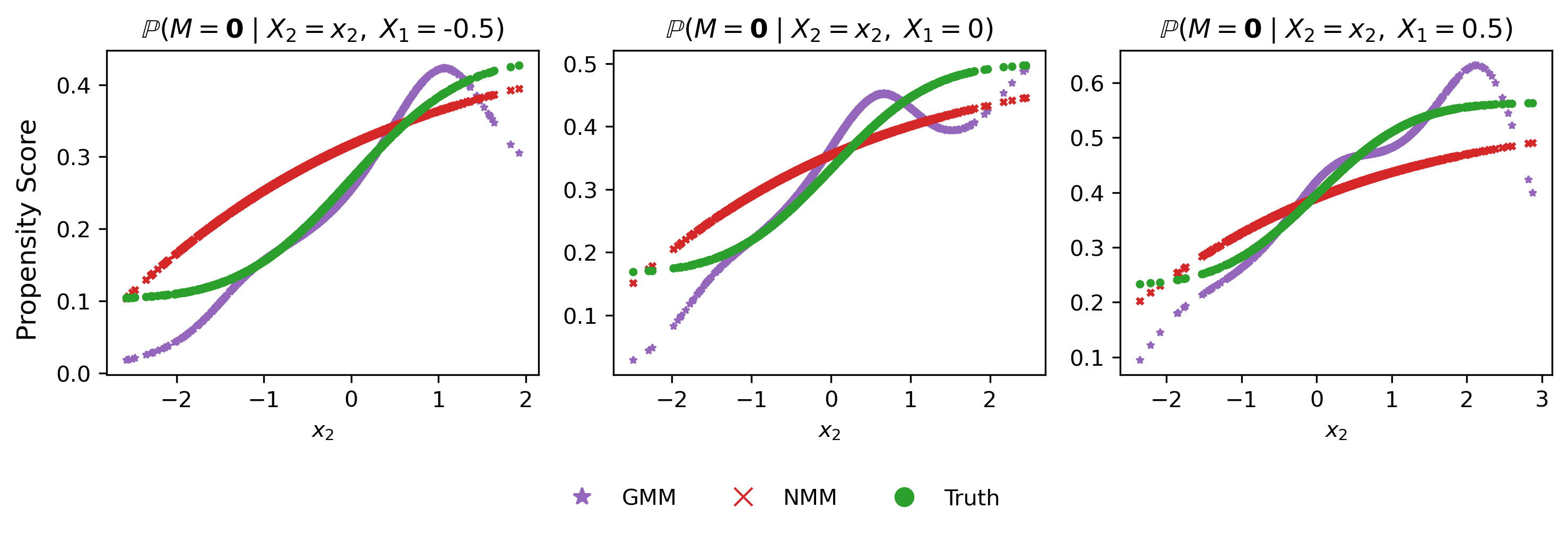}
    \caption{Illustration of the propensity score estimator, for different $X_1$ values. The true function is shown in green, the GMM estimator in purple and the NMMIPW estimator in red.}
    \label{fig:propensityscore2}
\end{figure}

\begin{figure}[H]
    \centering
    \includegraphics[width=0.8\linewidth]{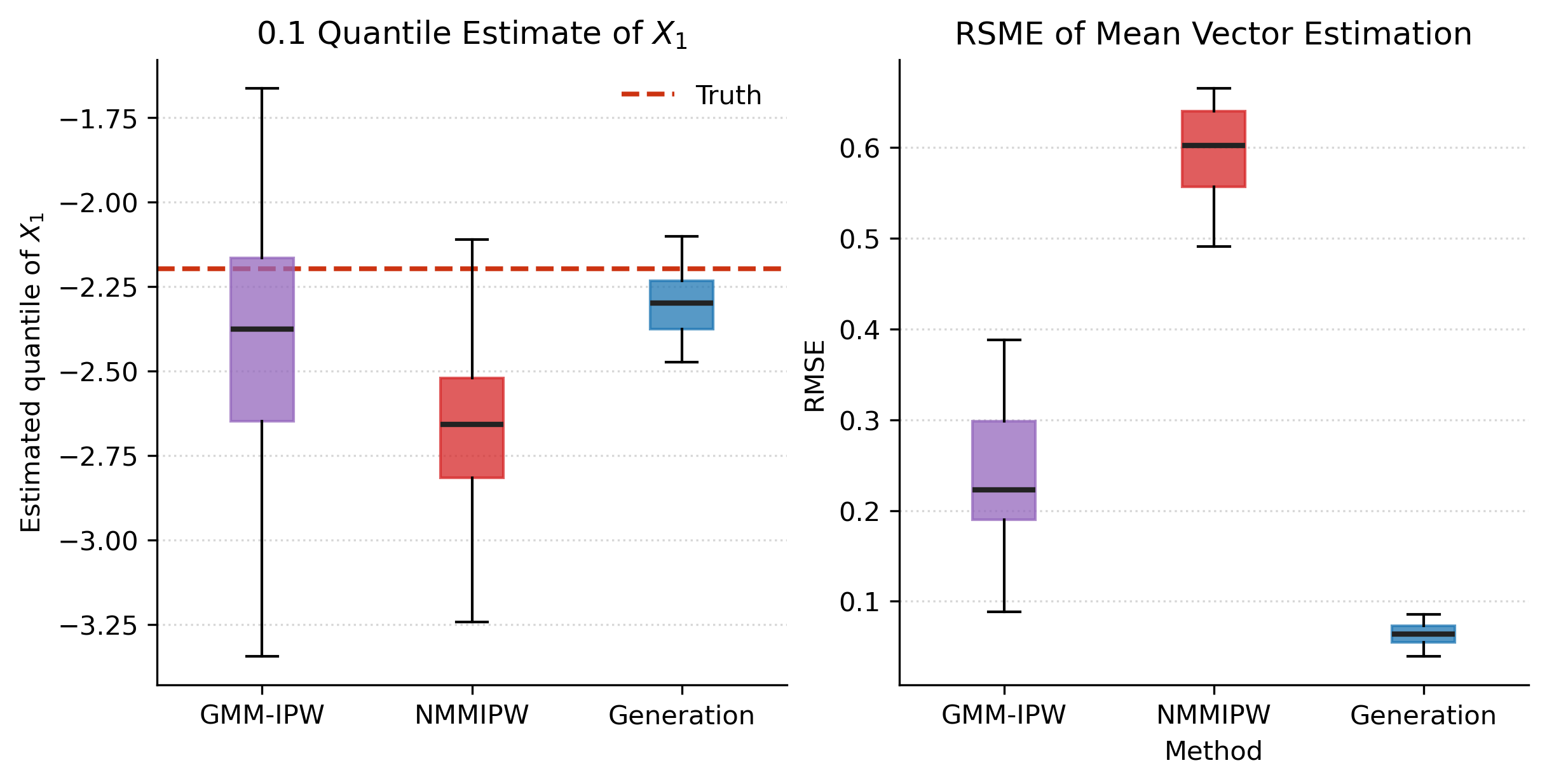}
    \caption{Results for the 0.1 quantile (left) and the RMSE for the mean estimation (right) in case~(III).}
    \label{fig:meanquantileestimtaion}
\end{figure}

\paragraph{Quantile and Mean Vector Estimation}

In Section \ref{Sec_Motivating} we already studied the estimation of the 0.1 quantile for Case (I). We now study the estimation of the same quantile and the mean vector under MAR missingness in the case of the Logistic distribution (Case (III)). We consider the NMMIPW estimator, the estimator based on generation from our density estimate as well as the IPW constructed from the GMM estimate of $\Prob(M=\mathbf{0} \mid x)$ introduced above. In the case of our chosen logistic distribution, the true 0.1 quantile is around -2.20, as compared to the value of around -1.28 for the Gaussian distribution. 

Figure \ref{fig:meanquantileestimtaion} shows the results. The generation estimator displays the best performance, followed by the nonparametric IPW. The NMMIPW estimator displays a rather clear bias, although there is less variance than the GMM-IPW. This bias appears interesting, as ampute uses a similar, albeit different, logistic probability for the propensity score.

\paragraph{Regression Estimation} So far, we considered the recovery of univariate quantities under MAR missingness. Here, we study regression as a simple example of multivariate parameters by regressing $X_1$ onto the remaining 19 variables $X_2, \ldots, X_{20}$ including an intercept. We again consider the Logistic case (III), as the choice of $\Sigma$ makes it possible to derive the population parameter as $\beta=(0.7,0,\ldots,0)^{\top}$. Since we are using Gaussian mixtures with diagonal covariances, we expect this to be more challenging for our estimator.

Figure \ref{fig:regression} on the left displays the result, clearly showing a bias of the GMM-estimator, for both generation and imputation. Inspecting the individual parameters shows that this bias occurs because the first element of $\beta$ (0.7) is underestimated. On the other hand, our approach using ampute to obtain the missing values has a notable little effect here, and in fact a complete-case analysis leads to the best results. As such we also consider Case (I) which can be seen on the right of Figure \ref{fig:regression} as well as in Figure~\ref{fig:regression_coefficients}. In this very simple example we are competitive, though we still appear to underestimate the correlation, which again might not be entirely surprising due to the GMM approach.

  \begin{figure}[htbp]
    \centering
    \begin{minipage}{0.48\textwidth}
      \centering
      \includegraphics[width=0.9\linewidth]{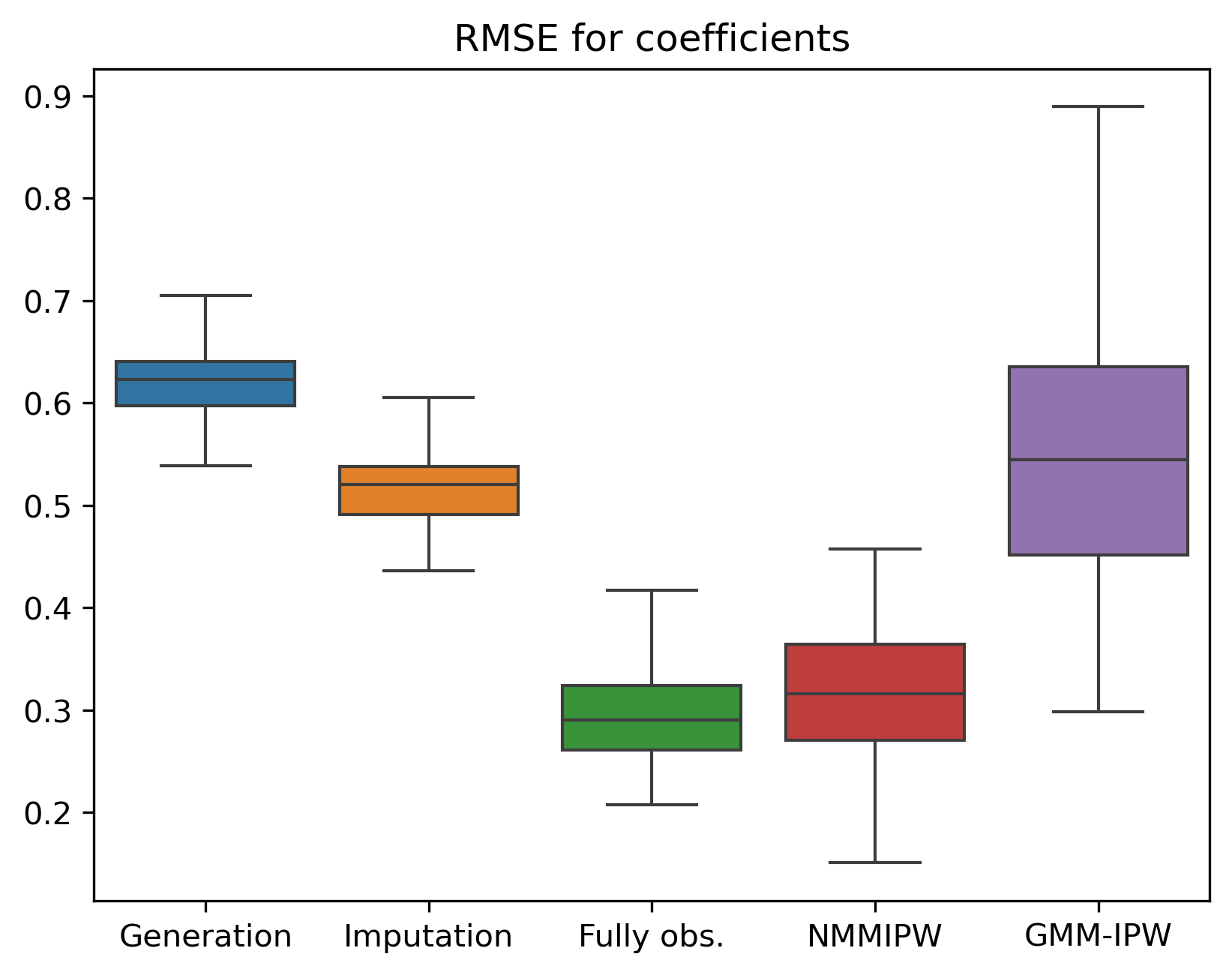}
    \end{minipage}\hfill
    \begin{minipage}{0.48\textwidth}
      \centering
      \includegraphics[width=0.9\linewidth]{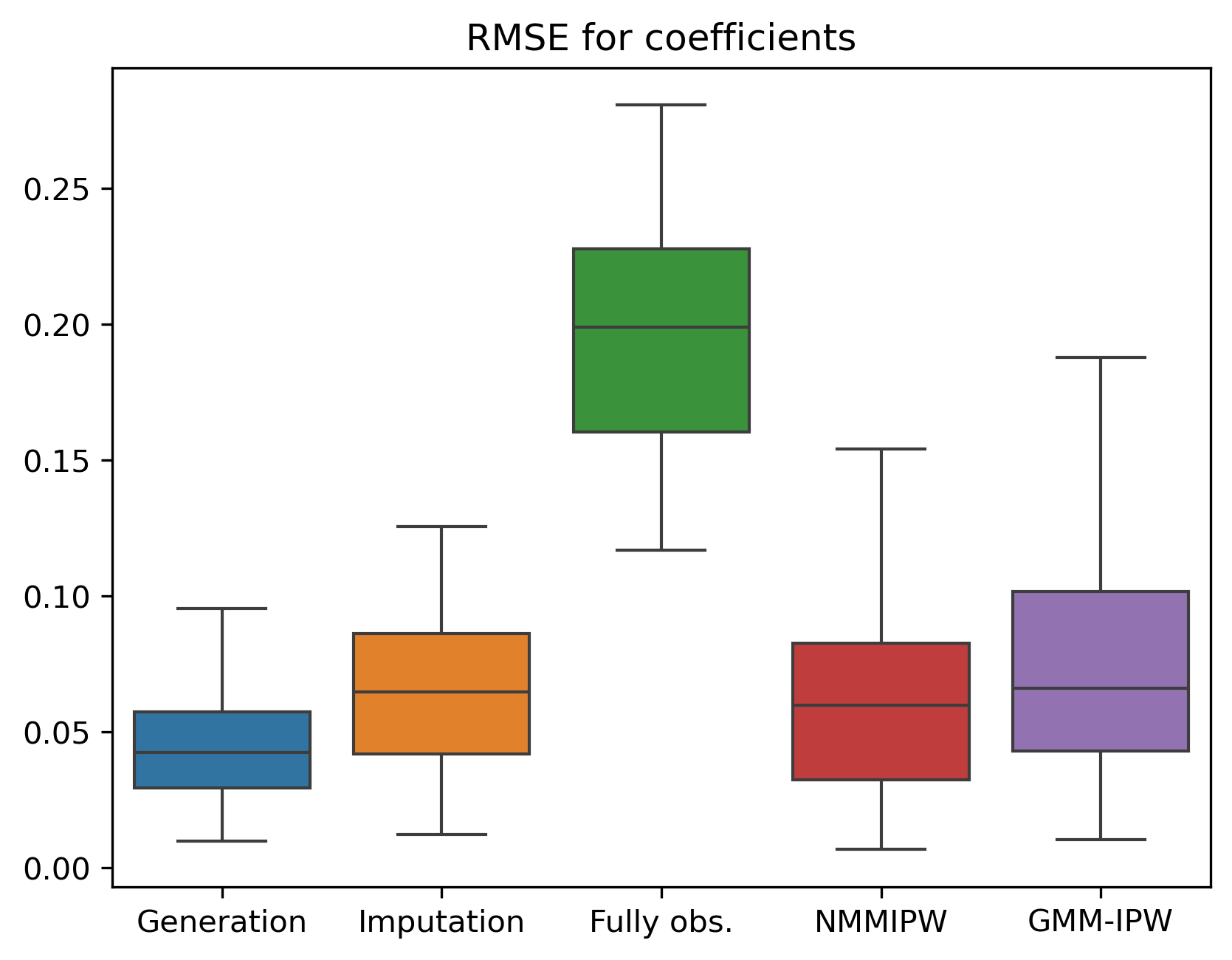}
    \end{minipage}
    \caption{RMSE for the 19 regression coefficients and intercept for Case (III) (Logistic, Left) and the one regression coefficient and intercept for Case (I) (Gaussian, Right). In both cases we add the estimation based only on observations $i$ with $M_i=\mathbf{0}$ (fully observed).}
    \label{fig:regression}
  \end{figure}

    \begin{figure}[htbp]
    \centering
      \includegraphics[width=0.9\linewidth]{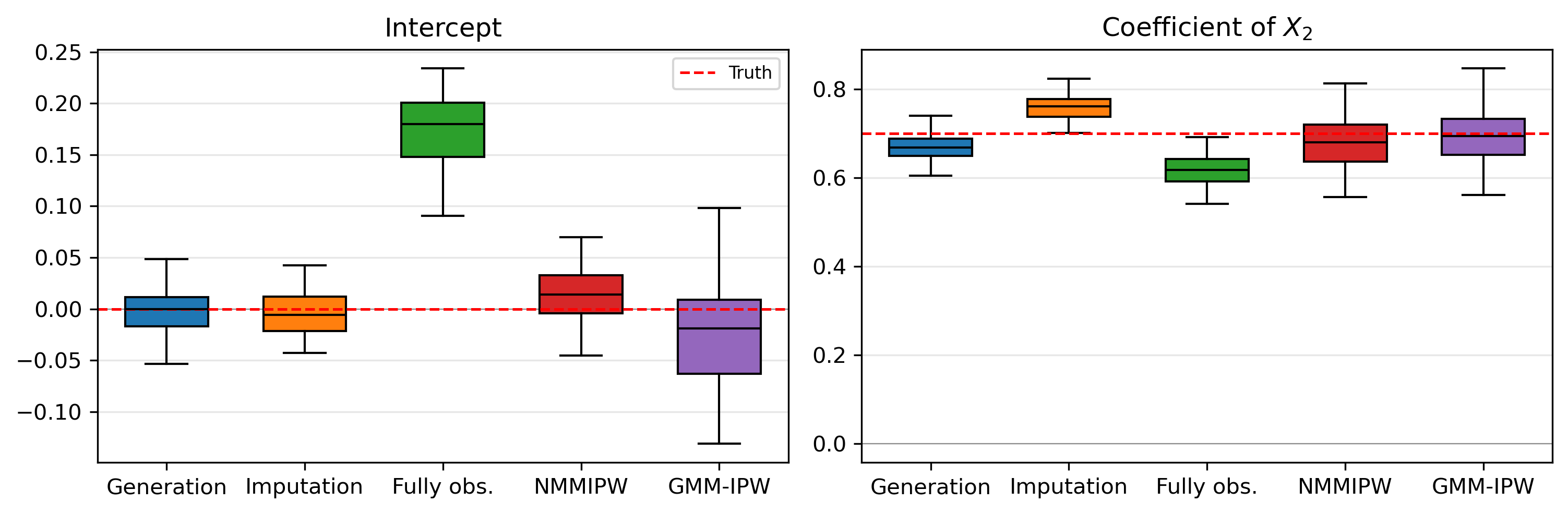}
    \caption{Estimated intercept (left) and coefficient for $X_2$ (right) for Case (I).}
    \label{fig:regression_coefficients}
  \end{figure}

\section{Application}\label{sec_Application}


In this section, we consider three statistical tasks to investigate applications of the proposed density estimator beyond M-estimation: 1. clustering, which makes use of the estimator itself, 2. principal component analysis (PCA) on imputed data, and finally 3. classification via imputation and generation. 

In all task, we consider the Wisconsin breast cancer dataset from the UCI machine learning repository \citep{breast_cancer}. It contains 569 observation of 30 features that are computed from a digitized image of a fine needle aspirate of a breast mass. They describe the characteristics of the cell nuclei present in the image. In particular, ten real-valued features are computed for each cell nucleus, and their mean, standard deviation, and ``worst'' value are recorded. In addition, each observation is diagnosed as malignant or benign. The advantage of this dataset is that the distinction between the two classes is quite clearly visible when plotting the first two principal components. This makes the dataset an ideal candidate to visualize the results. Moreover, the number of data points is rather modest compared to the dimensionality, making this a rather challenging dataset when further information is removed with missing values. To generate the missing values, we define six reasonable missingness patterns and use the ampute function from the mice package in \textsf{R} (see Appendix~\ref{sec_AdditionalResults} for details).

\subsection{Direct Use of the Density Estimator: Clustering under MAR}
\label{subsec_clustering}

Clustering is a crucial task in data science and machine learning \citep{clusteringimportant}. However, it is already challenging with complete-data and even more so with (non-monotone) missingness. Here we use the fact that our density estimator can directly be used for clustering with missing data, without the need for generation or imputation. Although the estimator is nonparametric, in the sense that it is consistent for any Hölder continuous density as per \Cref{Thm:DensityResult}, we can also directly consider the centers of the estimated Gaussian mixture, that approximates the true density, for clustering. We thus use our density estimator to apply this clustering directly to the $d=30$ features and then plot the first two PCA components as an illustration of the chosen clustering. 

Figure \ref{fig:clusteringGMcenters} shows the case where the number of clusters (components) is automatically chosen by BIC. The 13 chosen clusters appear to show very little overlap between the two (unobserved) categories of malignant and benign, as shown in the left graph of Figure \ref{fig:clusteringGMcenters}. However, as expected, the number of clusters were chosen to model the (non-Gaussian) data well, and seems to be much larger than what one would be interested in. In a next step we thus reduce the number of clusters from 13 to 2, once using a naive KMeans method and once using the entropy method of \cite{baudry2010combining}. 

\begin{figure}[H]
    \centering
    \includegraphics[width=0.95\linewidth]{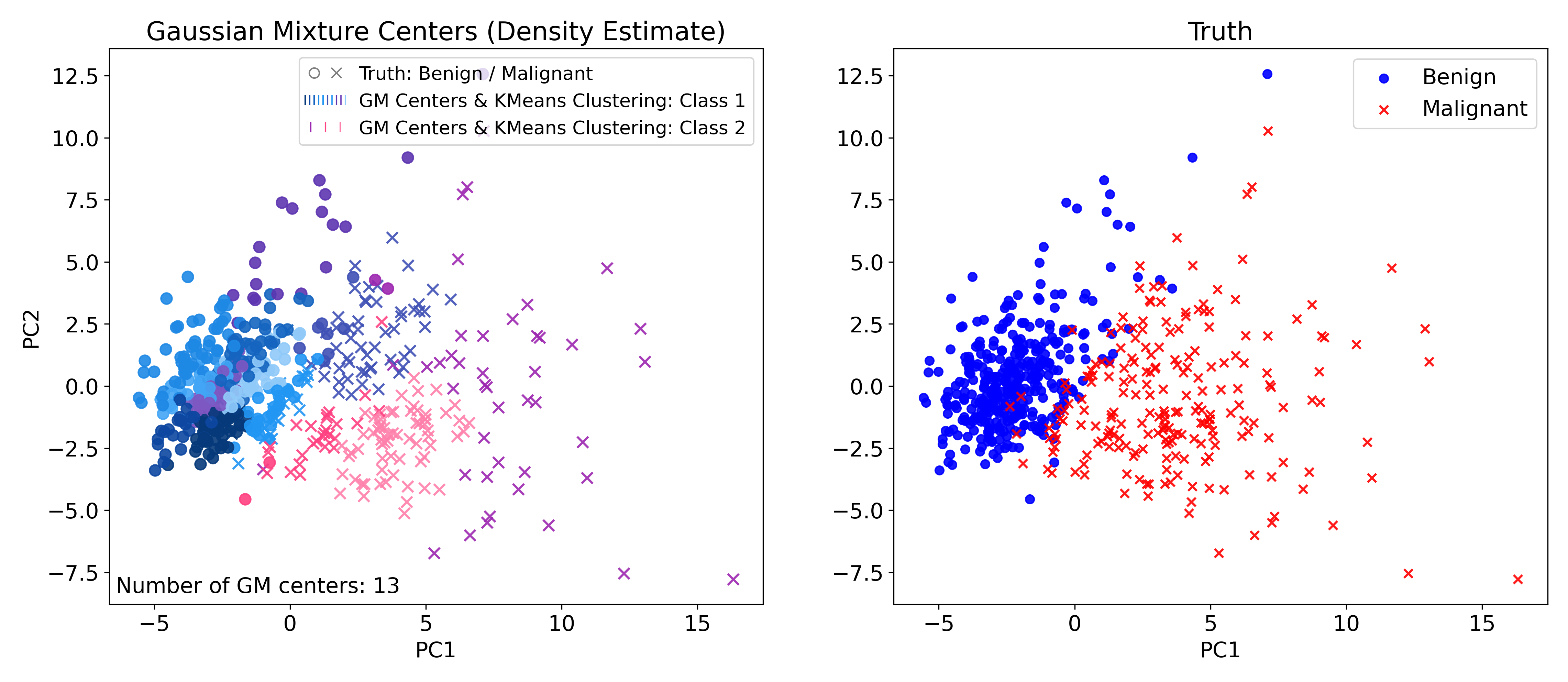}
    \caption{Clustering of the Wisconsin breast cancer dataset (569 observations) with general MAR missingness, using the estimated Gaussian mixture centers. We fit a Gaussian Mixture MAR on the incomplete observations and obtain an estimated Gaussian mixture to approximate the complete data density. The estimated number of components is $\hat N = 13$ here (selected by BIC), resulting in a natural 13-class clustering of the dataset. Each observation is clustered to the most likely Gaussian center, as shown in the different coloring of the left plot. The shape of the data point shown in the plot reveals the true labeling, where ``o" corresponds to BENIGN and ``$\times$" corresponds to MALIGNANT. In the left plot, most data points have a one-to-one pairing of its color and shape, except for the points in dark navy in the middle part relating both shapes. These are the rare cases where the natural Gaussian-centers-clustering fails to find the separation of the true labeling.}
    \label{fig:clusteringGMcenters}
\end{figure}

Algorithm~\ref{algo:clustering} presents the naive KMeans on Gaussian centers to reduce the number of clusters. For the Wisconsin breast cancer dataset, the target number of clusters is $m=2$. We first assign each data point to the corresponding most-likely Gaussian center which takes the largest ``responsibility" value among all centers, then re-assign it to one of the two clusters that this maximum-likelihood center belongs to. The left plot of Figure~\ref{fig:clusteringReducedCenters} shows that this KMeans on the Gaussian centers can separate the two groups to a large extent. It only fails in the boundary region where even in the truth plot on the right, BENIGN and MALIGNANT are fully mixed, as indicated by the tangling of the ``o" in blue and ``$\times$" in red. In addition, we utilize the method in \citet{baudry2010combining}, who develop an entropy method to reduce the number of clusters with the estimated responsibility matrix of the GMM results as input. Unfortunately, on this dataset, the entropy method exhibits unstable performance. While the middle plot of Figure~\ref{fig:clusteringReducedCenters} displays a satisfactory result, replications of the same experiment often yielded a separation into one cluster with most of the data points and the other with very few. 

\begin{figure}[h]
    \centering
    \includegraphics[width=0.95\linewidth]{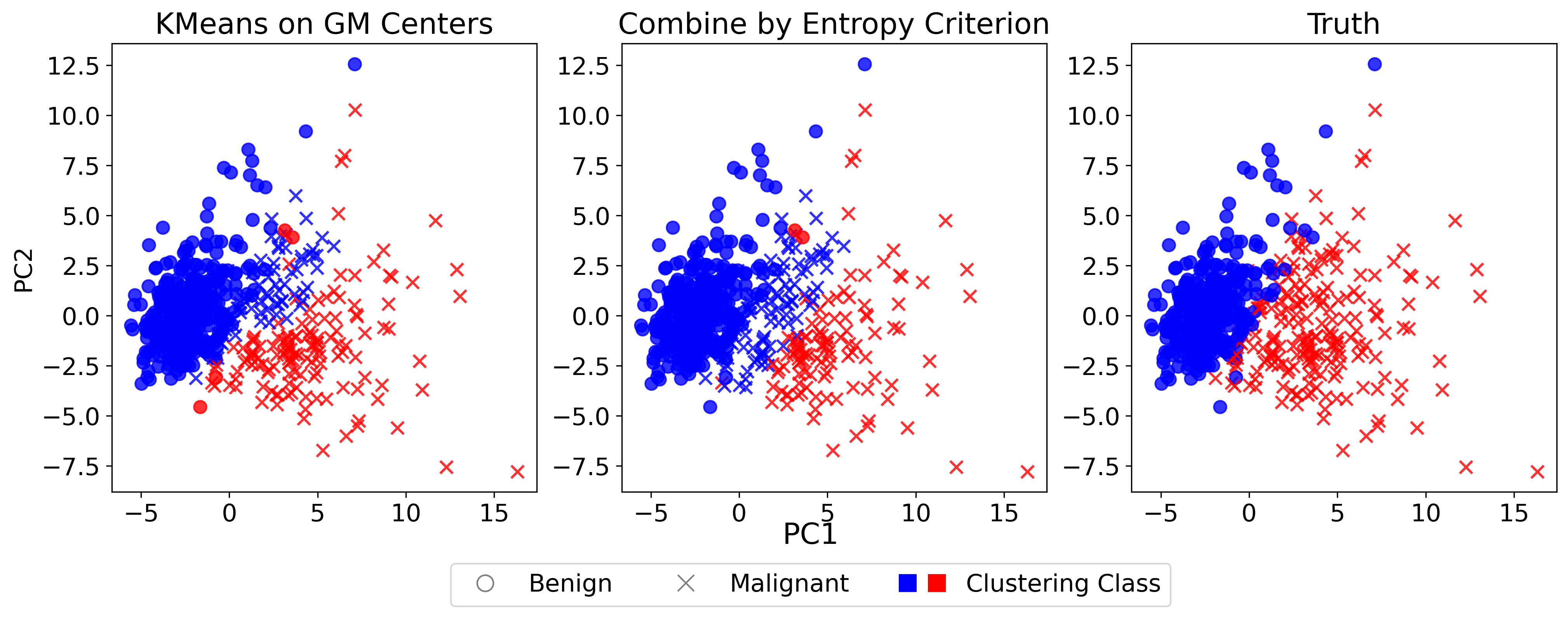}
    \caption{Reduction of number of clusters to the target $m=2$ by KMeans on the Gaussian centers (left), and entropy method (middle), based on the raw clustering resulted from the estimated Gaussian mixture on the MAR missing data as shown in the left plot of Figure~\ref{fig:clusteringGMcenters}. In all three plots, the shape of the data points indicate the true grouping, while the color shows the clustering of each method. The right plot exhibits the true labels with color and shape one-to-one mapped.}
    \label{fig:clusteringReducedCenters}
\end{figure}

For comparison, we conducted the same clustering exercise with a complete dataset. Figure~\ref{fig:clusteringReducedCenterscomplete} shows a similar result as in Figure~\ref{fig:clusteringReducedCenters}. It indicates that the performance of the clustering based on our Gaussian Mixture MAR algorithm does not necessarily weaken with general MAR missingness in the data and may be connected to the fact that the ignoring density estimator has the same consistency rate as for complete data, and MAR missingness only brings in a constant factor. 

\begin{figure}[h]
    \centering
    \includegraphics[width=0.95\linewidth]{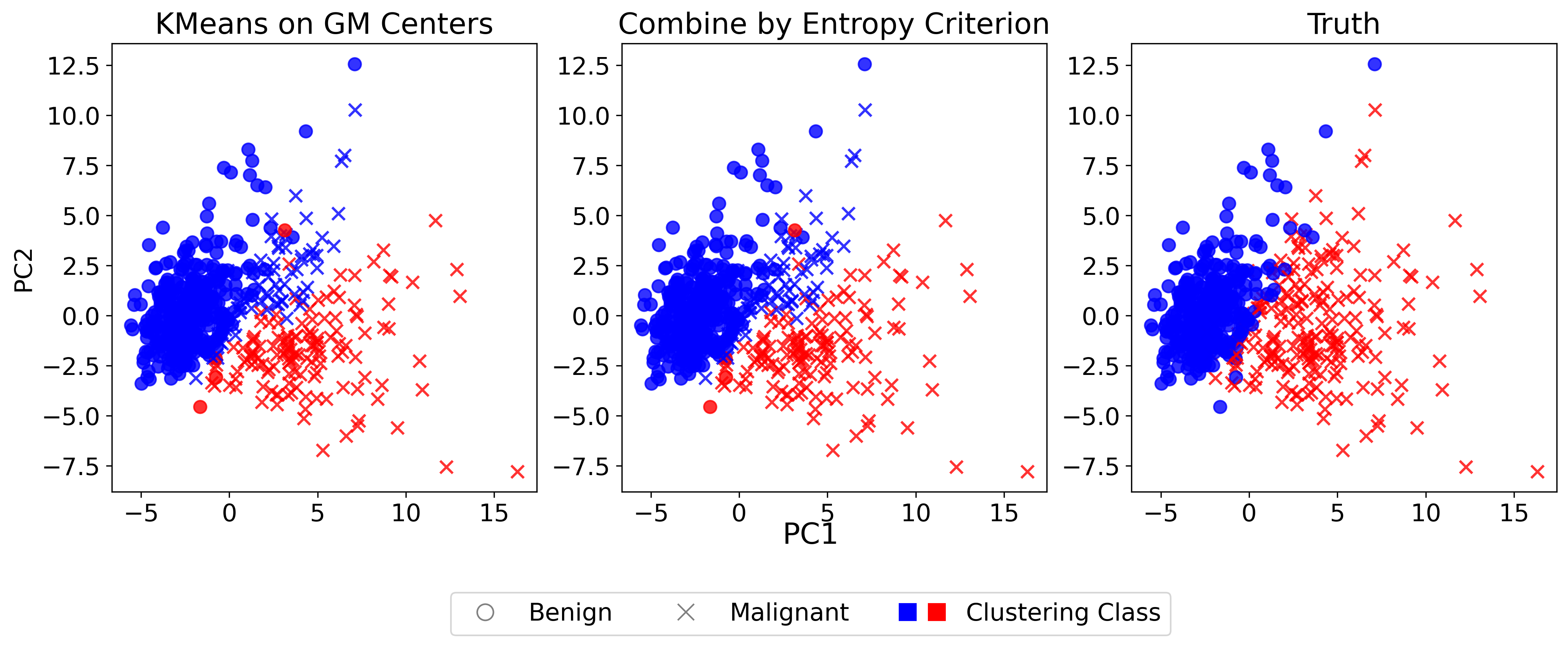}
    \caption{The same exercise as in Figure~\ref{fig:clusteringReducedCenters} only the data is complete.}
    \label{fig:clusteringReducedCenterscomplete}
\end{figure}

\subsection{Imputation: PCA under M(N)AR}
\label{subsec_pca}

As a complementary analysis of the previous clustering, we now want to investigate whether the target label (malignant/benign) is aligned with the dominant, unsupervised structure of the data, and how the conclusions change when the dataset is incomplete and missing values are imputed based on the proposed density estimator. More specifically, we perform a PCA and determine how well the two classes can be separated by a linear support vector machine (SVM) classifier based on the first two principal components, both under MAR and under an MNAR contamination. The breast cancer dataset is split into training (75\%) and test data (25\%). When the linear SVM is fitted on the complete training dataset, it yields a high accuracy of $94.41\%$ on the test data, indicating that the class effect is among the largest sources of variation in the data. Now, the question is whether this result can be reproduced when missing values are introduced. We ampute and fit the proposed density estimator on the feature data only (excluding the target labels) and apply imputation to obtain a complete sample. 
The results are shown in Figure~\ref{fig:cancer_pca}. While the point plots and separation lines based on the imputed datasets differ slightly from those for the uncontaminated original dataset, the prediction accuracy is only marginally reduced. Thus, the final conclusion of the classes being aligned with the dominant structure of the data remains the same. 

\begin{figure}[H]
    \centering
    \includegraphics[width=0.95\linewidth]{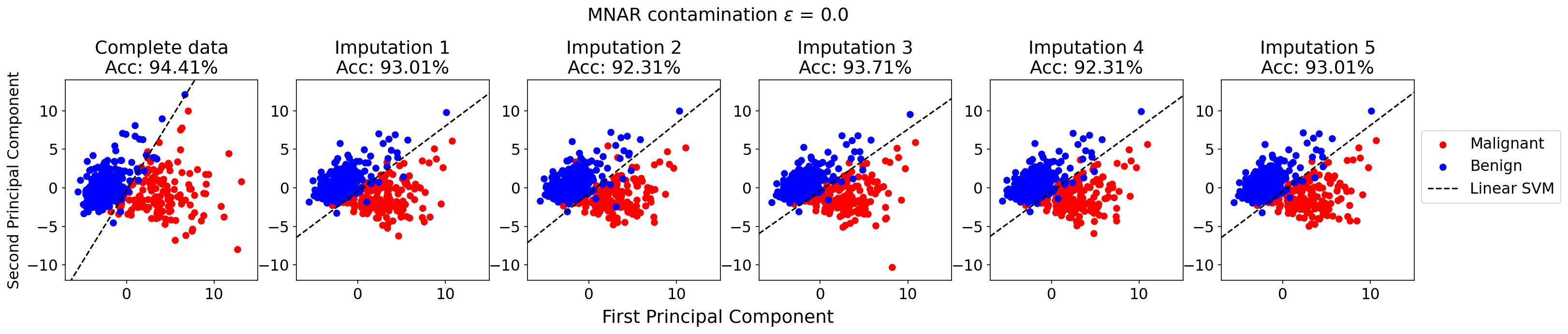}
    \includegraphics[width=0.95\linewidth]{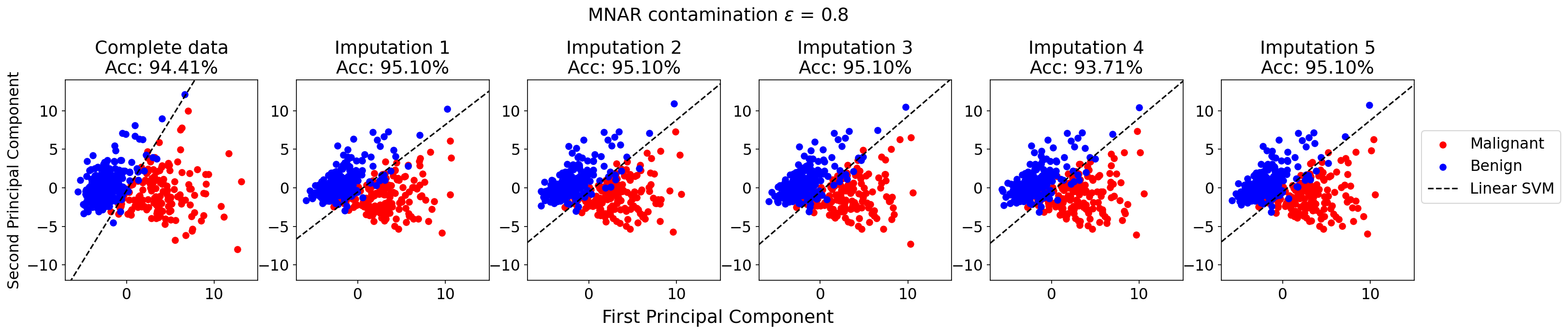}
    \caption{Classification based on the first two principal components. The top row shows the results under MAR, the bottom row under MNAR contamination with $\varepsilon=0.8$. The left most plot in each row is the result on the original complete dataset, the others correspond to five imputations of the same amputed dataset. The prediction accuracy is only marginally reduced under missingness.}
    \label{fig:cancer_pca}
\end{figure}

\subsection{Generation: Classification of Sensitive Data under M(N)AR}
\label{subsec_classification}


We already studied linear regression in Section \ref{Sec_Simulation_MEstimation}, which we extend here to a classification exercise where both features and labels can be missing. We fit a random forest (RF) classifier to the breast cancer data and analyze how the accuracy differs when the dataset is incomplete and missing values are treated with the proposed density estimator. We ampute the training dataset, use the density estimator based on Gaussian mixtures to recreate a complete dataset (either by imputation or generation), fit a classification RF on the resulting complete training dataset, and determine its accuracy on the unmodified test data. Finally, we compare the results to an RF classifier trained on the original complete training data. 
The results are illustrated in Figure~\ref{fig:cancer_classification}. Both imputation and generation based on the proposed density estimator yield a classification accuracy close to the reference on the complete dataset with a slightly higher accuracy under imputation. The performance only degrades marginally under MNAR contamination. 
\begin{figure}[h]
    \centering
    \includegraphics[width=0.9\linewidth]{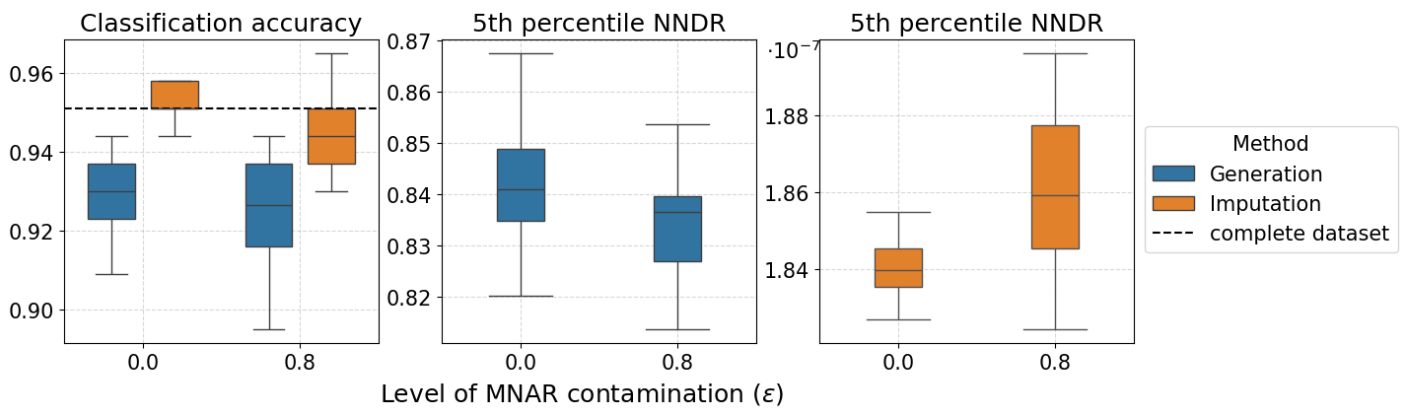}
    \caption{Classification with missingness in the features and target labels. Accuracy of a RF classifier trained on imputed or generated complete data (left) and privacy metric evaluated on the generated (middle) and imputed datasets (right). Note that the $y$-axis is scaled by $10^{-7}$ in the right most plot, showing that the imputation method does not preserve privacy at all.}
    \label{fig:cancer_classification}
\end{figure}

An intuitive advantage of the generation method (as opposed to imputation) is its better suitability for privacy sensitive applications where the original data should not be disclosed.\footnote{Though we note that for formal privacy guarantees the EM algorithm would need to be adapted, for instance using ideas in \cite{EMDP}.} To assess empirically how well privacy is preserved by our method, we compute the Nearest Neighbor Distance Ration (NNDR). For each generated datapoint, it is defined as the ratio between the nearest and second nearest original observation. When the ratio is about $1$, the synthetic record is equidistant to multiple original records, which is desirable for privacy. When the ratio is very low, the synthetic record is much nearer to one real record than to the rest, which suggests it may leak information about that entity and thus poses a privacy risk. To robustly summarize the NNDR values over all datapoints, we use the 5th percentile, capturing the worst case. In Figure~\ref{fig:cancer_classification}, it can be seen that generation yields high values around 0.81--0.86, while imputation entails values very close to~$0$, which is to be expected. \cite{synth_data_cancer} performed a similar analysis comparing several synthetic data generation methods for complete data. They also used the breast cancer dataset, computed the RF classifier accuracy and the 5th percentile NNDR. Compared to their results, our generation method achieves a better value of the privacy metric (the highest value reported in the paper is around $0.8$), 
but a lower prediction performance (reported accuracies lie roughly between $0.92$ and $0.98$). In view of the fact that we introduce missing values and our method thus has access to less information of the original dataset, this is not surprising. 

\section{Conclusion}\label{sec_conclusion}


The main message of this paper is that, for nonparametric density estimation, modeling the missingness mechanism is not necessary, and that guarantees can hold without restricting to narrower cases within MAR. In fact, they hold under a `MAR-like MNAR' condition even weaker than MAR. 
We achieve this through a sieve MLE approach. Theorem~\ref{thm:8mar} and Corollary~\ref{cor:MARlike} provide a general rate of convergence for the sieve MLE, showing that missingness affects the asymptotic rate only through a constant related to the positivity assumption. 
One important implication from it is that, since our sieve MLE uses all of the data in the likelihood while avoiding modeling the missingness mechanism, there is no need to discard samples, even those whose missingness pattern is observed only rarely.
The `MAR-like MNAR' condition cannot be weakened further in the sense that it must hold for the ignoring MLE to target the truth. If not, and if the solution is still non-empty, asymptotic bias will arise that may however be bounded by our measure of deviation from MAR.

Our result can give rise to many theoretically grounded algorithms for generating new data that follows the underlying target distribution or imputating the original dataset. The Gaussian mixture estimator discussed in Section~\ref{sec_densityestimation} is one example. Adapting it to new settings and targets opens new problems. One is to make our density estimator less sensitive to the curse of dimensionality. As discussed in Remark~\ref{rem:minimax}, our ignoring estimator already achieves a near-parametric rate if the true density is a finite mixture of Gaussians rather than a general smooth density. More elaborate constructions might exploit other lower-dimensional structure in the data to obtain near-parametric rates as well. Model selection under MAR also remains open. The choice of $N$ is not trivial, and our BIC/AIC approach only presents a first ad-hoc solution; this is exacerbated by the fact that the EM algorithm can get stuck in local optima. Assuming $\densityspacemodel$ to be the set of all Gaussian mixtures and changing the sieve in Theorem~\ref{Thm:DensityResult} slightly, it is likely possible to obtain similar results for an ignoring version of the estimator of \cite{MixtureofGaussians}. For complete data, this estimator turns estimation into an (infinite-dimensional) convex problem, without the need to choose $N$. 
Finally, our rate result is a first step toward efficient inference for functionals of the density under general MAR, covering most estimands of applied interest, such as means, quantiles and regression coefficients, a promising route to provide a formal basis to the common claim that an analysis is ``valid under MAR'' for the larger science comminity.

\section*{AI Statement}
We made limited use of LLMs, mainly using different versions of Claude AI as a guide to help understand the source material better, to help with plotting and to create initial code. For instance, J.N. used Claude AI to create a faster Python version of the \textsf{R} code in \cite{MARinverseweighting}. The only exception to the above was the use of Claude AI to provide feedback on proofs in the final version of the paper, which uncovered a technical error in Theorem \ref{thm:8mar}. The correction of this error was done by the authors, leading to a changed condition \eqref{eq:entropy} compared to \cite{kaji2026hellinger}.

\clearpage

\appendix

\section{Missing-Value Density Estimation and Imputation}

\subsection{Algorithm}

Here we first present the MAR Gaussian Mixture EM algorithms used to estimate the complete-data density. It includes two major layers. Algorithm~\eqref{algo:base} is the base model to find a (diagonal) Gaussian Mixture with fixed number of components approximately maximizing the observed-data likelihood. Algorithm~\eqref{algo:modelselection} uses the BIC/AIC criterion to select from a range of base models with different numbers of components. Then we show in Algorithm~\eqref{algo:impute} the imputation method based on the estimated complete data density.

\begin{algorithm}
    \caption{Density Estimation with MAR Gaussian Mixture (Base Model)}
    \label{algo:base}
    \begin{algorithmic}[1]
        \Require Data $X^{(M)}=\{x_i^{(m_i)}\}_{i=1}^{n} \subset \R^d$, number of components $N$, covariance type \texttt{cov\_type}, maximum iterations $T$, tolerance $\varepsilon$, regularization of covariance $\lambda$, number of initializations $M$
        \Ensure Estimated parameters $\Theta=\{\pi_k,\mu_k,\Sigma_k\}_{k=1}^{N}$, $\sum_{k=1}^N \pi_k = 1$, $\Sigma_k = \operatorname{diag}(\sigma_{k1}^2,\dots, \sigma_{kd}^2)$ $\forall \, k$

        \State $\mathcal S \gets \{\, i \mid m_{i} = \mathbf{0} \in \R^d,  \, i = 1, \dots, n \, \}$ \Comment{Indices of fully observed data points}
        \State $\ell_{\text{best}} \gets -\infty$
        \For{$m = 1,\dots, M$}
            \Comment{Multiple initiations}
            \State \textbf{Initialize}:
            \State $\mathcal I \gets \mathrm{Sample}(\mathcal S, N)$ \Comment{Uniformly without replacement}
            \State $\pi_k \gets \frac1N$, $\mu_k \gets x_{\mathcal I_k}$, $\Sigma_k \gets I_d$ for $k=1,\dots,N$
            \State $\ell_{\text{old}} \gets -\infty$
            \For{$t = 1, \dots, T$}
            
                \State \Call{UpdateResponsibility}{$X^{(M)}, \{\mu_k\}_{k=1}^N, \{\Sigma_k\}_{k=1}^N, \{\pi_k\}_{k=1}^N$} \Comment{E-step, Algorithm~\eqref{algo:resp}}
                \For{$k=1,\dots,N$} \Comment{M-step for mixture weights}
                    \State $\pi_{k} \gets \frac{1}{n} \sum_{i=1}^{n}\gamma_{ik}$ 
                \EndFor
                \For{$i=1,\dots,n$} \Comment{M-step for means}
                    \For{$k=1,\dots,N$}
                        \State $\mu_{kv} \gets \frac{1}{\sum_{i=1}^{n} (1 - m_{iv}) \gamma_{ik}} \sum_{i=1}^n (1 - m_{iv}) \gamma_{ik} x_{iv}$
                    \EndFor
                \EndFor
                
                \Comment{M-step for covariance matrices, Algorithm \eqref{algo:cov}}
                \State \Call{UpdateCovariance}{$X^{(M)}, \{\gamma_{ik}\}_{i=1,\dots,n; \, k=1,\dots,N}, \{\mu_k\}_{k=1}^N$, \texttt{cov\_type}, $\lambda$} $\to \{\Sigma_k\}_{k=1}^N$
                
                \For{$k=1,\dots,N$}
                    \State $\ell_k \gets \sum_{i=1}^n\sum_{v=1}^d (1 - m_{iv}) \phi\left({x_{iv} - \mu_{kv} \over \sigma_{kv}}\right) + \log \pi_k$
                \EndFor
                
                \State $\ell \gets \log\sum_{k=1}^N\exp{(\ell_k)}$ \Comment{Log-likelihood}
                
                \If{$|\ell-\ell_{\mathrm{old}}|<\varepsilon$}
                    \State \textbf{break}
                \EndIf
                \State{$\ell_{\text{old}} \gets \ell$}
            \EndFor
        \If{$\ell > \ell_{\text{best}}$}
            \State $\Theta_{\text{best}} \gets \Theta$
            \State $\ell_{\text{best}} \gets \ell$
        \EndIf   
        \EndFor
        \State \Return{$\Theta_{\text{best}}, \ell_{\text{best}}$}
    \end{algorithmic}
\end{algorithm}

\paragraph{E-step} Under MAR assumptions and distinctness of the missingness and data parameters, for all $k, \ell \in \{1, \dots, N\}$, we have,
\begin{equation}
    \label{eq:Estep}
    \Prob(M_i = m_i \mid x_i^{(m_i)}, \theta_k^{(t)}) = \Prob(M_i = m_i \mid x_i^{(m_i)}, \theta_\ell^{(t)}) = \Prob(M_i = m_i \mid x_i^{(m_i)}).
\end{equation}

The posterior of latent variables $Z_{i,k}, \, i=1,\dots,n, \, k=1,\dots,N$ given the observed data and current parameters $\pi^{(t)}, \, \boldsymbol{\theta}^{(t)}$ is
\[
 p(Z_{ik} \mid x_i^{(m_i)}, \boldsymbol{\theta}^{(t)})=  {\pi^{(t)}_k p_{\theta_k^{(t)}}^{(m_i)} (x_i^{(m_i)}) \Prob(M_i = m_i \mid x_i^{(m_i)}, \theta_k^{(t)}) \over \sum_{\ell=1}^N \pi^{(t)}_\ell p_{\theta_\ell^{(t)}}^{(m_i)} (x_i^{(m_i)}) \Prob(M_i = m_i \mid x_i^{(m_i)}, \theta_\ell^{(t)})}.
\]

Combined with \eqref{eq:Estep}, we obtain the update of the responsibility,
\[
\gamma_{ik}^{(t+1)} = {\pi^{(t)}_k p_{\theta_k^{(t)}}^{(m_i)} (x_i^{(m_i)}) \over \sum_{\ell=1}^N \pi^{(t)}_\ell p_{\theta_\ell^{(t)}}^{(m_i)} (x_i^{(m_i)})},
\]
where $p_\theta^{(m)}$ is the marginal density of the observed data and can be expressed as
\[
p_{\theta_k^{(t)}}^{(m_i)} (x_i^{(m_i)}) = \prod_{v=1}^d (1 - m_{iv})  \phi \left(\frac{x_{iv} - \mu_{kv}^{(t)}}{\sigma_{kv}^{(t)}}\right),
\]
with $\phi$ the p.d.f. of standard normal.

\begin{algorithm}
\caption{UpdateResponsibility($X^{(M)}, \{\mu_k\}_{k=1}^N, \{\Sigma_k\}_{k=1}^N, \{\pi_k\}_{k=1}^N$)}
\label{algo:resp}
\begin{algorithmic}[1]
\For{$i=1,\dots,n$}
\For{$k=1,\dots,N$}
\State $p_{ik} \gets \prod_{v=1}^d (1 - m_{iv}) \phi \left({x_{iv} - \mu_{kv} \over \sigma_{kv}}\right)$
\EndFor
\EndFor
\For{$i=1,\dots,n$}
\For{$k=1,\dots,N$}
\State $\gamma_{ik} \gets \frac{\pi_k p_{ik}}{\sum_{\ell=1}^N \pi_j p_{i\ell}}$
\EndFor
\EndFor
\State \Return $\{\gamma_{ik}\}_{i=1,\dots,n; \, k=1,\dots, N}$
\end{algorithmic}
\end{algorithm}

\paragraph{M-step}
Similarly as in the complete-data mixture model, the mixture weights are updated using the responsibilities
\[
\pi_k^{(t+1)} = \frac{\sum_{i=1}^{n}\gamma_{ik}^{(t+1)}}{n}.
\]

We update the location and covariance parameters by maximizing the expected log-likelihood of the observed data with respect to the posterior of the latent variables.
\begin{equation*}
    \begin{split}
        Q &= 
        \mathbb{E}_{Z \mid X,\Theta^{(t)}}
        \bigl[
        \sum_{i=1}^{n}
        \sum_{k=1}^{N} \one\{Z_{ik}=1\}  \log \left(
        \pi_k \,
        p_{\theta_k}^{(m_i)} (x_i^{(m_i)})\Prob(M_i = m_i \mid x_i^{(m_i)})
        \right)
        \bigr] \\
        &=
        \sum_{i=1}^{n}
        \sum_{k=1}^{N} \mathbb{E}_{Z \mid X,\Theta^{(t)}}
        \bigl[
         \one\{Z_{ik}=1\} \log \left(
        \pi_k \,
        p_{\theta_k}^{(m_i)} (x_i^{(m_i)})\Prob(M_i = m_i \mid x_i^{(m_i)})
        \right)
        \bigr] \\
        &=
        \sum_{i=1}^{n}
        \sum_{k=1}^{N} \mathbb{E}_{Z \mid X,\Theta^{(t)}}
        \bigl[
         \one\{Z_{ik}=1\}
        \bigr]
        \log \left(
        \pi_k \,
        p_{\theta_k}^{(m_i)} (x_i^{(m_i)})\Prob(M_i = m_i \mid x_i^{(m_i)})
        \right)
        \\
        &=
        \sum_{i=1}^{n}
        \sum_{k=1}^{N} \gamma_{ik} 
        \log \left(
        \pi_k \,
        p_{\theta_k}^{(m_i)} (x_i^{(m_i)})\Prob(M_i = m_i \mid x_i^{(m_i)})
        \right)
        \\
        &=
        \sum_{i=1}^{n}
        \sum_{k=1}^{N} \gamma_{ik} 
        \sum_{v=1}^d (1 - m_{iv})
         \left[- \frac{(x_{iv} - \mu_{kv})^2}{2{\sigma_{kv}}^2} - \log \sigma_{kv}\right]
        +
        \sum_{i=1}^{n}\sum_{k=1}^{N}
        \gamma_{ik} 
        \log\pi_k
        \\
        &+ \text{irrelevant term}.
    \end{split}
\end{equation*}

\paragraph{\emph{M-step for means}} The maximum-likelihood update of $\mu_{kv}$ is given by the first order condition 
$\frac{\partial Q}{\partial \mu_{kv}} 
= 
\sum_{i=1}^{n}
\gamma_{ik} (1 - m_{iv})
\frac{(x_{iv} - \mu_{kv})}{\sigma_{kv}^2} = 0$, that
\[
\mu_{kv}^{(t+1)} = \frac{1}{\sum_{i=1}^{n} (1 - m_{iv}) \gamma_{ik}^{(t+1)}} \sum_{i=1}^n (1 - m_{iv}) \gamma_{ik}^{(t+1)} x_{iv}.
\]

This update, letting $m_{ik}=0$ for all $i,k$, coincides with the common complete-data update:
\[
\mu_{kv}^{(t+1)} = \frac{1}{\sum_{i=1}^{n} \gamma_{ik}^{(t+1)}} \sum_{i=1}^n \gamma_{ik}^{(t+1)} x_{iv}.
\]

\paragraph{\emph{M-step for covariance matrices}} We consider four cases: 1) different diagonal covariance for each component; 2) different scaled identity covariance for each component; 3) shared diagonal covariance, and 4) shared scaled identity covariance.

\begin{itemize}
    \item \emph{Different diagonal covariance}. 
    $
    \frac{\partial Q}{\partial \sigma_{kv}}
    = \sum_{i=1}^{n}
    \gamma_{ik} (1 - m_{kv}) \left(\frac{(x_{kv} - \mu_{kv})^2}{\sigma_{kv}^3} - \frac{1}{\sigma_{kv}}\right)
    = 0
    $
    gives the maximum-likelihood updates of $\sigma_{kv}^2$:
    \[
    {\sigma_{kv}^2}^{(t+1)} = \frac{1}{\sum_{i=1}^{n} (1 - m_{iv})\gamma_{ik}^{(t+1)}} \sum_{i=1}^n (1 - m_{iv}) \gamma_{ik}^{(t+1)} \left(x_{iv} - \mu_{kv}^{(t+1)}\right)^2.
    \]
    
    Letting $m_{ik}=0$ for all $i,k$, it reduces to
    \[
    {\sigma_{kv}^2}^{(t+1)} = \frac{1}{\sum_{i=1}^{n} \gamma_{ik}^{(t+1)}} \sum_{i=1}^n \gamma_{ik}^{(t+1)} \left(x_{iv} - \mu_{kv}^{(t+1)}\right)^2.
    \]

    \item \emph{Different scaled identity}. We have
    $
    \sigma_{k v} = \sigma_{k w}$ for all $v, w \in \{1, \dots, d\}$ and for all $k \in \{1, \dots, N\}.
    $
    Let $\sigma_{k\cdot} = \sigma_{k v}$ for all $v$. We can show a reduced first order condition, 
    $
    \frac{\partial Q}{\partial \sigma_{k\cdot}}
    = 
    \sum_{i=1}^{n}
    \gamma_{ik} \sum_{v=1}^d (1 - m_{iv}) \left(\frac{(x_{iv} - \mu_{kv})^2}{\sigma_{k\cdot}^3} - \frac{1}{\sigma_{k\cdot}}\right) = 0
    $
    which gives
    \[
    {\sigma_{k\cdot}^2}^{(t+1)} = \frac{1}{\sum_{i=1}^n \gamma_{ik}^{(t+1)} \sum_{v=1}^d (1 - m_{iv})} \sum_{i=1}^{n}
    \gamma_{ik}^{(t+1)} \sum_{v=1}^d (1 - m_{iv}) (x_{iv} - \mu_{kv}^{(t+1)})^2.
    \]

    Letting $m_{ik} = 0$ for all $i,k$, it reduces to
    \[
    {\sigma_{k\cdot}^2}^{(t+1)} = \frac{1}{\sum_{i=1}^n \gamma_{ik}^{(t+1)}} \sum_{i=1}^{n}
    \gamma_{ik}^{(t+1)} \frac{1}{d} \sum_{v=1}^d (x_{iv} - \mu_{kv}^{(t+1)})^2.
    \]

    \item \emph{Shared diagonal covariance}. We have $\sigma_{kv} = \sigma_{\ell v}$ for all $k, \ell \in \{1, \dots, N\}$ and  for all $v \in \{1,\dots,d\}$.
    Let $\sigma_{\cdot v} = \sigma_{kv}$ for all $k$. We deduce the condition 
    $
    \frac{\partial Q}{\partial \sigma_{\cdot v}} =
    \sum_{i=1}^{n}
    \sum_{k=1}^{N} \gamma_{ik} 
    (1 - m_{iv}) 
    \left[\frac{(x_{iv} - \mu_{kv})^2}{{\sigma_{\cdot v}}^3} - \frac{1}{\sigma_{\cdot v}}\right]
    = 0
    $,
    which gives the update
    \[
    {\sigma_{\cdot v}^2}^{(t+1)} = \frac{1}{\sum_{i=1}^{n}\sum_{k=1}^N \gamma_{ik}^{(t+1)} (1 - m_{iv})}  \sum_{i=1}^{n} \sum_{k=1}^{N}
    (1 - m_{iv})\gamma_{ik}^{(t+1)} (x_{iv} - \mu_{kv}^{(t+1)})^2.
    \]
    
    If $m_{ik} = 0$ for all $i,k$, it reduces to
    \[
    {\sigma_{\cdot v}^2}^{(t+1)} = \frac{1}{n}  \sum_{i=1}^{n} \sum_{k=1}^{N}
    \gamma_{ik}^{(t+1)}(x_{iv} - \mu_{kv}^{(t+1)})^2.
    \]

    \item \emph{Scaled identity}. Let $\Sigma = \sigma^2 I_d$. Similarly, we can show the condition 
    $\frac{\partial Q}{\partial \sigma} =
    \sum_{i=1}^{n}
    \sum_{k=1}^{N} \gamma_{ij} 
    \sum_{v=1}^d (1 - m_{iv})
    \left[\frac{(x_{iv} - \mu_{kv})^2}{{\sigma}^3} - \frac{1}{\sigma}\right] = 0$,
    and obtain the update,
    \[
    {\sigma^2}^{(t+1)} = \frac{1}{\sum_{i=1}^{n}
    \sum_{k=1}^{N} \gamma_{ik}^{(t+1)} 
    \sum_{v=1}^d (1 - m_{iv})} \sum_{i=1}^{n}
    \sum_{k=1}^{N} \gamma_{ik}^{(t+1)} 
    \sum_{v=1}^d (1 - m_{iv}) (x_{iv} - \mu_{kv}^{(t+1)})^2.
    \]

    If $m_{ik} = 0$ for all $i,k$, it reduces to
    \[
    {\sigma^2}^{(t+1)} = \frac{1}{n
    \sum_{k=1}^{N} \gamma_{ik}^{(t+1)}} \sum_{i=1}^{n}
    \sum_{k=1}^{N} \gamma_{ik}^{(t+1)} \frac{1}{d}
    \sum_{v=1}^d (x_{iv} - \mu_{kv}^{(t+1)})^2.
    \]

\end{itemize}

\begin{algorithm}
\caption{UpdateCovariance($X^{(M)}, \{\gamma_{ik}\}_{i=1,\dots,n; \, k=1,\dots,N}, \{\mu_k\}_{k=1}^N$, \texttt{cov\_type}, $\lambda$)}
\label{algo:cov}
\begin{algorithmic}[1]


\For{$i=1,\dots,n$}
    \For{$k=1,\dots,N$}
        \For{$v=1,\dots,d$}
            \State $y_{ikv} \gets \gamma_{ik} (1 - m_{iv})(x_{iv} - \mu_{kv})^2$
            \State $r_{ikv} \gets \gamma_{ik} (1 - m_{iv})$
        \EndFor
    \EndFor
\EndFor
\If{\texttt{cov\_type} = \texttt{diag}} 
\State $\sigma_{kv}^2 \gets \sum_{i=1}^n y_{ikv} / \sum_{i=1}^n r_{ikv}$
\EndIf
\If{\texttt{cov\_type} = \texttt{id}} 
\State $\sigma_{kv}^2 \gets \sum_{i=1}^n \sum_{v=1}^d y_{ikv} / \sum_{i=1}^n \sum_{v=1}^d r_{ikv}$
\EndIf
\If{\texttt{cov\_type} = \texttt{uni\_diag}} 
\State $\sigma_{kv}^2 \gets \sum_{i=1}^n \sum_{k=1}^N y_{ikv} / \sum_{i=1}^n \sum_{k=1}^N r_{ikv}$
\EndIf
\If{\texttt{cov\_type} = \texttt{uni\_id}} 
\State $\sigma_{kv}^2 \gets \sum_{i=1}^n \sum_{k=1}^N \sum_{v=1}^d y_{ikv} / \sum_{i=1}^n \sum_{k=1}^N \sum_{v=1}^d r_{ikv}$
\EndIf
\State $\Sigma_k \gets \operatorname{diag}(\sigma_{k1}, \dots, \sigma_{kd}) + \lambda I_d$ \Comment{Regularize covariance to ensure positivity}

\State \Return $\Sigma_1, \dots, \Sigma_N$
\end{algorithmic}
\end{algorithm}


\begin{algorithm}
\caption{Gaussian Mixture Model Selection}
\label{algo:modelselection}
\begin{algorithmic}[1]
\Require Candidate numbers of components $\mathcal N$, criterion (BIC or AIC), all input of Algorithm~\eqref{algo:base} except for number of components $N$
\Ensure Best Gaussian mixture model $\mathcal M_{\mathrm{best}}$

\State $s_{\mathrm{best}}\gets+\infty$
\State $\mathcal M_{\mathrm{best}}\gets\varnothing$

\ForAll{$N\in\mathcal N$}
    \State Fit a base MAR Gaussian Mixture with $N$ components \Comment{Algorithm \eqref{algo:base}}
    \State Obtain fitted model $\mathcal M_N$ and log-likelihood $\ell_N$
    \If{criterion = BIC}
        \State
        $
        s
        =
        p\log n
        -2\ell_N
        $
    \Else
        \State
        $
        s
        =
        2p
        -2\ell_N
        $
    \EndIf
    
    \If{$s<s_{\mathrm{best}}$}
        \State $s_{\mathrm{best}}\gets s$
        \State $\mathcal M_{\mathrm{best}}\gets\mathcal M_N$
    \EndIf
\EndFor

\State \Return $\mathcal M_{\mathrm{best}}$
\end{algorithmic}
\end{algorithm}

\begin{algorithm}
    \caption{Imputation with MAR Gaussian Mixture}
    \label{algo:impute}
    \begin{algorithmic}[1]
        \Require Estimated $\{\hat\pi_k, \hat\mu_k, \hat\Sigma_k\}_{k=1}^{\hat N}$ of $\mathcal{M}_{\mathrm{best}}$ resulted from Algorithm \eqref{algo:modelselection}
        \Ensure Imputed data $\tilde{X}$
        \State $\widetilde{\log\pi_{ik}} \gets \log\hat\pi_k + \sum_{v=1}^d (1 - m_{iv}) \log \phi\big(\frac{x_{iv - \mu_{kv}}}{\sigma_{kv}}\big)$ for $i=1,\dots,n$ and $k=1,\dots,\hat N$
        \State $\tilde\pi_{ik} \gets \exp\big(\widetilde{\log\pi_{ik}} - \log\sum_{k=1}^{\hat N}\exp \widetilde{\log\pi_{ik}}\big)$ for $i=1,\dots,n$ and $k=1,\dots,\hat N$
        \For{$i = 1, \dots, n$}
        \For{$v = 1, \dots, d$}
        \If{$m_{iv} = 0$}
        \State $\tilde{X}_{iv} \gets X_{iv}$
        \Else
        \State $j \gets \operatorname{Multinomial}(\hat N, \{\tilde\pi_{i1}, \dots, \tilde\pi_{i \hat N}\})$
        \State $\tilde{X}_{iv} \gets \mathcal{N} (\mu_{j v}, \sigma_{j, v}^2)$
        \EndIf
        \EndFor
        \EndFor
    \end{algorithmic}
\end{algorithm}

\paragraph{Imputation} We sample from $p_n (x^{(\bar m)} \mid x^{(m)})$ to impute the missing part per each observation.
\[
p_n(x_i^{(\bar m)} \mid x_i^{(m)}) = \sum_{k=1}^{\hat N} \Prob (Z_{i} = k \mid x^{(m)}) p_{\hat\Theta_k}(x_i^{\bar m} \mid x_i^{(m)}, Z_i = k) = \sum_{k=1}^{\hat N} \Prob (Z_{i} = k \mid x^{(m)}) p_{\hat\Theta_k}(x_i^{\bar m} \mid Z_i = k),
\]
where the last equality comes from the setting that all Gaussian components of $p_n$ have diagonal covariances. The posterior mixture measure of the $i$-th observation can be shown as
\[
\Prob (Z_i = k \mid x_i^{(m)}) = \frac{p_n(x_i^{m} \mid Z_i=k) \pi_k}{\sum_{j=1}^{\hat N} p_n(x_i^{m} \mid Z_i=j) \pi_j}.
\]

\begin{algorithm}
    \caption{Clustering with MAR Gaussian Mixture by KMeans on Estimated Centers}
    \label{algo:clustering}
    \begin{algorithmic}[1]
        \Require Estimated responsibilities $\{\hat\gamma_{ik}\}_{n \times \hat N}$, and centers of each components $\{\hat\mu_k\}_{k=1}^{\hat N} \in \mathbb{R}^d$ of the best model $\mathcal{M}_{\mathrm{best}}$ resulted from Algorithm \eqref{algo:modelselection}
        \Ensure A partition of the observed data $(X^{(M)}, M)$ into $m$ clusters on the condition that $m \le \hat N$
        \State $k_i \gets \arg\max_{1 \le k \le \hat N} \hat{\gamma}_{ik}$ for $i=1,\dots,n$
        \State Perform \texttt{kmeans} on the estimated centers $\{\hat\mu_k\}_{k=1}^{\hat N}$ to obtain a clustering of the centers. Let $f: \{ 1 \le k \le \hat N\} \mapsto \{1,\dots,m\}$ encode this clustering results. $f(k) = s$ implies the center of the $k$-th component is clustered into cluster $s$ by \texttt{kmeans}.
        \State The $i$-th observation $X_i^{(M_i)}$ is assigned to cluster with index $f(k_i)$.
    \end{algorithmic}
\end{algorithm}

\subsection{Missingness patterns for the Breast Cancer Dataset}\label{sec_AdditionalResults}

Table \ref{tab:cancer_patterns} shows the missingness patterns that were used for the breast cancer dataset. Since a fraction of observations was left complete ($c_0=0.1$), this yields an average cell missingness of about $20\%$.
\begin{table}[h]
    \caption{Missingness patterns and their frequencies for the artifical amputation of the breast cancer dataset. The pattern frequencies in column Freq1 are used in Section~\ref{subsec_clustering} and \ref{subsec_pca}, while those in column Freq2 are employed in Section~\ref{subsec_classification}.}
    \label{tab:cancer_patterns}
    \centering
    \begin{tabular}{lccc}
         \textbf{Missing features} & \textbf{Count} & \textbf{Freq1} & \textbf{Freq2} \\
         \hline
         Mean values of radius, perimeter, and area & 3 & 0.15 & 0.2\\
         Mean values of compactness, concavity, and concave points & 3 & 0.15 &0.2\\
         All worst values & 10 & 0.2 &0.2\\
         All standard deviations & 10 & 0.2 &0.2\\
         All values of radius, area, and perimeter & 9 & 0.2 &0.2\\
         Target label & 1 & 0.1 &0
    \end{tabular}
\end{table}


\section{Proofs} \label{sec_proofs}


Here we restate the results in the main text and give their proofs. In addition to the notation used in the main text, we define:
\begin{itemize}
    \item $\logm(y)= \max(-\log(y), 0)$, for all $ y > 0$. Moreover, we will sometimes write $\logm^2(y)$ to mean $(\logm(y))^2$ 
    \item For $x \in \R^d$, $\|x\|_{\infty}=\max_{j} |x_j| $ and $\|x \|$ the Euclidean norm of $x$ and similarly with $\|x^{(m)}\|_{\infty}$ and $\|x^{(m)}\|$.
    \item For $f: \R^k \to \R$ bounded, $\| f \|_{\infty}=\sup_{x} |f(x)|$.
    \item To avoid an explosion of constants, we sometimes follow \cite{Fundamentals} and write $a \lesssim b$ to mean that $a \leq C b$ for some constant $0 < C < \infty$. We note that in the proof of Theorem \ref{Thm:DensityResult}, $C$ might depend on the dimension $d$. We write $a \asymp b$, if $a \lesssim b$ and $b \lesssim a$.
    \item We denote marginal distributions with a superscript, such as $p_0^{(m)}(x^{(m)})$, but we will write conditional densities without superscript, e.g, $p_0(x^{\bar{m}}\mid x^{(m)})$.
    \item We will abbreviate $\Prob(M=m\mid X=x)$ with $\Prob(M=m\mid x)$ and similarly for $\Prob(M=m\mid X^{(m)}=x^{(m)})$.
\end{itemize}

\subsection{Proofs of Section \ref{sec_contractionresults}}

We start by proving the properties of $\HeMAR$ in Definition \ref{def:Ht}:

\propppseudometric*

\begin{proof}
Write $w_m(x^{(m)}):=\Prob(M=m\mid x^{(m)})\in[0,1]$ and let
$\mathcal{H}:=\bigoplus_m L^2(dx^{(m)})$ with $\lVert (h_m)_m\rVert_{\mathcal{H}}^2=\sum_m\int h_m^2\,dx^{(m)}$, a Hilbert space as a finite direct sum. Define $\Phi(f):=\bigl(\sqrt{f^{(m)}w_m}\bigr)_m$. For $f\ge0$ integrable, $\int f^{(m)}w_m\,dx^{(m)}\le\int f^{(m)}\,dx^{(m)}=\lVert f\rVert_1<\infty$, so $\Phi(f)\in\mathcal{H}$. Since $w_m\ge0$,
\[
  \bigl(\sqrt{f^{(m)}}-\sqrt{g^{(m)}}\bigr)^2w_m
  =\bigl(\sqrt{f^{(m)}w_m}-\sqrt{g^{(m)}w_m}\bigr)^2,
\]
and integrating in $x^{(m)}$ and summing over $m$ gives
\begin{equation}\label{eq:HtIsNorm}
  \HeMAR(f,g)=\bigl\lVert \Phi(f)-\Phi(g)\bigr\rVert_{\mathcal{H}}.
\end{equation}

\emph{Nonnegativity, symmetry, $\HeMAR(f,f)=0$, and the triangle inequality} are the corresponding properties of the norm $\lVert\cdot\rVert_{\mathcal{H}}$ under \eqref{eq:HtIsNorm}.

\emph{Order compatibility.} If $\ell\le f\le u$ pointwise on $\mathcal X$ then
$\ell^{(m)}\le f^{(m)}\le u^{(m)}$ by monotonicity of the integral, and likewise for $g$. Fix $(m,x^{(m)})$ and set $\varphi(s,t):=(\sqrt s-\sqrt t)^2$ on $[\ell^{(m)}(x^{(m)}),u^{(m)}(x^{(m)})]^2$. For fixed $t$, $\partial_s\varphi(s,t)=1-\sqrt{t/s}$, which is negative for $s<t$ and positive for $s>t$, so $\varphi(\cdot,t)$ attains its maximum over the interval at an endpoint; by symmetry of $\varphi$ the same holds in $t$. Hence the maximum over the square is attained at a corner, and since the diagonal corners give $0$,
\begin{equation}\label{eq:cornerbound}
  \bigl(\sqrt{f^{(m)}}-\sqrt{g^{(m)}}\bigr)^2
  \le\bigl(\sqrt{u^{(m)}}-\sqrt{\ell^{(m)}}\bigr)^2 .
\end{equation}
Multiplying \eqref{eq:cornerbound} by $w_m\ge0$, integrating and summing over $m$ gives $\HeMAR^2(f,g)\le\HeMAR^2(\ell,u)$.

\emph{Separation.} Assume $\HeMAR(f,g)=0$. By \eqref{eq:HtIsNorm},
$f^{(m)}w_m=g^{(m)}w_m$ almost everywhere for every $m$. Taking $m=\mathbf 0$,
where $x^{(\mathbf 0)}=x$, $f^{(\mathbf 0)}=f$ and $g^{(\mathbf 0)}=g$,
Assumption~\ref{asm_positive_MDM} gives $w_{\mathbf 0}(x) \ge c_0>0$ for almost every $x$, so $f=g$ almost everywhere.
\end{proof}


\begin{dfn}[The functional $T$]\label{def:T}
For a set of measurable functions $(g^{(m)})_{m}$ of the observed data $x^{(m)}$, set
\begin{equation}\label{eq:Tdef}
  T\left[(g^{(m)})_{m} \right]\;:=\;\sum_m\int g^{(m)}\bigl(x^{(m)}\bigr)\,\mathbb{P}\bigl(M=m\mid x^{(m)}\bigr)\,dx^{(m)} .
\end{equation}
\end{dfn}

We then understand
\begin{align*}
T[(f^{(m)})_{m}]\;=\;\sum_m\int f^{(m)}\bigl(x^{(m)}\bigr)\,\mathbb{P}\bigl(M=m\mid x^{(m)}\bigr)\,dx^{(m)}\\
T[(g(\rho_{p_0, p}^{(m)}))_m]\;=\;\sum_m\int g(\rho_{p_0, p}^{(m)}\bigl(x^{(m)}\bigr))\,\mathbb{P}\bigl(M=m\mid x^{(m)}\bigr)\,dx^{(m)}
\end{align*}

The functional $T$ has the following
properties, used repeatedly below.

\begin{enumerate}[label=\textup{(T\arabic*)},leftmargin=3em]
\item\label{T:poslin} \emph{Positivity and linearity.} $g^{(m)}\ge 0\  \forall m \Rightarrow T[(g^{(m)})_m]\ge0$ (as
  $\mathbb{P}(M=m\mid x^{(m)})\ge0$), and $T[(\alpha g_{1}^{(m)}+\beta g_{2}^{(m)})_m]=\alpha T[(g_{1}^{(m)})_m] +\beta T[g^{(m)}_{2}]$.
  Consequently $T$ preserves pointwise inequalities.
\item\label{T:hellinger} \emph{Adapted Hellinger as a $T$-image.} For any nonnegative integrable functions
  $f_1,f_2$, Definition~\ref{def:Ht} reads
  \begin{equation}\label{eq:HtT}
    \HeMAR^2(f_1,f_2)\;=\;T\!\Bigl[\left(\bigl(\sqrt{f_1^{(m)}}-\sqrt{f_2^{(m)}}\bigr)^2\right)_m\Bigr].
  \end{equation}
\item\label{T:expectation} \emph{Divergences as weighted $T$-images.} Inserting the
  observed-data density $p_0^{(m)}$ as a weight recovers the $\Pjoint$-expectations: for any $g$,
  \begin{equation}\label{eq:TP0}
\E_{\Pjoint}\bigl[g(\rho_{p_0, p}^{(M)}(X^{(M)}))\bigr]
    =T\bigl[(g(\rho_{p_0, p}^{(m)})\,p_0^{(m)})_m\bigr],
  \end{equation}
  since under Assumption~\ref{asm_true_MDM} the joint density of $(X,M)$ is
  $p_0(x)\,\mathbb{P}(M=m\mid x^{(m)})$ and the integrand depends on $(X,M)$ only through
  $(M,X^{(M)})$, so the missing block integrates $p_0$ to $p_0^{(m)}$.
  In particular
  \begin{equation}\label{eq:KLVT}
    \KLMAR(P_{0}\|P)=T\bigl[((\log\rho_{p_0, p}^{(m)})\,p_0^{(m)})_m\bigr],\qquad
    \Vt_k(P_{0}\|P)=T\bigl[(\big|\log\rho_{p_0, p}^{(m)}\big|^{k}\,p_0^{(m)})_m\bigr].
  \end{equation}
\item\label{T:commute} \emph{Midpoint commutes with marginalization (exactly).} Because
  $\int(\cdot)dx^{(\bar{m})}$ is linear, 
  \begin{equation}\label{eq:commute}
    \Bigl(\frac{p_0+p}{2}\Bigr)^{(m)}
    =\int\frac{p_0+p}{2}\,dx^{(\bar m)}
    =\frac{p_0^{(m)}+p^{(m)}}{2},
  \end{equation}
  an identity of functions \emph{before} any square root is taken; the marginalization never
  interacts with the root, so no Jensen slack arises when \eqref{eq:HtT} is applied to the
  pair $(p_0,\tfrac{p_0+p}{2})$.
\end{enumerate}

All observed-data divergences in this paper are images of a single positive linear functional.
Collecting them here lets the latter proofs proceed by ``apply $T$ to a pointwise inequality.''

Finally, for a measurable function $f$ with associated $(f^{(m)})_{m},\tilde{f}$ we use the \emph{Bernstein ``norm''}
\citep[p.~338]{vandervaart2023weak}
\begin{align}\label{eq:Bnorm}
  \|\tilde{f}\|_{\Pjoint,B}^2 \;
  &:=2\E_{\Pjoint}[e^{|f^{(M)}(X^{(M)})|} - 1 - |f^{(M)}(X^{(M)})|]\; \nonumber \\
  &=2\,T\bigl[\big((e^{|f^{(m)}|}-1-|f^{(m)}|)\,p_0^{(m)}\big)_m\bigr],
\end{align}
the second equality by \ref{T:expectation}. It is not a norm but satisfies the Riesz property
$|f|\le|g|\Rightarrow\|\tilde{f}\|_{\Pjoint,B}\le\|\tilde{g}\|_{\Pjoint,B}$, and $\|\tilde{f}\|_{\Pjoint, 2}\le\|\tilde{f}\|_{\Pjoint,B}$ since
$x^2\le2(e^{|x|}-1-|x|)$ \citep[p.~132]{vandervaart2023weak}.

\vspace*{0.5cm}
The bounds of \citet{kaji2026hellinger}, Theorems~1--2, are \emph{pointwise} inequalities in
the log-ratio $\log(p_0/p)$, established by real-analysis bounds that hold
for any pair of densities. We adapt that to our setting.

\begin{prop}[Observed-data KL/KV bounds]\label{lem:kaji2obs}
Suppose Assumptions~\ref{asm_true_MDM}--\ref{asm_positive_MDM} hold. 
For any $k>0$, define the tail
\begin{align*}
    \tilde{g}(x) := \sum_{m} \log \rho_{p_0, p}^{(m)}(x^{(m)})\;\one\{\rho_{p_0, p}^{(m)}(x^{(m)})>4\}\\
    \tilde{g}_k(x) := \sum_{m} (\log\rho_{p_0, p}^{(m)}(x^{(m)}))^k\;\one\{\rho_{p_0, p}^{(m)}(x^{(m)})>4\}.
\end{align*}
Then the following inequalities hold:
\begin{enumerate}[label=\textup{(\roman*)},leftmargin=2.4em]
\item $\displaystyle
  \tfrac13\,\Pjoint \tilde{g}
  \le \KLMAR(P_{0}\|P)
  \le 3\,\HeMAR^2(p_0,p) + \Pjoint \tilde{g};$
\item  for every real $k\ge2$,
  $\displaystyle
  \Pjoint\tilde{g}_k
  \le \Vt_k(P_{0}\|P)
  \le 4\bigl([2(\log4)^{k-2}]\vee(k/e)^k\bigr)\HeMAR^2(p_0,p)
     + \Pjoint\tilde{g}_k;$
\item for every real $k'\ge k>0$,
  $\displaystyle
  \Pjoint \tilde{g}_k
  \le 4\,\HeMAR^{\,2(1-\frac{k}{k'})}(p_0,p)\Bigl[\Pjoint\tilde{g}_{k'}\Bigr]^{\frac{k}{k'}}.$
\end{enumerate}
\end{prop}

\begin{proof}[Proof of Proposition~\ref{lem:kaji2obs}]
The argument has two parts: a pointwise step, in which the real-variable inequalities behind
Theorems~1--2 of \citet{kaji2026hellinger} are recorded for a single pattern, and a
$T$-step, in which the functional of Definition~\ref{def:T} is applied. 
Throughout, we abbreviate
$$a:=p_0^{(m)}(x^{(m)}), \  b:=p^{(m)}(x^{(m)}), \  u:=\rho^{(m)}_{p_0,p}(x^{(m)})=a/b, \ t := b/a = 1/u,$$ 
for arbitrary but fixed pattern $m$ and point $x^{(m)}$.
\smallskip
\paragraph{Part (i).}
\emph{Pointwise step, (i)}. 
We show first that at every $(m, x^{(m)})$ with $a > 0$,
\begin{align}
    \tfrac13\,(\log u)\,\one\{u>4\}\,a
  \;\le\;
  a\,(t - 1 - \log t)
  \;\le\;
  3\bigl(\sqrt a-\sqrt b\bigr)^2 + (\log u)\,\one\{u>4\}\,a.
  \label{eq:pwKD}
\end{align}
For the upper bound:
When $u \leq 4$, $\one\{u > 4\} = 0$, and 
$
a\,(t - 1 - \log t)
\;\le\;
3a(\sqrt{t} - 1)^2 = 3(\sqrt{a} - \sqrt{b})^2
= 3\bigl(\sqrt a-\sqrt b\bigr)^2 + (\log u)\,\one\{u>4\}\,a.
$
The first inequality uses $t - 1 - \log t \leq 3(\sqrt{t} - 1)^2$ for all $t \geq 1/4$.
When $u > 4$, $\one\{u > 4\} = 1$, and
$
a\,(t - 1 - \log t) 
\;\le\;
a \log(1/t) 
= a \log u \one\{u > 4\}
\;\le\;
3\bigl(\sqrt a-\sqrt b\bigr)^2 + (\log u)\,\one\{u>4\}\,a.
$
The first inequality follows from $t - 1 < 0$, the second follows from $3\bigl(\sqrt a-\sqrt b\bigr)^2 \geq 0$.
For the lower bound:
Similarly,
when $u \leq 4$, 
$
\tfrac13\,(\log u)\,\one\{u>4\}\,a = 0
  \;\le\;
  a\,(t - 1 - \log t),
$
as $t - 1 - \log t \geq 0$ for all $t > 0$.
When $u > 4$,
$
\tfrac13\,(\log u)\,\one\{u>4\}\,a 
=
\frac{1}{3} a ( - \log t)
<
a (t - 1 - \log t),
$
where the inequality follows from scalar inequality
$\displaystyle \log\frac1x \;<\; 3\bigl(x-1-\log x\bigr)$
for all $0<x<\tfrac14$ (\citet{kaji2026hellinger}, proof of {Theorem 2(i)}).
Combining the upper and lower bounds complete the first step for (i).

\smallskip
\emph{$T$-step, (i)}. Applying $T$ to \eqref{eq:pwKD} preserves both inequalities
(property \ref{T:poslin}).
For the middle part, we have
\begin{align*}
  T\bigl[(p_0^{(m)}\,(1/\rho^{(m)}_{p_0,p} - 1 - \log (1/\rho^{(m)}_{p_0,p})))_m\bigr]
  &= T\bigl[(p^{(m)}-p_0^{(m)})_m\bigr]
     + T\bigl[(p_0^{(m)}\log \rho^{(m)}_{p_0,p})_m\bigr] \\
  &= \bigl(T[(p^{(m)})_m]-T[(p_0^{(m)})_m]\bigr)
     \;+\; \KLMAR(P_0\|P) \\
  &= \KLMAR(P_0\|P),
\end{align*}
where $T\bigl[(\log\rho_{p_0, p}^{(m)}\,p_0^{(m)})_m\bigr] =\KLMAR(P_{0}\|P)$ by \eqref{eq:KLVT},
and $T[(p^{(m)})_m] = T[(p_0^{(m)})_m] = 1$ for any $p$.
Further, by \eqref{eq:TP0} and \eqref{eq:HtT}:
\begin{align*}
    T\bigl[(\log\rho_{p_0, p}^{(m)}\one\{\rho_{p_0, p}^{(m)}>4\}\,p_0^{(m)})_m\bigr] &=\Pjoint \tilde{g},\\
  \quad
  T\bigl[\left( 3 (\sqrt{p_0^{(m)}}-\sqrt{p^{(m)}})^2\right)_m\bigr] &=3\,\HeMAR^2(p_0,p).
\end{align*}
Substituting completes the proof for (i).

\smallskip
\paragraph{Part (ii).}
\emph{Pointwise step, (ii)}.
Fix $k\ge2$ and set $C^*_k:=4\bigl([2(\log4)^{k-2}]\vee(k/e)^k\bigr)$. We
show in this step that at every $(m,x^{(m)})$ with $a>0$,
\begin{equation}\label{eq:pwKV}
  (\log u)^k\,\one\{u>4\}\,a
  \;\le\;
  |\log u|^{k}\,a
  \;\le\;
  C^*_k\,(\sqrt a-\sqrt b)^2 + (\log u)^k\,\one\{u>4\}\,a.
\end{equation}

The lower bound is immediate. On $\{u>4\}$ we have $\log u>\log4>0$, so $(\log u)^k=|\log u|^k$ and the indicator only removes mass; on $\{u\le4\}$ the left side is $0\le|\log u|^k a$, trivially true.

To verify the upper bound we first show the scalar inequality
\begin{equation}\label{eq:scalarKV}
  |\log t|^{k}\;\le\;C^*_k\,(\sqrt t-1)^2, \qquad t\ge\tfrac14,\; k\ge2,
\end{equation}
with $C^*_k:=4\bigl([2(\log4)^{k-2}]\vee(k/e)^k\bigr)$.
On $\tfrac14\le t\le4$, 
$|\log t|^k\le(\log4)^{k-2}(\log t)^2\le 8(\log4)^{k-2}(\sqrt t-1)^2$,
where
the first inequality follows from $|\log t| \leq \log 4$ and the second holds if we can verify $|\log \sqrt{t}| \leq \sqrt{2} |\sqrt{t} - 1|$ for $\sqrt{t} \geq 1/2$. 
This is true, as for $\sqrt{t} \geq 1$, $|\log \sqrt{t}| = \log \sqrt{t} \leq \sqrt{t} - 1 \leq \sqrt{2}(\sqrt{t} - 1)$, for $1/2 \leq \sqrt{t} \leq 1$, the function $\log \sqrt{t} + \sqrt{2}(1 - \sqrt{t})$ increases and then decreases, so its minimum is attained at the endpoints, which are both nonnegative. Hence $-\log \sqrt{t} \leq \sqrt{2}(1 - \sqrt{t})$.
On $t>4$, 
$(\log t)^k\le  (k/e)^k t \le 4(k/e)^k(\sqrt t-1)^2$.
taking $C^*_k$ to be the larger constant and combining the two cases proves \eqref{eq:scalarKV}.

Using \eqref{eq:scalarKV} we prove the upper bound of \eqref{eq:pwKV}.
On $\{u\le4\}$, i.e.\ $t\ge\tfrac14$, the tail term vanishes and, by \eqref{eq:scalarKV}
and $(\sqrt t-1)^2a=(\sqrt a-\sqrt b)^2$,
\[
  |\log u|^k a=|\log t|^k a\;\le\;C^*_k(\sqrt t-1)^2a=C^*_k(\sqrt a-\sqrt b)^2 .
\]
On $\{u>4\}$, the tail term equals $(\log u)^ka=|\log u|^k a$, so \eqref{eq:pwKV} holds with the
nonnegative $C^*_k(\sqrt a-\sqrt b)^2$. This proves \eqref{eq:pwKV}.

\emph{T step, (ii).}
Apply $T(\cdot)$ to the collection of functions that appeared in \eqref{eq:pwKV} and from its property (T1):
\[
  %
  T\bigl[(|\log\rho^{(m)}_{p_0,p}|^{k}\,p^{(m)}_0)_m\bigr]
  =\Vt_k(P_0\|P),
\]
\[
  T\bigl[(C^*_k(\sqrt{ p^{(m)}_0}-\sqrt{ p^{(m)}_{\theta}})^2)_m\bigr]=C^*_k\,\HeMAR^2(p_0,p),
  \quad
  \text{and }
  \quad
  T\bigl[((\log \rho^{(m)}_{p_0,p})^k\one\{\rho^{(m)}_{p_0,p}>4\}\, p_0^{(m)})_m\bigr]=\Pjoint\tilde g_k .
\]
Substituting into \eqref{eq:pwKV} yields
\[
  \Pjoint\tilde g_k
  \;\le\;\Vt_k(P_0\|P)
  \;\le\;C^*_k\,\HeMAR^2(p_0,p)+\Pjoint\tilde g_k ,
\]
which completes part~(ii).

\paragraph{Part (iii).}
Assume $k'>k$. If $k'=k$ the claim reads $\Pjoint\tilde g_k\le4\,\Pjoint\tilde g_k$ (as $\HeMAR^0(p_0,p )=1$), which is trivially true. Set
\[
  \alpha:=\frac{k'-k}{k'}\in(0,1),\qquad r:=\frac{k'}{k}>1,\qquad s:=\frac{k'}{k'-k}>1,
  \qquad \tfrac1r+\tfrac1s=1 .
\]
The proof is one application of H\"older's inequality, as 
$T[(\,\cdot\,p^{(m)}_0)_m]$ is expectation under the probability measure $\Pjoint$ (property \ref{T:expectation}) and H\"older applies unchanged. 

First, on $\{u>4\}$, i.e.\ $t<\tfrac14$, one has $1-\sqrt t>\tfrac12$, so
$(1-\sqrt t)^{2\alpha}>4^{-\alpha}\ge\tfrac14$, and
\begin{equation}\label{eq:scalarL1}
  \one\{u>4\}\;\le\;4\,(1-\sqrt t)^{2\alpha}\,\one\{u>4\}.
\end{equation}
Then we multiply two sides of \eqref{eq:scalarL1} by the nonnegative integrand $(\log u)^k\one\{u>4\}\,a$
and apply $T$. Both steps preserve the inequality by \ref{T:poslin}, and the weighted $T$ reads as a $\Pjoint$-expectation \ref{T:expectation}:
\[
\begin{gathered}
  \Pjoint\tilde g_k \le 4\,\E_{\Pjoint}[A(X^{(M)})\,B(X^{(M)})], \quad \text{where} \\
  A :=(\log\rho_{p_0, p}^{(M)})^k\,\one\{\rho_{p_0, p}^{(M)}>4\}, 
  \qquad
  B :=\bigg(1-\sqrt{p^{(M)}/ p_0^{(M)}}\bigg)^{2\alpha}\,\one\{\rho_{p_0, p}^{(M)}>4\}.
\end{gathered}
\]
Both factors are nonnegative. On $\{\rho_{p_0, p}^{(M)} > 4\}$ one has $\log \rho_{p_0, p}^{(M)} >\log4>0$ and $1-\sqrt{p^{(M)}/ p_0^{(M)}}>0$, and
the indicator on $B$ keeps the noninteger power $(\cdot)^{2\alpha}$ well-defined. H\"older with
the conjugate pair $(r,s)$ gives $\E_{\Pjoint}[A(X^{(M)})B(X^{(M)})]\le(\E_{\Pjoint}A(X^{(M)})^{r})^{1/r}
(\E_{\Pjoint}B(X^{(M)})^{s})^{1/s}$, and since $kr=k'$, $2\alpha s=2$,
\begin{align*}
  \bigl(\E_{\Pjoint}\left[A(X^{(M)})^{r}\right]\bigr)^{1/r} &=\bigl[\Pjoint\tilde g_{k'}\bigr]^{k/k'}, \\
  \bigl(\E_{\Pjoint}\left[ B(X^{(M)})^{s} \right]\bigr)^{1/s}
  &= \Bigl(\E_{ \Pjoint}\!\bigl[(1-\sqrt{p^{(M)}(X^{(M)})/ p_0^{(M)}(X^{(M)})})^2\one\{\rho_{p_0, p}^{(M)}(X^{(M)}) > 4\}\bigr]\Bigr)^{1-k/k'} \\
  &\;\le\; 
  \Bigl(\E_{\Pjoint}\!\bigl[(1-\sqrt{p^{(M)}(X^{(M)})/ p_0^{(M)}(X^{(M)})})^2\bigr]\Bigr)^{1-k/k'}\\
  &\;\le\; 
  \Bigl(T\left[((\sqrt{p_0^{(m)}} - \sqrt{p^{(m)}})^2)_m\right]\Bigr)^{1-k/k'}\\
  &=
  \HeMAR^{\,2(1-k/k')}(p_0,p),
\end{align*}
where the first `$\leq$' drops the indicator $\one\{\rho_{p_0,p}^{(M)}>4\}$.
Combining gives
\[
  \Pjoint\tilde g_k
  \;\le\;4\,\HeMAR^{\,2(1-k/k')}(p_0,p)\,\bigl[\Pjoint\tilde g_{k'}\bigr]^{k/k'} ,
\]
which proves part~(iii).

\end{proof}
\begin{comment}
\emph{Parts} (ii) \emph{and} (iii). Both follow the same pointwise-then-$T$ template. For (ii),
the pointwise bound behind Theorem~2(ii) of \citet{kaji2026hellinger} is, for $u>4$ and $k\ge2$,
$\lvert\log u\rvert^{k}\le 4\bigl([2(\log4)^{k-2}]\vee(k/e)^{k}\bigr)(\sqrt u-1)^2$; multiplying
by $p_0^{(m)}$ and applying $T$, with $(\sqrt u-1)^2 p_0^{(m)}
=(\sqrt{p_0^{(m)}}-\sqrt{p^{(m)}})^2$ and
$\Vt_k=T[\lvert\log\rho^{(m)}_{p_0,p}\rvert^{k}p_0^{(m)}]$ \eqref{eq:KLVT}, gives (ii). The
restriction $k\ge2$ is that of \citet{kaji2026hellinger}; $1\le k<2$ would additionally require
control of $\{\rho_{p_0, p}\approx1\}$. Part (iii) is Theorem~2(iii) of \citet{kaji2026hellinger}, proved
by H\"older with exponents $k'/k,\,k'/(k'-k)$ on the $\Pjoint$-integral; since that proof uses only
$\Pjoint$-expectations of functions of $\rho_{p_0, p}$, it transfers under $T[\,\cdot\,p_0^{(m)}]$ \eqref{eq:TP0} and $\HeMAR$ for $h$ \eqref{eq:HtT}, with the
constant unchanged. 
\textcolor{red}{Didn't get the wording of this last part, what is $h$? Did I adapt this correctly?}
\textcolor{brown}{Meant to be that through the same scalar inequality argument + applying our T the result follows and instead of having their complete data Hellinger distance $h$ we have $\tilde{H}$.}
\end{comment}

\vspace*{0.5cm}
We will also need the observed-data form of the Bernstein-norm bound of
\citet{kaji2026hellinger}, Theorem~1. Following \citet{vandervaart2023weak, kaji2026hellinger},
denote $\m{p_0}{q}:=\log\frac{p_0
+q}{2p_0}$ with
\begin{align}
    \mm{p_0}{q}(x^{(m)}):=\log\frac{p_0^{(m)}(x^{(m)})
+q^{(m)}(x^{(m)})}{2p_0^{(m)}(x^{(m)})}
\end{align}
for the Birg\'e--Massart transform of a density $q$.

\begin{prop}[Observed-data Bernstein-norm bound]\label{prop:bn}
For any density $p$,
\[
  \Bigl\|\tildem{p_0}{p}
   \Bigr\|_{\Pjoint,B}
  \;\le\;\sqrt{18}\;\HeMAR\!\;\Bigl(p_0,\tfrac{p_0+p}{2}\Bigr).
\]
Further, for integrable nonnegative functions $q_1\le q_2$,
$\bigl\| \tildem{p_0}{q_2}-\tildem{p_0}{q_1}\bigr\|_{\Pjoint,B}^2\le
4\,\HeMAR^2\bigl(p_0+q_1,\,p_0+q_2\bigr)$.
\end{prop}

\begin{proof}
First, we show that at each $(m,x^{(m)})$ with $a :=p_0^{(m)}(x^{(m)})$, $b := p^{(m)}(x^{(m)})$,
\begin{align}
  e^{\lvert l\rvert}-1-\lvert l\rvert
  \le 9\bigl(\sqrt a-\sqrt{\tfrac{a+b}2}\bigr)^2/a,
  \qquad l=\log\tfrac{a+b}{2a}.
  \label{eq:pw1}
\end{align}
We will use the inequality 
\begin{align}
  x-1-\log x\;\le\;3\bigl(\sqrt x-1\bigr)^2 
  \quad \text{ for all $x\ge\tfrac14$}
  \label{eq:kajiineq}
\end{align}
that was used in \citet{kaji2026hellinger}, proof of Theorem~1(i).
Observe that if we set 
$r:=\frac{a+b}{2a}\ge\frac12 \; (\text{since } b\ge0)$, $e^{l}=r$, and $|l|=|\log r|$,
with $\bigl(\sqrt a-\sqrt{\tfrac{a+b}2}\bigr)^2/a
  =\Bigl(1-\sqrt{\tfrac{a+b}{2a}}\Bigr)^2
  =(1-\sqrt r)^2$,
\eqref{eq:pw1} can be equivalently written as
\[
  \psi(r):=e^{|\log r|}-1-|\log r|\;\le\;9(1-\sqrt r)^2,
  \qquad r\ge\tfrac12.
\]
When $r\ge1$, $|\log r|=\log r$ and $e^{|\log r|}=r$, so by \eqref{eq:kajiineq}, 
$\psi(r)=r-1-\log r\;\le\;3(\sqrt r-1)^2\;\le\;9(\sqrt r-1)^2$.
When $\tfrac12\le r<1$, $|\log r|=\log (1/r)$ and $e^{|\log r|}=\tfrac1r$, so
$
\psi(r)=(1/r)-1-\log (1/r)
\;\le\;
3(\sqrt{1/r}-1)^2
\;\le\;
9(\sqrt{1/r} - 1)^2 / (1/r) 
\;=\; 
9(1 - \sqrt{r})^2,
$ 
where the first `$\leq$' is by \eqref{eq:kajiineq} again with $1/r\in(1,2]\subset[\tfrac14,\infty)$ and the second `$\leq$' follows from $r \geq 1/2 > 1/3$. This completes the proof for \eqref{eq:pw1}.

Hence, from the definition of the Bernstein norm and multiplying \eqref{eq:pw1} by $a=p_0^{(m)}(x^{(m)})$ and applying $T$  yields
\begin{align*}
    \|\tildem{p_0}{p}\|_{\Pjoint,B}^2&=2\,T\left[\big((e^{|\mm{p_0}{p}|}-1-|\mm{p_0}{p}|)p_0^{(m)}\big)_m\right]\\
&\le 18\,T\left[\bigg(\big(\sqrt {p_0^{(m)}}-\sqrt{\frac{p_0^{(m)}+p^{(m)}}{2}}\big)^2 \bigg)_m\right]\\
&=18\,\HeMAR^2(p_0,\tfrac{p_0+p}{2}).
\end{align*}

Now we prove for the second part of statement. 
Fix $m$, $x^{(m)}$ with $a:=p^{(m)}_0(x^{(m)})>0$, $b_1:=q_1^{(m)}(x^{(m)})$ and $b_2:=q_2^{(m)}(x^{(m)})$.
As $q_1 \leq q_2$ and that marginalization is monotone, $q_1^{(m)}(x^{(m)}) \leq q_2^{(m)}(x^{(m)})$ and 
$
y := \mm{p_0}{q_2}(x^{(m)}) - \mm{p_0}{q_1}(x^{(m)}) = \log\frac{a+b_2}{2a}-\log\frac{a+b_1}{2a}
=\log\frac{a+b_2}{a+b_1}\;\ge\;0.
$
Apply the scalar inequality 
$
  e^{|y|}-1-|y|\;\le\;2\bigl(e^{y/2}-1\bigr)^2
  \;\text{for all $y\ge0$}
$
with 
$e^{y/2}=\sqrt{\frac{a+b_2}{a+b_1}},\;
(e^{y/2}-1)^2=\frac{\bigl(\sqrt{a+b_2}-\sqrt{a+b_1}\bigr)^2}{a+b_1}
$ gives
\[
  e^{|y|}-1-|y|\;\le\;2(e^{y/2}-1)^2
  =\frac{2\bigl(\sqrt{a+b_2}-\sqrt{a+b_1}\bigr)^2}{a+b_1}.
\]
Multiply by $a=p^{(m)}_0(x^{(m)})$ and use $\dfrac{a}{a+b_1}\le1$ \;(as $b_1\ge0$):
\begin{equation}\label{eq:pw2}
  \bigl(e^{|y|}-1-|y|\bigr)\,a
  \;\le\;\frac{2a}{a+b_1}\bigl(\sqrt{a+b_2}-\sqrt{a+b_1}\bigr)^2
  \;\le\;2\bigl(\sqrt{a+b_2}-\sqrt{a+b_1}\bigr)^2 .
\end{equation}

Finally, we apply operation $T$. The difference $\tildem{p_0}{q_2}-\tildem{p_0}{q_1}$ is glued from
the collection $(\mm{p_0}{q_2}-\mm{p_0}{q_1})_m$, so by \eqref{eq:Bnorm} and then
\eqref{eq:pw2} under $T$ (preserved by (T1)),
\begin{align*}
  \bigl\|\tildem{p_0}{q_2}-\tildem{p_0}{q_1}\bigr\|_{\Pjoint,B}^2
  &=2\,T\bigl[((e^{|\mm{p_0}{q_2} - \mm{p_0}{q_1}|}-1-|\mm{p_0}{q_2} - \mm{p_0}{q_1}|)\,p^{(m)}_0)_m\bigr]\\
  &\;\le\;4\,T\Bigl[\bigl((\sqrt{p_0^{(m)}+q_2^{(m)}}-\sqrt{p_0^{(m)}+q_1^{(m)}})^2\bigr)_m\Bigr].    
\end{align*}
By (T4), $p_0^{(m)}+q_i^{(m)}=(p_0+q_i)^{(m)}$, and by \eqref{eq:HtT} for the nonnegative integrable pair
$(p_0+q_1,\,p_0+q_2)$,
\[
  T\Bigl[\bigl((\sqrt{p_0^{(m)}+q_2^{(m)}}-\sqrt{p_0^{(m)}+q_1^{(m)}})^2\bigr)_m\Bigr]
  =\HeMAR^2\bigl(p_0+q_1,\,p_0+q_2\bigr),
\]
completing the proof. 
\end{proof}

\begin{lemma}[Bracketing entropy passes to the observed-data sieve]\label{lem:entropyinherit}
Under Assumptions~\textup{\ref{asm_true_MDM}--\ref{asm_positive_MDM}}, for every $\delta>0$,
\[
  \Jt_{[\,]}\!\bigl(\delta,\Pobsloc{\delta},\HeMAR\bigr)
  \;\le\;
  \Jt_{[\,]}\!\bigl(\delta,\Pnloc{\delta/\sqrt{c_0}},\He\bigr).
\]
\end{lemma}

\begin{proof}
The upper bound $\HeMAR\le \He$ of Proposition~\ref{prop:sandwich} is stated for densities; since the
bracket endpoints below are not densities, we record the elementary fact that it holds for
arbitrary nonnegative integrable pairs, then transfer covering numbers. (The lower bound
$c_0 \He^2 \le\HeMAR^2$ is applied only to the model densities $p_0,p$, so
Proposition~\ref{prop:sandwich} suffices there.)

\emph{Step 1 (contraction for non-densities).} For nonnegative integrable $f,g$, write
$f^{(m)},g^{(m)}$ for the observed-coordinate marginals. That is, $f^{(m)}(x^{(m)})=\int f(x^{(m)},x^{(\bar m)})\,dx^{(\bar m)}$. Denote $w_m(x^{(m)})=\mathbb{P}(M=m\mid x^{(m)})$. 
At each $(m,x^{(m)})$,
by Cauchy--Schwarz,
$\sqrt{f^{(m)}g^{(m)}} \geq \int \sqrt{fg} \, dx^{(\bar{m})}$, hence
\begin{equation}\label{eq:pwcontraction}
  \bigl(\sqrt{f^{(m)}}-\sqrt{g^{(m)}}\bigr)^2
  \;\le\;\int\bigl(\sqrt f-\sqrt g\bigr)^2\,dx^{(\bar m)} .
\end{equation}
Apply \eqref{eq:pwcontraction} within the definition of $\HeMAR^2(f,g)$ and using $\sum_m w_m(x^{(m)})=1$ together with
$dx=dx^{(m)}\,dx^{(\bar m)}$ gives
\begin{equation}\label{eq:contraction}
\begin{aligned}
  \HeMAR^2(f,g)
  &=\sum_m\int\bigl(\sqrt{f^{(m)}}-\sqrt{g^{(m)}}\bigr)^2\,w_m(x^{(m)})\,dx^{(m)}\\
  &\le\sum_m\int\!\Bigl(\int\bigl(\sqrt f-\sqrt g\bigr)^2\,dx^{(\bar m)}\Bigr)\,w_m(x^{(m)})\,dx^{(m)}\\
  &=\int\bigl(\sqrt f-\sqrt g\bigr)^2\Bigl(\sum_m w_m(x^{(m)})\Bigr)\,dx
  =\int\bigl(\sqrt f-\sqrt g\bigr)^2\,dx=\He^2(f,g).
\end{aligned}
\end{equation}
This extends the upper bound of Proposition~\ref{prop:sandwich} beyond densities. The
construction of \citet[Prop.~B.1]{Bayes}, which passes through joint densities
$q_{\theta_j}$ on $\mathcal X\times\{0,1\}^d$ and a normalization $\sum_m\int q_{\theta_j}=1$,
requires $f,g$ to be densities and so does not apply to bracket endpoints directly.
Moreover, $\HeMAR$ is a pseudometric (Proposition~\ref{prop:pseudometric}) compatible with the pointwise partial order, so $N_{[\,]}(\cdot,\cdot,\HeMAR)$ is well-defined \citet[][Appendix~C]{Fundamentals}.

\emph{Step 2 (brackets transfer).} The two localizing balls differ, since $\HeMAR\le \He$ makes
$\Pobsloc{\delta}$ the larger set; we control it through the matching $\He$-ball. By the lower
bound of Proposition~\ref{prop:sandwich}, $\He(p_0,p)\le\HeMAR(p_0,
p)/\sqrt{c_0}$, so
\begin{equation}\label{eq:ballincl}
 \Pobsloc{\delta}= \bigl\{p\in \Pnsieve:\HeMAR(p_0,p)\le\delta\bigr\}
  \;\subseteq\;
  \bigl\{p \in \Pnsieve: \He(p_0,p)\le\delta/\sqrt{c_0} \bigr\}=\Pnloc{\delta/\sqrt{c_0}}.
\end{equation}
Let $\{[\ell_j,u_j]\}_{j\le N}$, $N=N_{[\,]}(\varepsilon,\Pnloc{\delta/\sqrt{c_0}},\He)$, be an
$\varepsilon$-bracketing of $\Pnloc{\delta/\sqrt{c_0}}$ in $\He$, with $\ell_j\le u_j$ not
necessarily densities \citep[][p.~528]{Fundamentals}. As \eqref{eq:contraction} gives
$\HeMAR(\ell_j,u_j)\le \He(\ell_j,u_j)<\varepsilon$, with \eqref{eq:ballincl},
the collection $\{[\ell_j,u_j]\}_{j\le N}$ is an $\varepsilon$-bracketing of
$\Pobsloc{\delta}$ in $\HeMAR$, hence
$N_{[\,]}(\varepsilon,\Pobsloc{\delta},\HeMAR)\le N_{[\,]}(\varepsilon,\Pnloc{\delta/\sqrt{c_0}},\He)$.


\emph{Step 3 (integrate).} $N_{[\,]}(\varepsilon,\Pobsloc{\delta},\HeMAR)\le
N_{[\,]}(\varepsilon,\Pnloc{\delta/\sqrt{c_0}},\He)$ from Step 2 holds for every $\varepsilon > 0$. Since $x \to \sqrt{1 + \log x}$ is increasing:
$\sqrt{1+\log N(\varepsilon,\Pobsloc{\delta},\HeMAR)}\;\le\;\sqrt{1+\log N(\varepsilon,\Pnloc{\delta/\sqrt{c_0}},\He)}$, and integrating over $\varepsilon \in (0, \delta)$ gives the desired result.
\end{proof}

In particular, a complete-data entropy bound
$\Jt_{[\,]}(\delta_n,\Pnloc{\delta_n/\sqrt{c_0}},\He)\le\delta_n^2\sqrt n$ implies the corresponding
observed-data bound $\Jt_{[\,]}(\delta_n,\Pobsloc{\delta_n},\HeMAR)\le\delta_n^2\sqrt n$, so the
local entropy condition of \citet{vandervaart2023weak} Theorem 3.4.1 may be posed and verified on the complete-data sieve.

\vspace{0.5cm}
\begin{lemma}\label{lem:shift}
Let $f,g,s$ be nonnegative integrable functions on $\mathcal X$. Then
\[
  \HeMAR^2(f+s,\;g+s)\;\le\;\HeMAR^2(f,g).
\]
In particular, taking $s=p_0$,
$\HeMAR(p_0+q_1,\;p_0+q_2)\le\HeMAR(q_1,q_2)$
for any nonnegative integrable $q_1,q_2$.
\end{lemma}
 
\begin{proof}
\emph{Scalar step.} We first show that for all reals $\alpha,\beta,\gamma\ge 0$,
\begin{align}
  \bigl(\sqrt{\alpha+\gamma}-\sqrt{\beta+\gamma}\bigr)^2
  \;\le\;
  \bigl(\sqrt{\alpha}-\sqrt{\beta}\bigr)^2 .
  \label{eq:ineqshift}
\end{align}
If $\alpha=\beta$ both sides vanish. Otherwise, 
\[
  \bigl|\sqrt{\alpha+\gamma}-\sqrt{\beta+\gamma}\bigr|
  =\frac{|\alpha-\beta|}{\sqrt{\alpha+\gamma}+\sqrt{\beta+\gamma}}
  \;\le\;
  \frac{|\alpha-\beta|}{\sqrt{\alpha}+\sqrt{\beta}}
  =\bigl|\sqrt{\alpha}-\sqrt{\beta}\bigr|,
\]
the inequality holding since
$\gamma\ge 0$ implies $\sqrt{\alpha+\gamma}+\sqrt{\beta+\gamma}\ge\sqrt{\alpha}+\sqrt{\beta}>0$,
where $\sqrt{\alpha}+\sqrt{\beta}>0$ as $\alpha\neq\beta\Rightarrow\alpha+\beta>0$.
Squaring gives \eqref{eq:ineqshift}.

\smallskip
\emph{Marginalization.} For each pattern $m$, marginalization over the missing block 
\[
  (f+s)^{(m)}(x^{(m)})
  =\int\bigl(f+s\bigr)(x^{(m)},x^{(\bar m)})\,dx^{(\bar m)}
  =f^{(m)}(x^{(m)})+s^{(m)}(x^{(m)}),
\]
and likewise $(g+s)^{(m)}=g^{(m)}+s^{(m)}$, with 
$f^{(m)},g^{(m)},s^{(m)} \geq 0$.
Now we fix $(m,x^{(m)})$ and apply \eqref{eq:ineqshift} with
$
  \alpha=f^{(m)}(x^{(m)}),\;
  \beta=g^{(m)}(x^{(m)}),\;
  \gamma=s^{(m)}(x^{(m)})\;\ge 0
$
to obtain
\[
  \Bigl(\sqrt{(f+s)^{(m)}}-\sqrt{(g+s)^{(m)}}\Bigr)^2
  \;\le\;
  \Bigl(\sqrt{f^{(m)}}-\sqrt{g^{(m)}}\Bigr)^2 .
\]
Multiplying by the nonnegative weight $\mathbb{P}\!\left(M=m \mid x^{(m)}\right)$ and
integrating over $x^{(m)}$, then summing over $m$, preserves the inequality:
\begin{align*}
      \HeMAR^2(f+s,g+s)
  = &\sum_m\int\Bigl(\sqrt{(f+s)^{(m)}}-\sqrt{(g+s)^{(m)}}\Bigr)^2
     P\!\left(M=m \mid x^{(m)}\right)\,dx^{(m)} \\
  \;\le\;
  &\sum_m\int\Bigl(\sqrt{f^{(m)}}-\sqrt{g^{(m)}}\Bigr)^2
     \mathbb{P}\!\left(M=m \mid x^{(m)}\right)\,dx^{(m)} \\
  = &\He^2(f,g).
\end{align*}
The specialization $s=p_0$ is immediate.
\end{proof}

\mainthmMAR*

\begin{proof}

We apply \citet[Theorem~3.4.1]{vandervaart2023weak},
with the dictionary
\[
  \theta=p,\quad \theta_{n,0}=p_0,\quad
  \theta_n=p_n,\quad \Theta_n=\Pnsieve,\quad
  d_n(\theta,\theta_{n,0})=\HeMAR(p_0,p),
\]
\[
  \mathbb M_n(\theta)=(4+2\sqrt2)^2\,\Pn \tildem{p_0}{p},\qquad
  M_n(\theta)=(4+2\sqrt2)^2\,\Pjoint \tildem{p_0}{p},
\]
where
\[
\tildem{p_0}{p}(x, m):= \sum_{m'} \log\frac{p_0^{(m')}(x^{(m')})+p^{(m')}(x^{(m')})}{2\,p_0^{(m')}(x^{(m')})} \one\{ m = m'\}
\]
and $\underline\delta_n=\HeMAR(p_0,p_n)$. 
The criterion functions $\mathbb M_n$ and $M_n$ are the empirical and population
averages of the measurable map
$(x,m)\mapsto\tildem{p_0}{p}(x,m)$ on
$\mathcal X\times\{0,1\}^d$. The anchoring is at the truth,
$\theta_{n,0}=p_0$ (\emph{not} $p_n$), which permits the
condition \eqref{eq:tail} in place of uniform boundedness. Throughout, $T$ is
the positive linear functional of Definition~\ref{def:T} with properties
\ref{T:poslin}--\ref{T:commute}, and all stochastic-order symbols $O_P,o_P$ are
with respect to the sample law $\Pjointn$. We verify the three requirements of \citet[Theorem~3.4.1]{vandervaart2023weak} in turn.

\paragraph{Step 0: the midpoint inequality.}
For any $p_0,p$,
\begin{equation}\label{eq:eight}
  \bigl(1-\tfrac1{\sqrt2}\bigr)^2\HeMAR^2(p_0,p)
  \;\le\;
  \HeMAR^2\!\Bigl(p_0,\tfrac{p_0+p}{2}\Bigr)
  \;\le\;
  \tfrac12\,\HeMAR^2(p_0,p).
\end{equation}

\emph{Proof of \eqref{eq:eight}.} Kaji's (8) is the pointwise statement, valid at every
$(a,b)$ with $a,b\ge0$,
\[
  \underbrace{\bigl(1-\tfrac1{\sqrt2}\bigr)^2\bigl(\sqrt a-\sqrt b\bigr)^2}_{=:L}
  \;\le\;
  \underbrace{\Bigl(\sqrt a-\sqrt{\tfrac{a+b}2}\Bigr)^2}_{=:C}
  \;\le\;
  \underbrace{\tfrac12\bigl(\sqrt a-\sqrt b\bigr)^2}_{=:R}.
\]
Evaluate at $(a,b)=\bigl(p_0^{(m)}(x^{(m)}),\,p^{(m)}(x^{(m)})\bigr)$,
making $L=L^{(m)},C=C^{(m)},R=R^{(m)}$ functions of $(m,x^{(m)})$. Applying $T$ to the pointwise chain
$L^{(m)}\le C^{(m)}\le R^{(m)}$, positivity gives $T[(C^{(m)}-L^{(m)})_m]\ge0$ and $T[(R^{(m)}-C^{(m)})_m]\ge0$, and linearity then yields
$T[(L^{(m)})_m]\le T[(C^{(m)})_m]\le T[(R^{(m)})_m]$. Evaluate each term using \eqref{eq:HtT}, \eqref{eq:commute}, and the
linearity of $T$:
\begin{align*}
  T[(L^{(m)})_m]&=\bigl(1-\tfrac1{\sqrt2}\bigr)^2\,
        T\;\!\Bigl[\big(\bigl(\sqrt{p_0^{(m)}}-\sqrt{p^{(m)}}\bigr)^2 \big)_m\Bigr]
       =\bigl(1-\tfrac1{\sqrt2}\bigr)^2\,\HeMAR^2(p_0,p),\\[2pt]
  T\;[(R^{(m)})_m]&=\tfrac12\,
        T\;\!\Bigl[\big(\bigl(\sqrt{p_0^{(m)}}-\sqrt{p^{(m)}}\bigr)^2\big)_m\Bigr]
       =\tfrac12\,\HeMAR^2(p_0,p),\\[2pt]
  T\;[(C^{(m)})_m]&=T\;\!\Bigl[\big(\bigl(\sqrt{p_0^{(m)}}-\sqrt{\tfrac{p_0^{(m)}+p^{(m)}}{2}}\bigr)^2\big)_m\Bigr] 
       {=} \; \HeMAR^2\!\Bigl(p_0,\tfrac{p_0+p}{2}\Bigr),
\end{align*}
Substituting into $T[(L^{(m)})_m]\le T[(C^{(m)})_m]\le T[(R^{(m)})_m]$ gives \eqref{eq:eight}.

\paragraph{Requirement (i): curvature / negative drift.}
We must show, for every $n$ and $\delta>\underline\delta_n$,
\[
  \sup_{\,\delta/2<\HeMAR(p_0,p)\le\delta}
  (4+2\sqrt2)^2\bigl(\Pjoint \tildem{p_0}{p}-\Pjoint \tildem{p_0}{p_0}\bigr)
  \;\le\;-\delta^2.
\]
Here $\tildem{p_0}{p_0}=\log1=0$, and 
\[
\Pjoint \tildem{p_0}{p}=-\E_{ \Pjoint}\left[ \log\frac{2p_0^{(M)}(X^{(M)})}{p_0^{(M)}(X^{(M)})+p^{(M)}(X^{(M)})}\right]=
-\KLMAR\bigl(P_{0}\,\|\,\tfrac{P_{0}+P}{2}\bigr)
\]
by
\eqref{eq:KLVT}. The unadapted $\He^2\le\KL$ is classical \citep[Lemma~B.1]{Fundamentals}. We now show that the adapted KL dominates the adapted squared Hellinger,
\begin{equation}\label{eq:KLgeH}
  \HeMAR^2\!\Bigl(p_0,\tfrac{p_0+p}{2}\Bigr)
  \;\le\;
  \KLMAR\bigl(P_{0}\,\|\,\tfrac{P_{0}+P}{2}\bigr)
\end{equation}
using the propensity-weighted affinity identity \eqref{eq:HrhoT}, which we
verify here. We prove \eqref{eq:KLgeH} for any pair of densities $p_1,p_2$ via the adapted
affinity
$\widetilde\rho(p_1,p_2):=\sum_m\int\sqrt{p_1^{(m)}p_2^{(m)}}\,\Prob(M{=}m\mid x^{(m)})\,dx^{(m)}
=T[(\sqrt{p_1^{(m)}p_2^{(m)}})_m]$ \citep[App.~B]{Bayes}. Each propensity-weighted
marginal is a probability, $T[(p_i^{(m)})_m]=\sum_m\int p_i^{(m)}(x^{(m)})\,\Prob(M{=}m\mid x^{(m)})
\,dx^{(m)}=\sum_m P_i(M=m)=1$,
so by \eqref{eq:HtT}
\begin{equation}\label{eq:HrhoT}
 \HeMAR^2(p_1,p_2)=T[(p_1^{(m)})_m]+T[(p_2^{(m)})_m]-2\,T[(\sqrt{p_1^{(m)}p_2^{(m)}})_m]
  =2\bigl(1-\widetilde\rho(p_1,p_2)\bigr).
\end{equation}
On the other hand, applying $-\log y\ge1-y$ at $y=\sqrt{p_2^{(m)}/p_1^{(m)}}$ inside the KL
integrand gives, pointwise, $p_1^{(m)}\log\tfrac{p_1^{(m)}}{p_2^{(m)}}\ge2\bigl(p_1^{(m)}
-\sqrt{p_1^{(m)}p_2^{(m)}}\bigr)$. Applying $T$ \textup{(}with the weight $p_1^{(m)}$ already present\textup{)} via \eqref{eq:KLVT},
\[
  \KLMAR(P_1\|P_2)=T\bigl[((\log\tfrac{p_1^{(m)}}{p_2^{(m)}})\,p_1^{(m)})_m\bigr]
  \ge 2\bigl(T[(p_1^{(m)})_m]-T[(\sqrt{p_1^{(m)}p_2^{(m)}})_m]\bigr)
  =2\bigl(1-\widetilde\rho(p_1,p_2)\bigr).
\]
Comparing with \eqref{eq:HrhoT} yields $\KLMAR(P_1\|P_2)\ge\HeMAR^2(p_1,p_2)$. Specializing to
$(P_1,P_2)=(P_{0},\tfrac{P_{0}+P}{2})$ gives \eqref{eq:KLgeH}.
Combining \eqref{eq:KLgeH} with \eqref{eq:eight},
\[
  \Pjoint \tildem{p_0}{p}
  \le -\,\HeMAR^2\!\Bigl(p_0,\tfrac{p_0+p}{2}\Bigr)
  \le -\bigl(1-\tfrac1{\sqrt2}\bigr)^2\HeMAR^2(p_0,p)
  \le -\bigl(1-\tfrac1{\sqrt2}\bigr)^2\tfrac{\delta^2}{4}
  = -\frac{\delta^2}{(4+2\sqrt2)^2}
\]
on $\{\HeMAR(p_0,p)>\delta/2\}$, using $(1-\tfrac1{\sqrt2})^2/4=(4+2\sqrt2)^{-2}$.
Multiplying by $(4+2\sqrt2)^2$ gives $\le-\delta^2$. \emph{No likelihood-ratio bound enters}.
This step uses only \eqref{eq:eight} and \eqref{eq:KLgeH},  pure $\HeMAR$-geometry through $T$.

\paragraph{Requirement (ii): modulus of continuity.}

We must bound 
\[
\mathbb E^*\sup_{\HeMAR(p_0,p)\le\delta}\sqrt n\bigl|(\Pn-\Pjoint)
\tildem{p_0}{p}\bigr|
\]
by a majorant $\phi_n(\delta)$ with $\phi_n(\delta)/\delta^\alpha$
decreasing for some $\alpha<2$. Here the outer expectation $\mathbb E^*$ is defined as in \cite[Chapter 1.2]{vandervaart2023weak}. We follow Kaji via the Bernstein-norm maximal inequality
\citet[Theorem~2.14.18$'$]{vandervaart2023weak}, applied to the function class
\[
\widetilde{\mathcal M}_{n,\delta}:=\{\tildem{p_0}{p}:
p \in \densityspacemodel_n, \;
\HeMAR(p_0,p)\le\delta\}.
\]
This necessitates an envelope and a bracketing
integral. \emph{This is where the observed/complete bridge sits, in two layers.}

\emph{(ii-a) Envelope.} By Proposition~\ref{prop:bn} and \eqref{eq:eight},
\[
  \bigl\|\tildem{p_0}{p}\bigr\|_{\Pjoint,B}
  \le\sqrt{18}\,\HeMAR\;\!\Bigl(p_0,\tfrac{p_0+p}{2}\Bigr)
  \le\sqrt{18}\cdot\tfrac1{\sqrt2}\,\HeMAR(p_0,p)\le 3\delta
\]
for $\HeMAR(p_0,p)\le\delta$, the second inequality by the upper half of
\eqref{eq:eight}. Thus $\widetilde{\mathcal M}_{n,\delta}$ has envelope of Bernstein ``norm'' bounded by $3\delta$.

\emph{(ii-b) Bracket transform (Bernstein $\to$ complete-data Hellinger).} The class
$\widetilde{\mathcal M}_{n,\delta}$
is the image of the $\HeMAR$-ball $\Pobsloc{\delta}$ under the transform
$p\mapsto \tildem{p_0}{p}$, and by \eqref{eq:ballincl}
$\Pobsloc{\delta}\subseteq \Pnloc{\delta/\sqrt{c_0}}$; hence it
suffices to bracket the transforms of $p\in\Pnloc{\delta/\sqrt{c_0}}$.
Let $[q_1,q_2]$, with $q_1\le q_2$ not necessarily densities, be an $\varepsilon$-bracket of
$\Pnloc{\delta/\sqrt{c_0}}$ in $\He$. By
monotonicity of marginalization (\S5) any $p$ in $[q_1,q_2]$ has
$\tildem{p_0}{p}\in[\tildem{p_0}{q_1},\tildem{p_0}{q_2}]$ (the
transform $q\mapsto \tildem{p_0}{q}$ is increasing in $q$), and the second bound of
Proposition~\ref{prop:bn} gives
\[
  \bigl\|\tildem{p_0}{q_1}-\tildem{p_0}{q_2}\bigr\|_{\Pjoint,B}^2
  \le 4\,\HeMAR^2\bigl(p_0+q_1,\,p_0+q_2\bigr)
  \le 4\,\HeMAR^2(q_1,q_2)
  \le 4\,\He^2(q_1,q_2)<4\varepsilon^2,
\]
the second inequality holds from Lemma~\ref{lem:shift}, and the third
by the contraction \eqref{eq:contraction}. Hence each complete-data $\varepsilon$-bracket
yields a $2\varepsilon$-bracket of $\widetilde{\mathcal M}_{n,\delta}$ in $\|\cdot\|_{\Pjoint,B}$,
so
\[
  N_{[\,]}\bigl(2\varepsilon,\widetilde{\mathcal M}_{n,\delta},\|\cdot\|_{\Pjoint,B}\bigr)
  \le N_{[\,]}\bigl(\varepsilon,\Pnloc{\delta/\sqrt{c_0}},\He\bigr).
\]
Therefore, writing $\widetilde{\mathcal M}_{n,\delta}$ for the Bernstein-norm class with its
$3\delta$ envelope from (ii-a), the change of variables $\varepsilon\mapsto2\varepsilon$ and
concavity of $J_{[\,]}$ give
\begin{align*}
    &J_{[\,]}\bigl(3\delta,\widetilde{\mathcal M}_{n,\delta},\|\cdot\|_{\Pjoint,B}\bigr)
  =\;\int_0^{3\delta}\!\!\sqrt{1+\log N_{[\,]}(\varepsilon,\widetilde{\mathcal M}_{n,\delta},\|\cdot\|_{\Pjoint,B})}\,d\varepsilon \\
  \le \; &2\!\int_0^{3\delta/2}\!\!\sqrt{1+\log N_{[\,]}(\varepsilon,\Pnloc{\delta/\sqrt{c_0}},\He)}\,d\varepsilon \\
  = \;&2\,J_{[\,]}\bigl(\tfrac32\delta,\Pnloc{\delta/\sqrt{c_0}},\He\bigr) \\
  \le \;&3\,J_{[\,]}\Big(\delta,\Pnloc{\delta/\sqrt{c_0}},\He\Big),
\end{align*}
the first inequality substituting $N_{[\,]}(2\varepsilon,\widetilde{\mathcal M}_{n,\delta},
\|\cdot\|_{\Pjoint,B})\le N_{[\,]}(\varepsilon,\Pnloc{\delta/\sqrt{c_0}},\He)$ and rescaling
$\varepsilon\mapsto2\varepsilon$ (Jacobian $2$, upper limit $3\delta\mapsto\tfrac32\delta$); the
middle equality is the definition of $J_{[\,]}$ at upper limit $\tfrac32\delta$; 
the last inequality since, for the fixed class $\Pnloc{\delta/\sqrt{c_0}}$, the integrand
$\varepsilon\mapsto\sqrt{1+\log N_{[\,]}(\varepsilon,\Pnloc{\delta/\sqrt{c_0}},\He)}$ is nonincreasing, hence
$t\mapsto J_{[\,]}(t,\Pnloc{\delta/\sqrt{c_0}},\He)/t$ is nonincreasing, giving
$J_{[\,]}(c\delta,\Pnloc{\delta/\sqrt{c_0}},\He)\le c\,J_{[\,]}(\delta,\Pnloc{\delta/\sqrt{c_0}},\He)$
for $c\ge 1$ (here $c=\tfrac32$).

Now combine (ii-a) and (ii-b): 
\begin{equation}\label{eq:modPsi}
\begin{aligned}
    &\omega_n(\delta) := \mathbb{E}^{*}\!\!\sup_{\HeMAR(p_0,p)\le\delta}\!
      \sqrt{n}\,\bigl|(\Pn-\Pjoint)\,\widetilde{\ell}_{p_0,p}\bigr| \\
  \;\lesssim\;
  &J_{[\,]}\!\bigl(3\delta,\widetilde{\mathcal{M}}_{n,\delta},\|\cdot\|_{\Pjoint,B}\bigr)
  \left[\,1+
    \frac{J_{[\,]}\!\bigl(3\delta,\widetilde{\mathcal{M}}_{n,\delta},\|\cdot\|_{P_0,B}\bigr)}{9\delta^{2}\sqrt{n}}
  \right] \\
  \;\lesssim\;
  &\,J_{[\,]}\bigl(\delta,\Pnloc{\delta/\sqrt{c_0}},\He\bigr)
  \left[\,1+
    \frac{\,J_{[\,]}\bigl(\delta,\Pnloc{\delta/\sqrt{c_0}},\He\bigr)}{3\delta^{2}\sqrt{n}}
  \right]
  =:\Psi_n(\delta) \quad \text{for all $\delta > 0$}.
\end{aligned}
\end{equation}
The first `$\lesssim$' follows from (ii-a) and \citet[Theorem~2.14.18$'$]{vandervaart2023weak}; the second by (ii-b) and that the map $f(x) = x(1 + x/(9 \delta^2 \sqrt{n}))$ is nondecreasing on $x \in [0, \infty)$.

We now use above to provide a function $\phi_n(\delta)$ and show it satisfies the requirements. 
Since the right side of \eqref{eq:entropy} is nondecreasing in $\alpha_0$ for $\delta\ge\delta_n$,
we may assume $\alpha_0\ge1$. Set
\[
  \phi_n(\delta) := \sqrt{n}\,\delta_n^2\Bigl(\frac{\delta\vee\delta_n}{\delta_n}\Bigr)^{\alpha_0},
  \qquad \delta>0.
\]

We first claim $J_{[\,]}(\delta,\Pnloc{\delta/\sqrt{c_0}},\He)\le K\phi_n(\delta)$ for all $\delta\ge\delta_n$, as:
On $[\delta_n,\delta_n^\star]$ this is \eqref{eq:entropy}.
On $(\delta_n^\star,\infty)$, $\delta/\sqrt{c_0}>\delta_n^\star/\sqrt{c_0}\ge\sqrt2\ge\He(p_0,p)$ for every
density $p$, so $\Pnloc{\delta/\sqrt{c_0}}=\Pnsieve$. 
By \eqref{eq:entropy} at $\delta_n^\star$, and $\alpha_0\ge1$,
\[
  J_{[\,]}(\delta,\Pnsieve,\He)
  \le\frac{\delta}{\delta_n^\star}\,J_{[\,]}(\delta_n^\star,\Pnsieve,\He)
  \le\frac{\delta}{\delta_n^\star}\,K\sqrt n\,\delta_n^2\Bigl(\frac{\delta_n^\star}{\delta_n}\Bigr)^{\alpha_0}
  \le K\sqrt n\,\delta_n^2\Bigl(\frac{\delta}{\delta_n}\Bigr)^{\alpha_0}
  =K\phi_n(\delta).
\]

We then claim $\omega_n(\delta)\lesssim \,\phi_n(\delta)$ for all $\delta>\underline\delta_n$ (with constant $CK(1+\tfrac K3)$ where $C$ denote the absolute constant implicit in \eqref{eq:modPsi}):
On $[\delta_n,\infty)$: since $\phi_n(\delta)/(\delta^2\sqrt n)=(\delta_n/\delta)^{2-\alpha_0}\le1$, \eqref{eq:modPsi} gives
\[
  \omega_n(\delta)\le C\,\Psi_n(\delta)
  \le CK\phi_n(\delta)\Bigl[1+\frac{K\phi_n(\delta)}{3\delta^2\sqrt n}\Bigr]
  \le CK\Bigl(1+\frac K3\Bigr)\phi_n(\delta).
\]
On $(\underline\delta_n,\delta_n)$: the map $\delta\mapsto\omega_n(\delta)$ is nondecreasing and $\phi_n(\delta)$ is a constant function, hence
\[
  \omega_n(\delta)\le\omega_n(\delta_n)\le CK\Bigl(1+\frac K3\Bigr)\phi_n(\delta_n)
  =CK\Bigl(1+\frac K3\Bigr)\phi_n(\delta).
\]
Finally, $\phi_n$ is nondecreasing, and $\phi_n(\delta)/\delta^{\alpha_0} = \sqrt n\,\delta_n^{2-\alpha_0}$ on $[\delta_n,\infty)$ and $\sqrt n\,\delta_n^2\,\delta^{-\alpha_0}$ on $(0,\delta_n)$, so it is nonincreasing with $\alpha_0<2$, completing Requirement~(ii).

\paragraph{Requirement (iii): the sequence $\delta_n$.}
By \eqref{eq:approx}, $\underline\delta_n\le C\,\delta_n$ for some constant $C$. We verify three conditions on $\delta_n':=L\,\delta_n$, where
$L:=\max\{1,\,2(\sqrt6+\sqrt3)\,C\}$.

\smallskip
\noindent\textbf{(iii-1)} $\phi_n(\delta_n')\le\sqrt n\,\delta_n'^2$. Since $L\ge1$ and $\alpha_0<2$,
\[
  \phi_n(\delta_n')=\sqrt n\,\delta_n^2\,L^{\alpha_0}\le\sqrt n\,\delta_n^2\,L^{2}=\sqrt n\,\delta_n'^2 .
\]

\smallskip
\noindent\textbf{(iii-2)} $\delta_n'^2\ge M_n(\theta_{n,0})-M_n(\theta_n)$. With
$\tildem{p_0}{p_0}=\log1=0$,
\[
  M_n(\theta_{n,0})-M_n(\theta_n)
  =(4+2\sqrt2)^2\bigl(\Pjoint \tildem{p_0}{p_0}-\Pjoint \tildem{p_0}{p_n}\bigr),
\]
and by the upper half of Prop.~\ref{lem:kaji2obs}(i) on the pair
$(p_0,\tfrac{p_0+p_n}{2})$ together with \eqref{eq:eight},
\[
  \Pjoint \tildem{p_0}{p_0}-\Pjoint \tildem{p_0}{p_n}
  =\KLMAR\left( P_0 \; || \; \frac{P_0 + P_{n}}{2}  \right)
  \le 3\,\HeMAR^2\!\Bigl(p_0,\tfrac{p_0+p_n}{2}\Bigr)
  \le \tfrac32\,\HeMAR^2(p_0,p_n)
  =\tfrac32\,\underline\delta_n^2 .
\]
Here the second inequality follows from Prop.~\ref{lem:kaji2obs}(i) since for all $m$,
$$\one\left \{\frac{2p_0^{(m)}}{p_0^{(m)}+p_n^{(m)}}  >4 \right\}=0.$$
Hence 
$
  M_n(\theta_{n,0})-M_n(\theta_n)
  \leq (4+2\sqrt2)^2\tfrac32 \underline \delta_n^2
  = \bigl(2(\sqrt6+\sqrt3)\,\underline\delta_n\bigr)^2
  \le\bigl(2(\sqrt6+\sqrt3)\,C\,\delta_n\bigr)^2
  \le\delta_n'^2 .
$

\smallskip
\noindent\textbf{(iii-3)} $\delta_n'\ge\underline\delta_n$. Immediate as
$
  \delta_n'\ge2(\sqrt6+\sqrt3)\,C\,\delta_n
  \ge
  C\,\delta_n
  \ge
  \underline\delta_n .
$

\paragraph{The approximate-maximizer step.}
It remains to verify $\mathbb M_n(\hat\theta_n)\ge\mathbb M_n(\theta_n)-O_P(\delta_n^2)$, i.e.\
$\Pn \tildem{p_0}{p_{\hat\theta_n}}\ge\Pn \tildem{p_0}{p_n}-O_P(\delta_n^2)$. 
We note that as $L\ge1$ introduced in the previous step is a constant, $O_P(\delta_n')=O_P(\delta_n)$.
By concavity of
$\log$ \textup{(}identical to \citeauthor{kaji2026hellinger}\textup{)},
\[
  2\Pn \tildem{p_0}{p_{\hat\theta_n}}
  \ge \frac{1}{n} \sum_{i=1}^{n} \log\frac{\hat{p}_n^{(M_i)}(X_i^{(M_i)})}{p_0^{(M_i)}(X_i^{(M_i)})} 
  =\frac{1}{n} \sum_{i=1}^{n} \log\frac{p_n^{(M_i)}(X_i^{(M_i)})}{p_0^{(M_i)}(X_i^{(M_i)})}  + \frac{1}{n} \sum_{i=1}^{n} \log\frac{\hat{p}_n^{(M_i)}(X_i^{(M_i)})}{p_n^{(M_i)}(X_i^{(M_i)})}
\]
the first step using $\log\tfrac{a+b}{2}\ge\tfrac12(\log a+\log b)$ and the fact that $\log\tfrac{p_0^{(m)}(x^{(m)})}{p_0^{(m)}(x^{(m)})}=0$ for all $x$ and $m$. The last term is $\ge-O_P(\delta_n^2)$ by the
approximate-maximizer hypothesis \eqref{eq:approxmax}. Writing $\rho^{(m)}_{p_0, p_n}(x)=p_0^{(m)}(x)
/p_n^{(m)}(x)$ as above, and 
$$\frac{1}{n} \sum_{i=1}^{n} \log\frac{p_n^{(M_i)}(X_i^{(M_i)})}{p_0^{(M_i)}(X_i^{(M_i)})}
=-(\Pn-\Pjoint)\log\tilde{\rho}_{p_0, p_n}-\Pjoint \log \tilde{\rho}_{p_0, p_n},$$
it suffices 
to show
\begin{align}\label{eq:center}
&-\tfrac{1}{2}(\Pn-\Pjoint)\log\tilde{\rho}_{p_0, p_n}-\tfrac{1}{2}\Pjoint \log \tilde{\rho}_{p_0, p_n} - \Pn \tildem{p_0}{p_n}
  \;\ge\;-O_P(\delta_n^2) \nonumber \\
  \iff &-\tfrac{1}{2}(\Pn-\Pjoint)\log\tilde{\rho}_{p_0, p_n}-\tfrac{1}{2}\Pjoint \log \tilde{\rho}_{p_0, p_n} \\
  &- (\Pn \tildem{p_0}{p_n} 
  -\Pjoint \tildem{p_0}{p_n}) - \Pjoint \tildem{p_0}{p_n}
  \;\ge\;-O_P(\delta_n^2). \nonumber
\end{align}

\emph{Nonrandom terms.} As above, $-\Pjoint \log\tilde{\rho}_{p_0, p_n}=-\KLMAR(P_{0}\|P_{n})$
and $\Pjoint \tildem{p_0}{p_n}=-\KLMAR(P_{0}\|\tfrac{P_{0}+P_{n}}{2})$.
Proposition~\ref{lem:kaji2obs} (i) gives $\KLMAR(P_{0}\|P_{n})\le3\HeMAR^2(p_0,
p_n)+\Pjoint[\log\tilde{\rho}_{p_0, p_n}\one\{\tilde{\rho}_{p_0, p_n}>4\}]$, and Proposition~\ref{lem:kaji2obs} (iii) with
$k=1,k'=2$ with \eqref{eq:tail} bounds the exposed tail,
\begin{align*}
    &\Pjoint\big[\log\tilde{\rho}_{p_0, p_n}\one\{\tilde{\rho}_{p_0, p_n}>4\}\big] \\
  \le\; &4\,\HeMAR\bigl(p_0,p_n\bigr)
     \bigl(  \E_{\Pjoint}\!\left[\Bigl(\log\rho_{ p_0, p_n}^{(M)}(X^{(M)})\Bigr)^2\,\one\{\rho_{ p_0, p_n}^{(M)}(X^{(M)})>4\}\right]\bigr)^{1/2} \\
  \le\; &4\,\delta_n\sqrt{M_{\mathrm{MAR}}}\,\delta_n
  = 4\sqrt{M_{\mathrm{MAR}}}\,\delta_n^2,
\end{align*}
using \eqref{eq:approx} and \eqref{eq:tail}. Hence $\KLMAR(P_{0}\|P_{n})\le(3+4\sqrt{M_{\mathrm{MAR}}})
\delta_n^2$, and the nonrandom part of \eqref{eq:center} is
\begin{align*}
    &-\tfrac12 \Pjoint \log\tilde{\rho}_{p_0, p_n}-\Pjoint \tildem{p_0}{p_n}
  =-\tfrac12\KLMAR(P_{0}\|P_{n})
   +\KLMAR\bigl(P_{0}\|\tfrac{P_{0}+P_{n}}{2}\bigr) \\
  \ge &-\tfrac12\bigl(3+4\sqrt{M_{\mathrm{MAR}}}\bigr)\delta_n^2
  \;\ge\; -O_P(\delta_n^2),
\end{align*}
the first inequality applying the bound $\KLMAR(P_{0}\|P_{n})\le(3+4\sqrt{M_{\mathrm{MAR}}})\delta_n^2$
just obtained and dropping $\KLMAR\bigl(P_{0}\|\tfrac{P_{0}+P_{n}}{2}\bigr)$, which is nonnegative since $\KLMAR\ge0$
\textup{(}Prop.~\ref{informationloss_KL}\textup{)}; the last by \eqref{eq:approx}. The tail constant
$M_{\mathrm{MAR}}$ enters here through $3+4\sqrt{M_{\mathrm{MAR}}}$ (and, below, through
$8+M_{\mathrm{MAR}}$ in the variance).

\emph{Random terms.} We control each of the two empirical-process terms in \eqref{eq:center}. 
For $-\tfrac12(\Pn-\Pjoint)\log\tilde{\rho}_{p_0, p_n}$: 
\begin{align*}
    &\Var\!\bigl((\Pn-\Pjoint)\log\tilde\rho_{p_0,p_n}\bigr)
  =\frac1n\Var_{\Pjoint}\!\bigl(\log\tilde\rho_{p_0,p_n}\bigr) \\
  \;\le\;
  &\frac{1}{n}\Pjoint\bigl(\log\tilde\rho_{p_0,p_n}\bigr)^2 
  =
  \frac{1}{n}\Vt_2(P_{0}\|P_{n})\\
  \;\le\;
  &\frac{1}{n} \big( 8\,\HeMAR^2(p_0,p_n)+\Pjoint\tilde g_2\big)
  \;\le\;\frac{8+M_{\mathrm{MAR}}}{n} \;\delta_n^2.
\end{align*}
The equality by i.i.d. data; the first `$\leq$' by that centered second moment upper bounded by raw second moment; the second `$\leq$' by
Proposition~\ref{lem:kaji2obs} (ii) with $k=2$ giving the constant
$C_2^{*}=4\bigl[2(\log4)^0\vee(2/e)^2\bigr]=8$; and the last `$\leq$' by
assumption \eqref{eq:approx} and \eqref{eq:tail}.
So by Chebyshev $(\Pn-\Pjoint)\log\tilde\rho_{p_0,p_n}=O_P(\delta_n/\sqrt n)=O_P(\delta_n^2\vee n^{-1})$.

For $-(\Pn - \Pjoint)\tildem{p_0}{p_n}$, the centered transform at the
fixed point $\theta_n$:
\begin{align*}
  &\Var\!\bigl((\Pn-\Pjoint)\tildem{p_0}{p_n}\bigr)
  =\tfrac1n\Var_{\Pjoint}\!\bigl(\tildem{p_0}{p_n}\bigr)
  \le\tfrac1n\bigl\|\tildem{p_0}{p_n}\bigr\|_{\Pjoint,2}^2
  \le\tfrac1n\bigl\|\tildem{p_0}{p_n}\bigr\|_{\Pjoint,B}^2,
\end{align*}
where the last `$\leq$' follows from the pointwise inequality $x^2\le2(e^{|x|}-1-|x|)$. From Proposition~\ref{prop:bn}, \eqref{eq:eight}, and assumption \eqref{eq:approx}
\[
  \bigl\|\tildem{p_0}{p_n}\bigr\|_{\Pjoint,B}
  \le\sqrt{18}\,\HeMAR\;\!\Bigl(p_0,\tfrac{p_0+p_n}{2}\Bigr)
  \le \;\sqrt{18}\cdot\tfrac{1}{\sqrt2}\,\HeMAR(p_0,p_n)
  \le \;3\delta_n.
\]
Hence
$\Var\bigl((\Pn-\Pjoint)\tildem{p_0}{p_n}\bigr)\le 9\delta_n^2/n$ and
$
  (\Pn - \Pjoint)\tildem{p_0}{p_n}
  =O_P(\delta_n^2\vee n^{-1}).
$
Both are $O_P(\delta_n^2\vee n^{-1})$. 
Finally, collecting the nonrandom and random bounds gives \eqref{eq:center}, hence $\mathbb M_n(\hat\theta_n)\ge\mathbb M_n(\theta_n)-O_P(\delta_n^2)$.

\end{proof}

We next derive two auxiliary lemmas that will help bound the bracketing number of the sieve application in Section \ref{sec_densityestimation}. 
We recall that each $p_{F,\Sigma}\in\densityspacemodel_n$ is indexed by
$\theta=(w,z,\Sigma)\in\Theta_n
:=\Delta_{\N}\times[-a,a]^{d\N}\times\mathcal D_{\sigma}$. We
metrize the space by
\begin{equation*}
\varrho(\theta,\theta')
   :=\lVert w-w' \rVert_{1}
     +\max_{1\le j\le \N}\lVert z_{j}-z_{j}' \rVert_{\infty}
     +\lVert \Sigma-\Sigma' \rVert_{\mathrm{op}},
\end{equation*}
where $\max_{1\le j\le \N}\lVert z_{j}-z_{j}'\rVert_{\infty} = \max_{j,k}|z_{j,k} - z^{'}_{j,k}|$. 

\begin{lemma}[Lipschitz in parameter, with a Gaussian envelope]\label{lem:lip}
Define the function $G_{n}:\mathbb{R}^{d}\to[0,\infty)$, as
\[
  G_{n}(x):=\sup_{\substack{\Sigma\in\mathcal D_{\sigma}\\
     \lVert z \rVert_{\infty}\le a}}
     \bigl(1+\lVert \Sigma^{-1/2}(x-z) \rVert_{2}^{2}\bigr)\phi(x-z;\Sigma).
\]
Assume $\sigma\le1$ and $a\ge\sqrt\pi\,\sigma^{-1}$,
then there are constants $C_{\mathrm L}=C_{\mathrm L}(d)$ and $C_{G}=C_{G}(d)$, such that
$\int_{\mathbb{R}^{d}}G_{n}\le C_{G}\,\sigma^{-4}(a/\sigma)^{d}$, and, for all
$\theta,\theta'\in\Theta_n$ and all $x\in\mathbb{R}^{d}$,
\begin{align}\label{eq:lipbound}
  \lvert p_{F,\Sigma}(x)-p_{F',\Sigma'}(x) \rvert
   \le C_{\mathrm L}\,\sigma^{-(d+2)}\,G_{n}(x)\,\varrho(\theta,\theta'). 
\end{align}
\end{lemma}

\begin{proof}
Write $p_{\theta}=p_{F,\Sigma}$ and bound the three coordinate directions
of~\eqref{eq:paramnorm} separately.

\emph{Weights.} For $(z,\Sigma)=(z',\Sigma')$,
\[
  \lvert p_{\theta}(x)-p_{\theta'}(x) \rvert
   =\Bigl\lvert\sum_{j}(w_{j}-w_{j}')\phi(x-z_{j};\Sigma)\Bigr\rvert
   \le\Bigl(\max_{j}\phi(x-z_{j};\Sigma)\Bigr)\lVert w-w' \rVert_{1}
   \le G_{n}(x)\lVert w-w' \rVert_{1}.
\]

\emph{Locations.} We have
$\nabla_{z_{j}}\phi(x-z_{j};\Sigma)=\Sigma^{-1}(x-z_{j})\phi(x-z_{j};\Sigma)$, and for any $v$,
$\lVert \Sigma^{-1}v \rVert_{2}\le\lVert \Sigma^{-1/2} \rVert_{\mathrm{op}}\lVert \Sigma^{-1/2}v \rVert_{2}
\le\sigma^{-1}\lVert \Sigma^{-1/2}v \rVert_{2}$ since $\mathrm{eig}_{1}(\Sigma)\ge\sigma^{2}$. 
With $s\le1+s^{2}$,
\[
  \bigl\lVert \nabla_{z}\phi(x-z;\Sigma) \bigr\rVert_{2}
  \le\sigma^{-1}\bigl(1+\lVert \Sigma^{-1/2}(x-z) \rVert_{2}^{2}\bigr)\phi(x-z;\Sigma)
  \le\sigma^{-1}G_{n}(x),\qquad \lVert z \rVert_{\infty}\le a .
\]
The segment from $z_{j}$ to $z_{j}'$ stays in $[-a,a]^{d}$, 
so the mean value theorem and the Cauchy--Schwarz inequality give
$\lvert \phi(x-z_{j};\Sigma)-\phi(x-z_{j}';\Sigma) \rvert\le\sigma^{-1}G_{n}(x)\lVert z_{j}-z_{j}' \rVert_{2}$.
Summing over the mixture and using $\lVert z_{j}-z_{j}' \rVert_{2}\le\sqrt d\,\lVert z_{j}-z_{j}' \rVert_{\infty}$,
$\lvert p_{\theta}(x)-p_{\theta'}(x) \rvert\le\sqrt d\,\sigma^{-1}G_{n}(x)\max_{j}\lVert z_{j}-z_{j}' \rVert_{\infty}$.

\emph{Covariance.} For $s\in[0,1]$, $\Sigma_{s}=(1-s)\Sigma+s\Sigma'\in\mathcal
D_{\sigma}$,
\[
  \partial_{s}\phi(u;\Sigma_{s})
   =\tfrac12\phi(u;\Sigma_{s})\,
     \mathrm{tr}\!\bigl[\Sigma_{s}^{-1}(\Sigma'-\Sigma)
       (\Sigma_{s}^{-1}uu^{\top}-I)\bigr],
\]
so that, with $u=x-z$ and $\lVert \Sigma_{s}^{-1} \rVert_{\mathrm{op}}\le\sigma^{-2}$,
\[
  \lvert \partial_{s}\phi(x-z;\Sigma_{s}) \rvert
   \le\tfrac d2\sigma^{-2}\bigl(1+\lVert \Sigma_{s}^{-1/2}(x-z) \rVert_{2}^{2}\bigr)
      \phi(x-z;\Sigma_{s})\,\lVert \Sigma-\Sigma' \rVert_{\mathrm{op}}
   \le\tfrac d2\sigma^{-2}G_{n}(x)\,\lVert \Sigma-\Sigma' \rVert_{\mathrm{op}}.
\]
Integrating in $s$ and summing over the mixture gives the covariance bound.

Using the triangle inequality to bound $\lvert p_{F,\Sigma}(x)-p_{F',\Sigma'}(x) \rvert$ by the difference in the three contributions, and thus by the sum of the three bounds and further majorizing each prefactor by $\sigma^{-(d+2)}$
(valid for $\sigma<1$) yields~\eqref{eq:lipbound}.

\emph{Envelope integrability.}
We first bound each integrand of $G_{n}$ by a function of $\lVert x-z \rVert_{2}$, removing $\Sigma$: Since
$\mathrm{eig}_{1}(\Sigma)\ge\sigma^{2}$ and $\mathrm{eig}_{d}(\Sigma)\le \sigma^{-2}$ for $\Sigma\in\mathcal D_{\sigma}$,
we have $\det\Sigma\ge\sigma^{2d}$ and
\[
  \mathrm{eig}_{d}(\Sigma)^{-1}\lVert x-z \rVert_{2}^{2}
  \;\le\;\lVert \Sigma^{-1/2}(x-z) \rVert_{2}^{2}
  \;\le\;\mathrm{eig}_{1}(\Sigma)^{-1}\lVert x-z \rVert_{2}^{2},
\]
hence
$\sigma^{2}\lVert x-z \rVert_{2}^{2}
\le\lVert \Sigma^{-1/2}(x-z) \rVert_{2}^{2}
\le\sigma^{-2}\lVert x-z \rVert_{2}^{2}$.
Writing $q(t):=(1+\sigma^{-2}t^{2})\,e^{-\sigma^2t^{2}/2}$ for $t\ge0$,
these bounds combine into
\begin{align*}
  \bigl(1+\lVert \Sigma^{-1/2}(x-z) \rVert_{2}^{2}\bigr)\phi(x-z;\Sigma)
   &=\bigl(1+\lVert \Sigma^{-1/2}(x-z) \rVert_{2}^{2}\bigr)
     \frac{e^{-\frac12\lVert \Sigma^{-1/2}(x-z) \rVert_{2}^{2}}}
          {(2\pi)^{d/2}\sqrt{\det\Sigma}} \\
   &\le(2\pi)^{-d/2}\sigma^{-d}
     \bigl(1+\sigma^{-2}\lVert x-z \rVert_{2}^{2}\bigr)
     e^{-\sigma^2\lVert x-z \rVert_{2}^{2}/2} \\
   &=(2\pi)^{-d/2}\sigma^{-d}\,q \;\!\bigl(\lVert x-z \rVert_{2}\bigr),
\end{align*}
where the inequality applies $\det\Sigma\ge\sigma^{2d}$ to the denominator,
$\lVert \Sigma^{-1/2}(x-z) \rVert_{2}^{2}\le\sigma^{-2}\lVert x-z \rVert_{2}^{2}$
to the quadratic weight, and
$\lVert \Sigma^{-1/2}(x-z) \rVert_{2}^{2}\ge \sigma^{2}\lVert x-z \rVert_{2}^{2}$
to the exponent, the latter being available because $\mathcal D_{\sigma}$ caps the largest eigenvalue.

We further bound $q$ by a function that is nonincreasing in $t$:
Write $u:=\sigma^{-2}t^{2}$, so that
$\sigma^2 t^{2}/4 = \sigma^4 u/4$. Since
$\sup_{u\ge0}(1+u)e^{-\sigma^4 u/4}=\frac{4}{\sigma^4} e^{-1+\sigma^4/4}\le\frac{4}{\sigma^4}$, with the supremum attained at $u = 4\sigma^{-4} - 1 \geq 0$ as $\sigma \leq 1$
\begin{equation}\label{eq:qmajor}
  q(t)
   =(1+\sigma^{-2}t^{2})e^{-\sigma^2 t^{2}/4}\cdot e^{-\sigma^2t^{2}/4}
   \;\le\;4 \sigma^{-4}\,e^{-\sigma^2 t^{2}/4},
   \qquad t\ge 0.
\end{equation}
Next we take supremum over $z$ in the integrand bound:
Let
$\mathrm{proj}(x):=\arg\min_{\lVert z \rVert_{\infty}\le a}\lVert x-z \rVert_{2}$
denote the Euclidean projection of $x$ onto $[-a,a]^{d}$, and
$\mathrm{dist}(x):=\lVert x-\mathrm{proj}(x) \rVert_{2}$. Since the map
$z\mapsto q(\lVert x-z \rVert_{2})$ is bounded over
$\lVert z \rVert_{\infty}\le a$ by its nonincreasing majorant \eqref{eq:qmajor},
which is largest at the smallest available
$\lVert x-z \rVert_{2}=\mathrm{dist}(x)$,
\[
  G_{n}(x)
   \;\le\;(2\pi)^{-d/2}\sigma^{-d}\,
     \sup_{\lVert z \rVert_{\infty}\le a} q\bigl(\lVert x-z \rVert_{2}\bigr)
   \;\le\;4 (2\pi)^{-d/2}\sigma^{-(d+4)}\,e^{-\sigma^2\mathrm{dist}(x)^{2}/4}.
\]

Finally, we integrate this bound. Since $[-a,a]^{d}$ is a product of intervals and $\mathrm{proj}$ acts
coordinatewise, $\mathrm{dist}(x)^{2}=\sum_{i=1}^{d}(\lvert x_{i}\rvert-a)_{+}^{2}$, so
\[
  \int_{\mathbb{R}^{d}}e^{-\sigma^{2}\mathrm{dist}(x)^{2}/4}\,\mathrm dx
  =\prod_{i=1}^{d}\int_{\mathbb{R}}e^{-\sigma^{2}(\lvert x_{i}\rvert-a)_{+}^{2}/4}\,\mathrm dx_{i}
  =\bigl(2a+2\sqrt\pi\,\sigma^{-1}\bigr)^{d}
  \le(4a)^{d},
\]
using $\int_{0}^{\infty}e^{-\sigma^{2}t^{2}/4}\,\mathrm dt=\sqrt\pi\,\sigma^{-1}$ and
$a\ge\sqrt\pi\,\sigma^{-1}$. Therefore
\[
  \int_{\mathbb{R}^{d}}G_{n}(x)\,\mathrm dx
  \;\le\;4(2\pi)^{-d/2}\sigma^{-(d+4)}(4a)^{d}
  \;=\;C_{G}\,\sigma^{-4}(a/\sigma)^{d},
\]
with $C_{G}:=4^{d+1}(2\pi)^{-d/2}\ge1$ depending only on $d$.
\end{proof}

The bound \eqref{eq:lipbound} can then be used to bound the bracketing number of $\Pnsieve$:


\begin{restatable}[Covering of $\Theta_n$ $\Rightarrow$ bracketing of $\mathcal{P}_n$]{lemma}{bracketingresult}\label{lem:cov2br}
For every $\eta\in(0,1)$,
\begin{equation}\label{eq:cov2br}
  \log{N_{[\,]}}(\eta,\Pnsieve,\He)
   \;\le\;\log{N}\!\;\Bigl(\tfrac{\eta^{2}}{8C_{\mathrm L}C_{G}(a/\sigma)^{d}\sigma^{-(d+6)}},
      \Theta_n,\varrho\Bigr),
\end{equation}
for constants $C_{\mathrm L}=C_{\mathrm L}(d)\ge1$ and $C_{G}=C_{G}(d)\ge1$ as in Lemma~\ref{lem:lip}.
Consequently there is a constant $A_{1}$ depending only on $d$ such that,
for all $\eta\in(0,1)$,
\begin{equation}\label{eq:brentropy}
  \log{N_{[\,]}}(\eta,\Pnsieve,\He)
   \;\le\;A_{1}\bigl(\N+d^{2}\bigr)\Bigl[\log\tfrac a\sigma
        +\log\tfrac1\eta+\log\tfrac1\sigma\Bigr].
\end{equation}
\end{restatable}

\begin{proof}
The parameter set
$\Theta_n=\Delta_{\N}\times[-a,a]^{d\N}\times\mathcal D_{\sigma}$
is a bounded subset of a finite--dimensional Euclidean space, hence totally
bounded. Consequently, there exists a finite $r$-cover for any $r > 0$, i.e., ${N}(r,\Theta_n,\varrho)<\infty$, and centers
$\theta^{(1)},\dots,\theta^{(K)}\in\Theta_n$ with $K={N}(r,\Theta_n,\varrho)$
whose closed $r$--balls cover $\Theta_n$:
\begin{equation}\label{eq:rcover}
  \forall\,\theta\in\Theta_n\ \ \exists\,k\in\{1,\dots,K\}:\quad
  \varrho\bigl(\theta,\theta^{(k)}\bigr)\le r.
\end{equation}
Now we construct brackets from this covering. For $\eta \in (0,1)$ arbitrary, set  
\[
r:=\eta^{2}/(8C_{\mathrm L}C_{G}(a/\sigma)^{d}\sigma^{-(d+6)})
\qquad  
c_{n}:=C_{\mathrm L}\sigma^{-(d+2)}r.
\] For each $\theta^{(k)}$, construct
\[
  \ell^{(k)}:=\bigl(p_{\theta^{(k)}}-c_{n}G_{n}\bigr)_{+},
  \qquad
  u^{(k)}:=p_{\theta^{(k)}}+c_{n}G_{n},
\]
where $G_n$ is from Lemma~\ref{lem:lip}.
This gives a valid bracket, as
given any $p_{\theta}\in\Pnsieve$, which corresponds to some $\theta\in\Theta_n$, 
there exists some $k$ such that $\varrho\bigl(\theta,\theta^{(k)}\bigr)\le r$ which forces
$\lvert p_{\theta}-p_{\theta^{(k)}} \rvert\le C_{\mathrm L}\sigma^{-(d+2)}r\,G_{n}
=c_{n}G_{n}$ pointwise by Lemma~\ref{lem:lip}. Hence $p_{\theta} \in [\ell^{(k)}, u^{(k)}]$.

Using $(\sqrt u-\sqrt\ell)^{2}\le\lvert u-\ell \rvert$ for $u\ge\ell\ge0$, and
writing $\bar C_{G}:=C_{G}\,\sigma^{-4}(a/\sigma)^{d}$ for the envelope bound of
Lemma~\ref{lem:lip},
\[
  \He^{2}(\ell^{(k)},u^{(k)})
   =\int(\sqrt{u^{(k)}}-\sqrt{\ell^{(k)}})^{2}
   \le\int\lvert u^{(k)}-\ell^{(k)} \rvert
   \le\int 2c_{n}G_{n}
   =2c_{n}\!\int G_{n}\le 2c_{n}\bar C_{G}.
\]
With our $r$, $2c_{n}\bar C_{G}=\eta^{2}/4$, hence
$\He(\ell^{(k)},u^{(k)})\le\eta/2\le\eta$. Thus, we've constructed ${N}(r,\Theta_n,\varrho)$ $\eta$--brackets covering $\Pnsieve$, proving~\eqref{eq:cov2br}.

To show $\eqref{eq:brentropy}$, it remains to bound $K={N}(r,\Theta_n,\varrho)$. We first reduce the product
space $\Theta_n$ to its factors, then bound each factor based on volumes (Lemma~5.7, \citet{wainwright2019high}).
The metric $\varrho$ in~\eqref{eq:paramnorm} is the \emph{sum} of the three
factor metrics
$\varrho=\varrho_{\Delta}+\varrho_{Z}+\varrho_{\Sigma}$, with
$\varrho_{\Delta}=\lVert w-w' \rVert_{1}$,
$\varrho_{Z}=\max_{j}\lVert z_{j}-z_{j}' \rVert_{\infty}$, and
$\varrho_{\Sigma}=\lVert \Sigma-\Sigma' \rVert_{\mathrm{op}}$. If
$\mathcal N_{\Delta},\mathcal N_{Z},\mathcal N_{\Sigma}$ are $(r/3)$--covers of the
three factors in their respective metrics, then for any
$\theta=(w,z,\Sigma)$ one may choose factor--centers within $r/3$ of $w,z,\Sigma$
respectively, and the sum metric gives
\[
  \varrho\bigl(\theta,\text{chosen center}\bigr)
   =\varrho_{\Delta}+\varrho_{Z}+\varrho_{\Sigma}
   \le\tfrac r3+\tfrac r3+\tfrac r3=r.
\]
Hence the product grid
$\mathcal N_{\Delta}\times\mathcal N_{Z}\times\mathcal N_{\Sigma}$ is an
$r$--cover of $\Theta_n$, so
\begin{equation}\label{eq:prod}
  {N}(r,\Theta_n,\varrho)
   \le
   {N}(\tfrac r3,\Delta_{\N},\varrho_{\Delta})\cdot
   {N}(\tfrac r3,[-a,a]^{d\N},\varrho_{Z})\cdot
   {N}(\tfrac r3,\mathcal D_{\sigma},\varrho_{\Sigma}).
\end{equation}

Since $a=\sigma^{-2}$ by \eqref{eq:paramrates}, the radius chosen above satisfies $r=\eta^{2}\sigma^{4d+6}/(8C_{\mathrm L}C_{G})\le\sigma^{2}$, using $\eta\le1$, $\sigma\le1$ and $C_{\mathrm L},C_{G}\ge1$. Together with $\sigma^{2}\le1\le a$ and $\sigma^{2}\le \sigma^{-2}$, the latter since $\sigma \leq 1$, yields $1+6X/r\le 7X/r$ for each of the three scale factors $X$ arising in the estimates below, namely $X=1$, $X=a$ and $X=D$ with $D$ as in \eqref{eq:Ddef}. We write $C:=7$ for this constant throughout.
First, $([-a,a]^{d\N}, \varrho_Z)$ is the $\ell_{\infty}$--ball of radius $a$
in dimension $d\N$, so Wainwright~(2019, Example~5.2), applied at mesh
$\rho=r/3$, gives
\begin{equation}\label{eq:cube_number}
  N\bigl(\tfrac r3,[-a,a]^{d\N},\varrho_{Z}\bigr)
   \le \Bigl(1+\tfrac{2a}{r/3}\Bigr)^{d\N}
    =\Bigl(1+\tfrac{6a}{r}\Bigr)^{d\N}
   \le \Bigl(\tfrac{Ca}r\Bigr)^{d\N}.
\end{equation}

The set
$\mathcal D_{\sigma}=\{\Sigma\in\mathrm{Sym}(d):\sigma^{2}\le
\mathrm{eig}_{1}(\Sigma)\le\cdots\le\mathrm{eig}_{d}(\Sigma)\le\sigma^{-2}\}$
lies in the vector space $\mathrm{Sym}(d)$ of dimension $m=\tfrac{d(d+1)}2$, on
which $\lVert \cdot \rVert_{\mathrm{op}}$ is a norm. Since every
$\Sigma\in\mathcal D_{\sigma}$ is symmetric positive definite,
$\lVert \Sigma \rVert_{\mathrm{op}}=\mathrm{eig}_{d}(\Sigma)$, and the eigenvalue cap in the definition of $\mathcal{D}_\sigma$ gives
\begin{equation}\label{eq:Ddef}
  D:=\sigma^{-2},
  \qquad \lVert \Sigma \rVert_{\mathrm{op}}\le D
  \ \ \text{for all }\Sigma\in\mathcal D_{\sigma}.
\end{equation}
Consequently
\begin{equation}
  \mathcal D_{\sigma}\ \subseteq\ \mathbb B_{\mathrm{op}}(0,D)
   :=\{\Sigma\in\mathrm{Sym}(d):\lVert \Sigma \rVert_{\mathrm{op}}\le D\}.
\end{equation}
Covering a subset is no harder than covering a superset:
$N(\rho,\mathcal D_{\sigma},\lVert \cdot \rVert_{\mathrm{op}})\le
N(\rho,\mathbb B_{\mathrm{op}}(0,D),\lVert \cdot \rVert_{\mathrm{op}})$, because any $\rho$--cover of the
larger set restricts to a $\rho$--cover of the smaller.
Since $\mathbb B_{\mathrm{op}}(0,D)=D\,\mathbb B_{\mathrm{op}}(0,1)$, we have
$N(\rho,\mathbb B_{\mathrm{op}}(0,D),\lVert \cdot \rVert_{\mathrm{op}})
=N(\rho/D,\mathbb B_{\mathrm{op}}(0,1),\lVert \cdot \rVert_{\mathrm{op}})$, and
Lemma~5.7 (self-covering, Example~5.8 of \citet{wainwright2019high}) applied at mesh $\rho=r/3$ with dimension
$m=\tfrac{d(d+1)}2$ yields
\begin{equation}\label{eq:Sigma_number}
  N\bigl(\tfrac r3,\mathcal D_{\sigma},\lVert \cdot \rVert_{\mathrm{op}}\bigr)
   \le N\bigl(\tfrac{r}{3D},\mathbb B_{\mathrm{op}}(0,1),\lVert \cdot \rVert_{\mathrm{op}}\bigr)
   \le\Bigl(1+\tfrac{2D}{r/3}\Bigr)^{d(d+1)/2}
   \le\Bigl(\tfrac{CD}{r}\Bigr)^{d(d+1)/2},
\end{equation}
using $1+\tfrac{6D}r\le\tfrac{CD}r$ for $C=7$, since $r\le\sigma^{2}\le D$.

Finally, we bound $N\!\left(\tfrac{r}{3},\, \Delta_{\N},\, \varrho_\Delta = \lVert \cdot \rVert_1 \right)$.
Since every $w \in \Delta_{\N}$ satisfies $w \ge 0$ and $\lVert{w}\rVert_1 = 1$, we have
$\Delta_{\N} \subseteq \mathbb{B}_1^{\N}$, the $\ell_1$ unit ball in $\R^{\N}$. This
ball {is} a norm's unit ball, so self-covering applies at this larger space, and
\begin{equation}\label{eq:Delta_number}
N(\delta,\, \Delta_{\N},\, \ell_1) \leq 
     N(\delta,\, \mathbb{B}_1^{\N},\, \ell_1)
    \;\le\; \!\left(1 + \tfrac{2}{\delta}\right)^{\N}
    =\Bigl(1+\tfrac{6}{r}\Bigr)^{\N} 
    \le\Bigl(\tfrac{C}{r}\Bigr)^{\N}.
\end{equation}

Substituting \eqref{eq:cube_number}, \eqref{eq:Sigma_number}, and
\eqref{eq:Delta_number} into the product bound \eqref{eq:prod} and taking
logarithms gives
\begin{equation}\label{eq:logprod}
  \log{N}(r,\Theta_n,\varrho)
   \le
   \N\log\tfrac{C}{r}
   +d\N\log\tfrac{Ca}{r}
   +\tfrac{d(d+1)}{2}\log\tfrac{CD}{r}.
\end{equation}
Group the first two ($\N$--proportional) terms and use
$\log D= - 2\log \sigma$ from \eqref{eq:Ddef} in the third:
\begin{equation}\label{eq:logprod2}
  \log{N}(r,\Theta_n,\varrho)
   \le
   d\,\N\Bigl[2\log C+\log\tfrac a r+\log\tfrac1r\Bigr]
   +\tfrac{d(d+1)}{2}\Bigl[\log\tfrac{C}{r} - 2\log \sigma\Bigr].
\end{equation}
$r=\eta^{2}/\bigl(8C_{\mathrm L}C_{G}(a/\sigma)^{d}\sigma^{-(d+6)}\bigr)$ gives
\[
  \log\tfrac1r
   =2\log\tfrac1\eta+(d+6)\log\tfrac1\sigma+d\log\tfrac a\sigma+\log(8C_{\mathrm L}C_{G}),
\]
so each $\log(\cdot/r)$ term in \eqref{eq:logprod2} is bounded by a
$d$--dependent multiple of
$\log\tfrac a\sigma+\log\tfrac1\eta+\log\tfrac1\sigma$ (the additive constants
$\log(8C_{\mathrm L}C_{G})$ and $- 2\log\sigma$, the extra $d\log\tfrac a\sigma$ from
the envelope factor $(a/\sigma)^{d}$, and the extra $4\log\tfrac1\sigma$ from the factor
$\sigma^{-4}$ in $\bar C_{G}$ are all of this order and absorbed into $A_{1}$).
Collecting the $\N$--proportional terms, bounding
$\tfrac{d(d+1)}2\le d^{2}$ for the covariance term, and using
$\log{N_{[\,]}}(\eta,\Pnsieve,\He)\le\log{N}(r,\Theta_n,\varrho)$
from~\eqref{eq:cov2br}, yields~\eqref{eq:brentropy}, with $A_{1}$ depending only on $d$.
\end{proof}

Finally, we need the following additional Lemma to prove Theorem \ref{Thm:DensityResult}:

\begin{lemma}\label{lem:kaji}
Let $\delta\in\bigl(0,(\log4)^{-1}\bigr)$. For every $\rho$ with $4<\rho\le e^{1/\delta}$,
\begin{equation}\label{eq:kaji-pt}
   (\log\rho)^{2}\;\le\;\frac{4}{\delta^{2}}\,\frac{(\sqrt{\rho}-1)^{2}}{\rho}.
\end{equation}
\end{lemma}

\begin{proof}
Put $x=1/\rho$, so that $e^{-1/\delta}\le x<1/4$. On this range
$\bigl(\log(1/x)\bigr)^{2}\le\delta^{-2}$, and since $x\mapsto(\sqrt{x}-1)^{2}$ is
decreasing on $(0,1)$,
\[
   \delta^{-2}\;=\;4\delta^{-2}\bigl(\sqrt{1/4}-1\bigr)^{2}\;\le\;4\delta^{-2}(\sqrt{x}-1)^{2}.
\]
Substituting back and using $\bigl(\sqrt{1/\rho}-1\bigr)^{2}=(\sqrt{\rho}-1)^{2}/\rho$
gives \eqref{eq:kaji-pt}.
\end{proof}

\densityresult*

\begin{proof}


We recall that
\begin{align*}
    \delta_n :=  C_4 n^{-\beta/(2\beta + d)} (\log(n))^t, \qquad    \N := C\frac{n\delta_n^2}{\log(n\delta_n^2)},\qquad
    a := n \delta_n^2, \qquad
    \sigma^{-2} := n\delta_n^2,
\end{align*}
as defined in \eqref{eq:paramrates}.
Denote $\sigma_A$ the common standard deviation of the approximating point covariance $\sigma^2_A I$ used below. Since $n\delta_n^{2}\to\infty$, we assume throughout that $n$ is large enough that
$n\delta_n^{2}\ge \pi$, so that $\sigma \leq \pi^{-1/2} < 1 \leq a$ and $a = \sigma^{-2} \geq \sqrt{\pi} \sigma^{-1}$, as required in Lemma~\ref{lem:lip}. This is in any case required
for $\N$ to be a positive integer.

\paragraph{Step 1} We first show that there exists $p_n \in \Pnsieve$, such that \eqref{eq:approx} and \eqref{eq:tail} hold. We start by following \cite[Proposition 9.14]{Fundamentals}: Let for $a_0$ large, $a_{\sigma_A}=a_0 \logm(\sigma_A)^{1/\tau}$, there exists a discrete probability measure
\[
F^*_{\sigma_A} =\sum_{j=1}^{\NI} w_j^* \delta_{z_j},
\]
with at most $\NI \lesssim (a_{\sigma_A}/\sigma_A)^d (\logm (\epsilon))^d$ support points $z_1, \ldots, z_{\NI}$ inside $[-a_{\sigma_A}, a_{\sigma_A}]^d$. The support points and weights of $F^*_{\sigma_A}$ can be specifically chosen, such that $\He(p_0, p_{F^*_{\sigma_A}, \sigma_A^2 I})\lesssim \epsilon + \sigma^{\beta}_A$. 
Moreover, as in \cite[Proposition 9.14]{Fundamentals}, it is also possible to construct a
partition $U_1, \ldots, U_{\NII}$ of $[-a_{\sigma_A}, a_{\sigma_A}]^d$ of size
$\NII \lesssim \NI$ into nonempty cells of diameter at most $\sigma_A$, refined so that
every cell containing a support point $z_j$ has diameter $\lesssim \sigma_A \epsilon^2$.
Throughout we write
\[
   w_j^* := F^*_{\sigma_A}(U_j), \qquad j = 1, \ldots, \NII,
\]
for the mass the approximating measure assigns to cell $U_j$, so that $\sum_{j \le \NII} w_j^* = 1$.
We then define the set
\[
B := \left\{ p_{F, \Sigma} : \sum_{j=1}^{\NII} |F(U_j) - w_j^*| \le \epsilon^2, \ \min_{1 \le j \le \NII} F(U_j) \ge \epsilon^4, \ \frac{1}{\sigma_A^2} \le  \mathrm{eig}_{1}(\Sigma^{-1}) \le \mathrm{eig}_{d}(\Sigma^{-1}) \le \frac{1 + \sigma_A^\beta}{\sigma_A^2} \right\},
\]
which is an adapted version of the set $B$ in \cite[Equation (9.14)]{Fundamentals}. The main difference is that the conditions on $F$ are formulated over the cube partition $U_1, \ldots, U_{\NII}$ of $[-a_{\sigma_A}, a_{\sigma_A}]^d$, rather than a partition of all of $\R^d$ as in the book.

 In the following we choose
\begin{align}\label{eq:sigmaAcalib}
    \epsilon^2 \asymp \min(\sigma_A^d, \sigma_A^{2\beta})
       \bigl(\logm(\sigma_A)\bigr)^{-(d/\tau + d + 1)},
    \qquad
    \sigma_A^{2\beta + d} \asymp n^{-1} (\log n)^{d/\tau + d - 1},
\end{align}
so that $\sigma_A \asymp n^{-1/(2\beta+d)}(\log n)^{(d/\tau+d-1)/(2\beta+d)}$ and
$\delta_n \asymp \sigma_A^\beta \logm(\sigma_A)$.
 \cite[Proof of Theorem 5.1]{Bayes} show that with this choice, for all $n$ large enough,
\begin{align}\label{eq_Cond1}
\frac{\epsilon^4}{\sigma_A} < 0.4 e \max_m \|p^{(m)}_{0}\|_{\infty} 
\end{align}
\begin{align}\label{eq_Cond2}
    (\sigma_A^{2\beta} + \epsilon^2) \logm^2\left(\frac{\epsilon^4}{\sigma_A}\right) \lesssim \delta_n^2, \ \ \frac{\epsilon^4}{\sigma_A} < 1
\end{align}
 \begin{align}\label{eq_Cond3}
     (\logm(\sigma_A))^{d/\tau} \sigma_A^{-d} (\logm(\epsilon))^{d+1} \lesssim (\logm(\sigma_A))^{d/\tau+d+1} \sigma_A^{-d} \lesssim n \delta_n^2, 
 \end{align}
where the first inequality in \eqref{eq_Cond3} follows from the choice of $\epsilon$. Then we have that:

\claim{Claim 1:} For any $p_{n, \infty} \in B$, \eqref{eq:approx} holds.

It immediately follows from \cite[Chapter 9.4]{Fundamentals}, in particular page 249, that
\[
\HeMAR(p_0,p_{n, \infty}) \leq \He(p_0,p_{n, \infty}) \lesssim \sigma_A^{\beta} + \epsilon \lesssim \delta_n,
\]
where the last inequality follows by \eqref{eq_Cond2}.

\claim{Claim 2:} For any $p_{n, \infty} \in B$, \eqref{eq:tail} holds.

Following \cite[Proof of Proposition 9.14]{Fundamentals} 
for any density $p_{n, \infty} \in B$,
\begin{align}\label{eq_lowerbound_1}
    \frac{1}{p_{n, \infty}(x)}  \leq \frac{\sigma_A^d}{e^{-1} \epsilon^4}, \text{ for } x \in [-a_{\sigma_A}, a_{\sigma_A}]^d, \text{ and }   \frac{1}{p_{n, \infty}(x)} \lesssim \sigma_A^d e^{2d \|x\|^2/\sigma_A^2}, \text{ for } x \notin [-a_{\sigma_A}, a_{\sigma_A}]^d,
\end{align}
where $a_{\sigma_A}=a_0 \logm(\sigma_A)^{1/\tau}$ with $a_0$ large.
Let
\[
p_{m,0,F,\Sigma}(x) := p_{0}^{(m)}(x^{(m)})p_{n,\infty}(x^{(\bar{m})} \mid x^{(m)}),
\]
which has the property that for $\sigma_A < 1$, 
\begin{align}\label{eq_upperbound_1}
    p_{m,0,F,\Sigma}(x) \leq \frac{2\max_{m}\|p_{0}^{(m)} \|_{\infty}}{\sigma_A^{d-1}} \;\; \text{ for all } x \in \mathcal{X}, m\in\{0,1\}^d\setminus \{1\}^d.
\end{align}
Factoring $p_{n,\infty}(x)=p_{n,\infty}^{(m)}(x^{(m)})\,p_{n,\infty}(x^{(\bar{m})}\mid x^{(m)})$, the conditional cancels and
\begin{align*}
    \rho^{(m)}(x^{(m)}) \;=\; \frac{p_{0}^{(m)}(x^{(m)})}{p_{n, \infty}^{(m)}(x^{(m)})} = \frac{p_{m,0,F,\Sigma}(x)}{p_{n, \infty}(x)},
    \quad\text{for every }x.
\end{align*}
This transfers the joint bounds \eqref{eq_lowerbound_1} and
\eqref{eq_upperbound_1} onto the marginal ratio we need in \eqref{eq:tail}.
Set 
$
   \tilde\epsilon\;:=\;\frac{2\max_{m}\lVert p_{0}^{(m)}\rVert_{\infty}\,\sigma_{A}}{e^{-1}\epsilon^{4}} .
$
From \eqref{eq_lowerbound_1} and \eqref{eq_upperbound_1}, 
\begin{align}\label{eq:localise}
    \frac{p_{0}^{(m)}(x^{(m)})}{p_{n, \infty}^{(m)}(x^{(m)})}  > \tilde{\epsilon} \implies \|x \|_{\infty} > a_{\sigma_A} \implies \log\left(\frac{1}{p_{n, \infty}(x)}\right) \lesssim  \logm(\sigma_A) +\frac{\|x\|^2}{\sigma_A^2}.
\end{align}
Since $\epsilon^{4}/\sigma_{A}\to0$ by the choice of $\epsilon$ and $\sigma_{A}$,
\begin{equation}\label{eq_larger4}
   \tilde\epsilon\;>\;4\qquad\text{for all $n$ large enough.}
\end{equation}
Moreover, writing $c:=\max(d,2\beta)$, \eqref{eq:sigmaAcalib} gives
$\epsilon^{4}/\sigma_{A}\asymp\sigma_{A}^{2c-1}\bigl(\logm(\sigma_{A})\bigr)^{-2(d/\tau+d+1)}$, so that
$\logm(\epsilon^{4}/\sigma_{A})=(2c-1)\logm(\sigma_{A})+2(d/\tau+d+1)\log\logm(\sigma_{A})+O(1)\asymp\logm(\sigma_{A})$, and hence
\begin{equation}\label{eq:deltasigma}
   \log\tilde\epsilon
   \;=\;\log\bigl(2e\max_{m}\lVert p_{0}^{(m)}\rVert_{\infty}\bigr)
        +\logm\!\Bigl(\frac{\epsilon^{4}}{\sigma_{A}}\Bigr)
   \;\asymp\;\logm(\sigma_{A}).
\end{equation}
We now show a pointwise bound on the tail $\{x \in \mathcal{X}: \rho^{(m)}(x^{(m)})>\tilde\epsilon\}$:
On this set, for any $m \in \{0, 1\}^d \setminus \{1\}^d$, 
the denominator of $\rho_{p_0, p_{n,\infty}}^{(m)}(x^{(m)})$ obeys the second implication of \eqref{eq:localise}; 
and its numerator obeys \eqref{eq_upperbound_1}, so
\begin{align}
    \Bigl(\log\rho_{p_0, p_{n,\infty}}^{(m)}(x^{(m)})\Bigr)^{2} 
    &\leq \left(  \log(2\max_{m}\|p_{0}^{(m)} \|_{\infty}) + (d-1)\logm(\sigma_A)  + \logm(\sigma_A) +\frac{\|x\|^2}{\sigma_A^2} \right)^2 \nonumber \\
    &\leq  2 \left(  \logm(\sigma_A) +\frac{\|x\|^2}{\sigma_A^2}  \right)^2 + 2\left(  \log(2\max_{m}\|p_{0}^{(m)} \|_{\infty}) + (d-1)\logm(\sigma_A)  \right)^2 \nonumber \\
    & \lesssim \left(  \logm(\sigma_A) +\frac{\|x\|^2}{\sigma_A^2}  \right)^2, \label{eq:pointwise-tail}
\end{align}
where, for $\sigma_A < 1$ small, the last step followed because
\begin{align*}
   \log(2\max_{m}\|p_{0}^{(m)} \|_{\infty}) + (d-1)\logm(\sigma_A)\leq C \left(\logm(\sigma_A) +\frac{\|x\|^2}{\sigma_A^2}\right).
\end{align*} 
On $\{4<\rho^{(M)}\le\tilde\epsilon\}$:
Let $\delta_{\sigma}:=1/\log\tilde\epsilon$, so that $e^{1/\delta_{\sigma}}=\tilde\epsilon$ and $\delta_{\sigma}\in\bigl(0,(\log4)^{-1}\bigr)$. Hence Lemma~\ref{lem:kaji} applies pointwise, giving us
\begin{align}
   &\E_{ \Pjoint}\Bigl[\bigl(\log\rho^{(M)}_{ p_0, p_{n, \infty}} (X^{(M)}) \bigr)^{2}\one\{4<\rho_{ p_0, p_{n, \infty}}^{(M)} (X^{(M)}) \le\tilde\epsilon\}\Bigr] \\
   \le &\frac{4}{\delta_{\sigma}^{2}}\sum_{m}\int
      \frac{\bigl(\sqrt{\rho_{ p_0, p_{n, \infty}}^{(m)}(x^{(m)})}-1\bigr)^{2}}{\rho_{ p_0, p_{n, \infty}}^{(m)}(x^{(m)})}
      \,p_{0}(x)\,\Prob(M=m\mid x)\,\mathrm{d}x \notag\\
   = &\frac{4}{\delta_{\sigma}^{2}}\sum_{m}\int
      \Bigl(\sqrt{p_{0}^{(m)}}-\sqrt{p_{n,\infty}^{(m)}}\Bigr)^{2}
      \Prob(M=m\mid x^{(m)})\,\mathrm{d}x^{(m)} \notag\\
   = &\frac{4}{\delta_{\sigma}^{2}}\,\HeMAR^{2}(p_{0},p_{n,\infty})
   \;\le\;\frac{4}{\delta_{\sigma}^{2}}\,\He^{2}(p_{0},p_{n,\infty}) \notag\\
   \lesssim &\logm^{2}(\sigma_{A})\bigl(\sigma_{A}^{2\beta}+\epsilon^{2}\bigr)
   \;\lesssim\;\delta_{n}^{2},
   \label{eq:annulus}
\end{align}
where the last two steps use \eqref{eq:deltasigma}, claim 1, and \eqref{eq_Cond2}.

On $\{\rho^{(M)}>\tilde\epsilon\}$: 
\begin{align*}
    & \E_{\Pjoint}\!\left[\Bigl(\log\rho_{ p_0, p_{n, \infty}}^{(M)}(X^{(M)})\Bigr)^2\,\one\{\rho_{ p_0, p_{n, \infty}}^{(M)}(X^{(M)})>\tilde{\epsilon}\}\right]\\
    &=\sum_{m} \int \log\left(\frac{p_{m,0,F,\Sigma}(x)}{p_{n, \infty}(x)}\right)^2\mathbf{1}\left \{\frac{ p_{m,0,F,\Sigma}(x)}{p_{n, \infty}(x)} > \tilde{\epsilon} \right \}p_0(x) \Prob(M=m\mid x) \mathrm{d}x\\
    &\lesssim  \int_{\|x \|_{\infty} > a_{\sigma_A}} \left(\logm(\sigma_A)^2  +  \frac{\|x\|^4}{\sigma_A^4} + 2 \logm(\sigma_A) \frac{\|x\|^2}{\sigma_A^2}\right)p_0(x)  \mathrm{d}x,
\end{align*}
using \eqref{eq:pointwise-tail} and $\sum_{m} \Prob(M=m \mid x )=1$. Now by the tail assumption on $p_0$, the fact that $\|x \| \geq \|x \|_{\infty}$ and the choice of $a_{\sigma_A}=a_0 \logm(\sigma_A)^{1/\tau}$, for large enough $a_0$, 
\begin{align*}
     &\int_{\|x \| > a_{\sigma_A}} \left(\logm(\sigma_A)^2  +  \frac{\|x\|^4}{\sigma_A^4} + 2 \logm(\sigma_A) \frac{\|x\|^2}{\sigma_A^2}\right)c e^{-b \|x\|^{\tau}}  \mathrm{d}x\\
          &= \frac{e^{-b a_{\sigma_A}^{\tau}}}{\sigma_A^4}\int_{\|x \| > a_{\sigma_A}} \left(\sigma_A^4\logm(\sigma_A)^2  +  \|x\|^4 + 2 \logm(\sigma_A) \sigma_A^2\|x\|^2\right)c e^{-b (\|x\|^{\tau} -a_{\sigma_A}^{\tau})}  \mathrm{d}x \\
          &\lesssim \frac{e^{-b a_{\sigma_A}^{\tau}}}{\sigma_A^4}\\
          &=\sigma_A^{b a_0^{\tau}-4},
\end{align*}
which is strictly smaller than $\sigma_A^{2\beta}$ for $a_0$ large. It follows from \eqref{eq_Cond2} that \eqref{eq:tail} holds.

\bigskip

\claim{Claim 3:} There exists $p_n \in B \cap \Pnsieve$.

We start by showing that
\begin{align}\label{eq_firstclaim}
    (\logm(\sigma_A))^{d/\tau} \sigma_A^{-d} (\logm(\epsilon))^{d} \lesssim \N
\end{align}
If this holds, we have that,
$$\NII \lesssim \NI \lesssim \left(\frac{a_{\sigma_A}}{\sigma_A}\right)^d \logm(\epsilon)^d \lesssim  (\logm(\sigma_A))^{d/\tau} \sigma_A^{-d} (\logm(\epsilon))^{d} \lesssim \N, $$
and we can choose the constant in $N$ large enough, such that $\NII \leq \N$.

From \eqref{eq_Cond3}, and that $\logm(\epsilon) \asymp \logm(\min(\sigma_A^d, \sigma_A^{2\beta})) \lesssim \logm(\sigma_A)$ it follows that
\begin{align}\label{eq_Cond3d}
    (\logm(\sigma_A))^{d/\tau} \sigma_A^{-d} (\logm(\epsilon))^{d} \lesssim (\logm(\sigma_A))^{d/\tau+d} \sigma_A^{-d} \lesssim \frac{n \delta_n^2}{\logm(\sigma_A)}.
\end{align}
Moreover, since,
\begin{align*}
    (\logm(n^{-1} \log(n)^{d/\tau + d -1}))^{d/\tau + d +1}=(\log(n) - (d/\tau + d -1)\log(\log(n)))^{d/\tau + d +1} \gtrsim \log(n)^{d/\tau + d +1},
\end{align*}
a refinement shows that 
\begin{align*}
    (\logm(\sigma_A))^{d/\tau+d+1} \sigma_A^{-d}&= (\logm(n^{-1} \log(n)^{d/\tau + d -1}))^{d/\tau + d +1} n^{d/(2\beta + d)} (\log(n))^{-(d/\tau + d -1) \cdot d/(2\beta + d)}\\
    &\gtrsim \log(n)^{d/\tau + d +1} n^{d/(2\beta + d)} (\log(n))^{-(d/\tau + d -1) \cdot d/(2\beta + d)}
\end{align*}
Collecting the exponents, we have $n^{d/(2\beta + d)}= n n^{d/(2\beta + d)-1}=n n^{-2\beta /(2\beta + d)}$, while the power of $\log(n)$ is
\begin{align*}
d/\tau+d+1 - \frac{(d/\tau+d-1)d}{2\beta+d} 
&= \frac{(d/\tau+d+1)(2\beta+d) - (d/\tau+d-1)d}{2\beta+d}\\
&= \frac{2\beta d/\tau + 2\beta d + 2\beta + 2d}{2\beta+d}\\
&= \frac{2(\beta d/\tau + \beta d + \beta + d)}{2\beta+d},
\end{align*}
so that we recover $n \delta_n^2$, showing that 
$$\log(n \delta_n^2) \lesssim d\log(\sigma_A^{-1}) + (d/\tau+d+1)\log(\logm(\sigma_A) ) \lesssim \logm(\sigma_A) $$

Plugging the result back into \eqref{eq_Cond3d}, we obtain
\[
(\logm(\sigma_A))^{d/\tau} \sigma_A^{-d} (\logm(\epsilon))^{d} \lesssim \frac{n \delta_n^2}{\log(n \delta_n^2)} \lesssim \N,
\]
as desired.

Take the partition $\{U_j\}_{j=1}^{\NII}$ constructed above, from which we build $\tilde F \in \mathcal{F}_{\N, \delta_n, a}$ and verify the constraints on $F$ defining the set $B$.
Assume
\begin{equation}\label{eq:witness-cond}
   2\,\NII\,\epsilon^{2}\;\le\;1,
\end{equation}
which is possible as $\NII\,\epsilon^{2} \to 0$.
Let $\lambda:=\NII\epsilon^{4}\in(0,1)$ and define the probability vector
\[
   v_{j}\;:=\;(1-\lambda)\,w_{j}^{*}+\frac{\lambda}{\NII}
   \;=\;(1-\NII\epsilon^{4})\,w_{j}^{*}+\epsilon^{4},
   \qquad j=1,\dots,\NII.
\]
Then $\sum_{j\le\NII}v_{j}=(1-\lambda)+\lambda=1$, and
\[
   \min_{j\le\NII}v_{j}\;\ge\;\frac{\lambda}{\NII}\;=\;\epsilon^{4},
   \qquad
   \sum_{j\le\NII}\lvert v_{j}-w_{j}^{*}\rvert
   =\lambda\sum_{j\le\NII}\Bigl\lvert\frac{1}{\NII}-w_{j}^{*}\Bigr\rvert
   \;\le\;2\lambda\;=\;2\NII\epsilon^{4}\;\le\;\epsilon^{2},
\]
the last step by \eqref{eq:witness-cond}. 
Pick a $z_{j}\in U_{j}$ for each $j\le\NII$; by $\NII \lesssim \N$ we may pick further points $z_{\NII+1},\dots,z_{\N}\in[-a,a]^{d}$ arbitrarily and set
\[
   \tilde F\;:=\;\sum_{j=1}^{\N}v_{j}\delta_{z_{j}},
   \quad v_{\NII+1}=\dots=v_{\N}:=0.
\]
Hence $\tilde F(U_{j})\;=\;v_{j}$ for every $j\le\NII$ and $\tilde F$ inherits both $F$-constraints. 
Since $U_{1},\dots,U_{\NII}$ are disjoint and $z_{j}\in U_{j}$ for $j\le\NII$, while $v_{\NII+1}=\dots=v_{\N}=0$, we have $\tilde F(U_{j})=v_{j}$ for every $j\le\NII$, so $\tilde F$ inherits both $F$-constraints of the set $B$. Moreover $a_{\sigma_{A}}=a_{0}\logm(\sigma_{A})^{1/\tau}\lesssim n\delta_n^{2}=a$ by \eqref{eq_Cond3}, so $[-a_{\sigma_{A}},a_{\sigma_{A}}]^{d}\subset[-a,a]^{d}$ for $n$ large; hence $z_{j}\in[-a,a]^{d}$ for every $j\le\N$, giving $\tilde F \in \mathcal{F}_{\N, \delta_n, a}$.


For $\Sigma$, the conditions in $B$ dictate that
\begin{align*}
 \frac{\sigma_A^2}{1 + \sigma_A^\beta}\le \mathrm{eig}_{1}(\Sigma) \le   \mathrm{eig}_{d}(\Sigma)  \le \sigma_A^2.
\end{align*}
Now because $\sigma_A, \epsilon < 1$ for all $n$ large enough, \eqref{eq_Cond3} implies that,
\begin{align*}
    \sigma_A^2 \geq  \sigma_A^{d} (\logm(\sigma_A))^{-d/\tau}  (\logm(\epsilon))^{-(d+1)} \gtrsim (n \delta_n^2)^{-1}=\sigma^2 .
\end{align*}
Moreover, since $\sigma^2_A \to 0$ and the eigenvalue upper bound $ \sigma^{-2} \to \infty$, $\sigma^2_A\leq \sigma^{-2}$ for $n$ large enough. Thus, we can simply take $\Sigma=\sigma_A^2 I$ to meet the eigenvalue criteria of both the set $B$ and $\mathcal{D}_{\sigma}$. Thus we have by construction that, 
\[
p_n:=p_{\tilde{F}, \sigma_A^2 I} \in B \cap \Pnsieve,
\]
as required.

\paragraph{Step 2} We show that \eqref{eq:entropy} is satisfied with $\alpha_0 = 1$. Specifically, that 
there are constants $C_1$ and $C_2$, independent of $n$,
such that for all large $n$,
\begin{equation}\label{eq:clearance}
  J_{[\,]}\bigl(\delta_{n},\mathcal P_{n,\delta_{n}/\sqrt{c_{0}}},\He\bigr)
   \;\le\;C_1\,\delta_{n}\sqrt{\N \log n}
   \;\le\; C_2 \delta_{n}^{2}\sqrt n .
\end{equation}
Lemma~\ref{lem:cov2br} applies to $\Pnsieve$ directly. Since 
$\log\tfrac a\sigma+\log\tfrac1\sigma=\log\tfrac a{\sigma^{2}}$, setting
\[
  K_{n}:=A_{1}\bigl(\N+d^{2}\bigr),
  \qquad
  B_{n}:=\frac{a}{\sigma^{2}}=(n\delta_{n}^{2})^{2},
\]
\eqref{eq:brentropy} becomes, for every $\eta\in(0,\delta_{n}]$
\begin{equation}\label{eq:entropysingle}
  1+\log N_{[\,]}(\eta,\mathcal P_{n,\delta_{n}/\sqrt{c_{0}}},\He)
   \;\le\;K_{n}\log\tfrac{B_{n}}\eta+1 .
\end{equation}
Since $B_{n}=(n\delta_{n}^{2})^{2}\to\infty$ while $\delta_{n}\to0$, we have $B_{n}\ge e$ and
$\delta_{n}\le B_{n}/e$ for all $n$ large, which the integral bound below requires.

For $B\ge e$ and $\delta\in(0,B/e]$, $\displaystyle\int_{0}^{\delta}\sqrt{\log\tfrac B\eta}\,\mathrm{d}\eta \le 2\delta\sqrt{\log\tfrac B\delta}$: Indeed, writing $s_{0} := \log(B/\delta) \geq 1$, the substitution $\eta=\delta e^{-s}$ gives
\begin{align}\label{eq:step3x}
    &\int_{0}^{\delta}\sqrt{\log\tfrac B\eta}\,\mathrm{d}\eta
   = \delta\int_{0}^{\infty}\sqrt{s_{0}+s}\;e^{-s}\,\mathrm{d} s \nonumber \\
   \le &\delta\int_{0}^{\infty}\bigl(\sqrt{s_{0}}+\sqrt s\bigr)e^{-s}\,\mathrm{d} s \\
   = &\delta\Bigl(\sqrt{s_{0}}+\tfrac{\sqrt\pi}{2}\Bigr)
   \le \delta\bigl(\sqrt{s_{0}}+1\bigr) \le 2\delta\sqrt{s_{0}}, \nonumber
\end{align}
where the second equality follows from $\int_0^\infty e^{-s}\mathrm{d} s=1,\ \int_0^\infty\sqrt s\,e^{-s}\mathrm{d} s=\tfrac{\sqrt\pi}2$; the second inequality follows from $\tfrac{\sqrt\pi}2\le1$, and the last from $s_{0}\ge1$, which holds since $\delta\le B/e$.
Now we bound $J_{[\,]}(\delta_{n},\mathcal P_{n,\delta_{n}/\sqrt{c_{0}}},\He)$:
\begin{align}
  J_{[\,]}(\delta_{n},\mathcal P_{n,\delta_{n}/\sqrt{c_{0}}},\He)
  &=\int_0^{\delta_n}\sqrt{1+\log N_{[\,]}(\eta,\mathcal P_{n,\delta_{n}/\sqrt{c_{0}}},\He)}\,\mathrm{d}\eta \notag\\
  &\le\sqrt{K_{n}}\int_{0}^{\delta_{n}}\sqrt{\log\tfrac{B_{n}}\eta}\,\mathrm{d}\eta
      +\delta_{n} \notag\\
   &\le2\sqrt{K_{n}}\,\delta_{n}\sqrt{\log\tfrac{B_{n}}{\delta_{n}}}
      +\delta_{n}
   \;\le\;C_1\,\delta_{n}\sqrt{\N \log n},
      \label{eq:Jbound}
\end{align}
where the first inequality uses $\sqrt{x+y}\le\sqrt x+\sqrt y$ with
$x=K_{n}\log(B_{n}/\eta)$ and $y=1$, by \eqref{eq:entropysingle}; the second is
\eqref{eq:step3x}; and the last uses $K_{n}\le 2A_{1}\N$ for $n$ large (as $\N\to\infty$
with $d$ fixed) together with
$\log(B_{n}/\delta_{n})=\log B_{n}+\log(1/\delta_{n})\asymp\log n$, the trailing $\delta_{n}$
being absorbed since $\N\log n\ge1$. This proves the first inequality
in~\eqref{eq:clearance}.
By \eqref{eq:paramrates}, $\log(n\delta_{n}^{2})\asymp\log n$ and hence
$\N \log n\asymp n\delta_{n}^{2}\asymp n^{\frac d{2\beta+d}}(\log n)^{2t}$, so both sides
of \eqref{eq:clearance} carry the same order,
\[
  \delta_{n}\sqrt{\N \log n}
   \asymp n^{-\frac\beta{2\beta+d}}(\log n)^{t}\cdot n^{\frac d{2(2\beta+d)}}(\log n)^{t}
   =n^{\frac{d-2\beta}{2(2\beta+d)}}(\log n)^{2t},
  \qquad
  \delta_{n}^{2}\sqrt n\asymp n^{\frac{d-2\beta}{2(2\beta+d)}}(\log n)^{2t}.
\]
Hence \eqref{eq:clearance} holds. Finally, for every $\delta\ge\delta_n$, 
since $\Pnloc{\delta/\sqrt{c_0}}\subseteq\Pnsieve$,
\[
  J_{[\,]}\bigl(\delta,\Pnloc{\delta/\sqrt{c_0}},\He\bigr)
  \le J_{[\,]}(\delta,\Pnsieve,\He)
  \le\frac{\delta}{\delta_n}\,J_{[\,]}(\delta_n,\Pnsieve,\He)
  \le C_2\sqrt n\,\delta_n^{2}\Bigl(\frac{\delta}{\delta_n}\Bigr),
\]
which is \eqref{eq:entropy} with $K=C_2$ and $\alpha_0=1$, for all large $n$ and every $\delta\ge\delta_n$, in particular on $[\delta_n,\delta_n^\star]$.

\end{proof}

\imputationresult*

\begin{proof}



We first note that $\He(p_0, p_n)=O_P(r_n)$ implies $\HeMAR(p_0, p_n)=O_P(r_n)$.
Let $\mathbb{q}_n (x, m) = \hat{p}_n(x) \PMmxm$. We have
    \begin{equation}
        \begin{split}
            \mathbb{p}_0 (x, m) &= p_0 (x) \PMmxm, \\
            \hat{\mathbb{p}}_n (x, m) &= \hat{p}_n (x^{\bar m} \mid x^{(m)}) p_0^{(m)} (x^{(m)}) \PMmxm, \\
            \mathbb{q}_n (x, m) &= \hat{p}_n (x) \PMmxm.
        \end{split}
    \end{equation}
    
    \begin{equation*}
        \begin{split}
            \He (\mathbb{p}_0, \hat{\mathbb{p}}_n) &= \left(\sum_m \int \big(\sqrt{\mathbb{p}_0(x, m)} - \sqrt{\hat{\mathbb{p}}_n(x, m)}\big)^2 \dir x\right)^{1/2} \\
            &= \left(\sum_m \int \big(\sqrt{\mathbb{p}_0(x, m)} - \sqrt{\mathbb{q}_n (x, m)} + \sqrt{\mathbb{q}_n (x, m)} - \sqrt{\hat{\mathbb{p}}_n(x, m)}\big)^2 \dir x\right)^{1/2} \\
            &\overset{\text{(i)}}{\le} \underbrace{\left(\sum_m \int \big(\sqrt{\mathbb{p}_0(x, m)} - \sqrt{\mathbb{q}_n (x, m)}\big)^2 \dir x\right)^{1/2}}_{h_1} + \underbrace{\left(\sum_m \int \big(\sqrt{\mathbb{q}_n(x, m)} - \sqrt{\mathbb{p}_n (x, m)}\big)^2 \dir x\right)^{1/2}}_{h_2} 
        \end{split}
    \end{equation*}

    (i) uses the triangular inequality of the $\|\cdot\|_{\mathcal{H}}$ norm as defined in the proof~\eqref{eq:HtIsNorm}.

    \begin{equation*}
        \begin{split}
            h_1 &= \left(\sum_m \int \big(\sqrt{p_0(x)} - \sqrt{\hat{p}_n (x)}\big)^2 \PMmxm \dir x\right)^{1/2} \\
                &= \left(\int \big(\sqrt{p_0(x)} - \sqrt{\hat{p}_n (x)}\big)^2 \dir x \right)^{1/2} \\
                &= \He(\hat{p}_n, p_0) \vspace{10pt} \\ 
            h_2 &= \left(\sum_m \int \big(\sqrt{\hat{p}_n^{(m)} (x^{(m)})} - \sqrt{p_0^{(m)} (x^{(m)})}\big)^2 \hat{p}_n (x^{\bar m} \mid x^{(m)}) \PMmxm \dir x\right)^{1/2} \\
                &= \left(\sum_m \int \big(\sqrt{\hat{p}_n^{(m)} (x^{(m)})} - \sqrt{p_0^{(m)} (x^{(m)})}\big)^2 \PMmxm \dir x^{(m)}\right)^{1/2} \\
                &= \HeMAR (\hat{p}_n, p_0)
        \end{split}
    \end{equation*}

\end{proof}

\subsection{Proofs of Section \ref{sec_MNAR}}


\begin{lemma}\label{lemma_aeproperty}
    Assume $p \in \densityspace$ has $p(x) > 0$, a.e. Then for all $m$, the set $E_m:=\{x:p^{(m)}(x^{(m)})=0\}$ has Lebesgue-measure 0, or $p^{(m)}(x^{(m)}) > 0$ a.e.
\end{lemma}

\begin{proof}
     Since by Tonelli, $$P(E_m)=\int\mathbf 1\{p^{(m)}(x^{(m)})=0\}\,p^{(m)}(x^{(m)})\,dx^{(m)}=0,$$ $p(x) > 0$ a.e. implies that $1\{p^{(m)}(x^{(m)})=0\}=0$ a.e.
\end{proof}

\propofPhi

\begin{proof}
(i) Immediately follows from the definition.

\emph{Proof of (ii)}
Let $p\in\mathcal{D}$ with $p(x) > 0$ a.e. Let $E_m:=\{x:p^{(m)}(x^{(m)})=0\}$, as in Lemma \ref{lemma_aeproperty}, so that $P(E_m^c)=P_0(E_m^c)=1$. Thus we have that
\begin{align*}
  \E_P\bigl[\score{p}\bigr]
  &=\sum_m\int\!\!\int \frac{\mathbb{p}_0(x^{(m)},m)}{p^{(m)}(x^{(m)})}\,p(x)\,dx^{(\bar m)}\,dx^{(m)}\\
  &=\sum_m\int_{E_m^c} \frac{\mathbb{p}_0(x^{(m)},m)}{p^{(m)}(x^{(m)})}\,p^{(m)}(x^{(m)})\,dx^{(m)}\\
  &=\sum_m\int_{E_m^c} \mathbb{p}_0(x^{(m)},m)\,dx^{(m)}
   =\sum_m \mathbb{P}(M=m) = 1.
\end{align*}
\end{proof}

\solutioncharacterization*

\begin{proof}
We start with some preliminaries. First, for each $m$ with $\Prob(M=m) > 0$, we define the sets
\[
A_{m} = \{x^{(m)}: \mathbb{p}(x^{(m)},m) > 0 \}.
\]
We note that for any $p$ such that $L(p) > -\infty$, $p^{(m)}(x^{(m)}) > 0$ for all $x^{(m)} \in A_m$ up to null sets. Since for any measurable function $f$, $\int f(x^{(m)}) \mathbb{p}(x^{(m)},m) d x = \int_{A_m} f(x^{(m)}) \mathbb{p}(x^{(m)},m) d x^{(m)}$, whenever $f$ involves $p^{(m)}(x^{(m)})$, we might take it to be strictly larger zero and implicitly understand integration over $A_m$. Moreover, due to Assumption \ref{asm_positive_MDM}, it holds that $\mathbb{p}(x,\mathbf{0}) \geq c_0 p_0(x) > 0$ a.e. Thus, if a property holds for all $x \in A_{\mathbf{0}}$, it holds a.e. 


\smallskip

\emph{Proof of (i)} 
Immediate as $p\mapsto p^{(m)}$ is linear and $\log$ is concave.

\smallskip
\emph{Proof of \textup{(ii-a)}.}
Suppose $\score{p}(x) = 1$ a.e. 
Then for all $m$, $\mathbb{p}_0(x^{(m)},m)\le p^{(m)}(x^{(m)})$ a.e., and $p^{(m)}(x^{(m)})>0$
for $x^{(m)} \in A_m$ up to a null set.
Let $q\in\densityspace$. If $L(q)=-\infty$, then $L(p)\ge L(q)$ trivially, so assume $L(q)$ is finite, which again implies $q^{(m)}(x^{(m)})>0$
for $x^{(m)} \in A_m$ up to null sets. Applying $-\log t\ge 1-t$ at $t=q^{(m)}(x^{(m)})/p^{(m)}(x^{(m)})$,
\begin{align*}
  L(p)-L(q)
  &=\sum_m\int \log\left(\frac{p^{(m)}(x^{(m)})}{q^{(m)}(x^{(m)})}\right)\;\mathbb{p}_0\bigl(x^{(m)},m\bigr)\,dx^{(m)}\\
  &\ge\sum_m\int\Bigl(1-\frac{q^{(m)}(x^{(m)})}{p^{(m)}(x^{(m)})}\Bigr)\mathbb{p}_0\bigl(x^{(m)},m\bigr)\,dx^{(m)}\\
  &=\sum_m\int \mathbb{p}_0\bigl(x^{(m)},m\bigr)\,dx^{(m)}
    -\sum_m\int \frac{q^{(m)}(x^{(m)})}{p^{(m)}(x^{(m)})}\,\mathbb{p}_0\bigl(x^{(m)},m\bigr)\,dx^{(m)}\\
  &= 1-\int q(x)\,\score{p}(x)\,dx\;=\;0,
\end{align*}
where the last equality uses \eqref{eq:stat}. Since $q\in\densityspace$ was arbitrary, $p$ maximizes $L$ over $\densityspace$, hence over $\densityspacemodel\subseteq\densityspace$, and $p\in\Solutionspace$.

\smallskip
\emph{Proof of \textup{(ii-b)}.}
Let $p\in\Solutionspace$. Then $L(p)$ is finite, and thus $p^{(m)}(x^{(m)})>0$
for $x^{(m)} \in A_m$\ for every $m$.
Let $q\in\mathscr{P}$ and, for $t\in(0,1]$,
$p_t:=(1-t)p+tq$, so that $p_t^{(m)}=(1-t)p^{(m)}+tq^{(m)}$. 
By convexity $p_t\in\mathscr{P}$, hence $L(p_t)\le L(p)$. And since $p_t^{(m)}(x^{(m)})\ge(1-t)p^{(m)}(x^{(m)})$ a.e., $L(p_t) > -\infty$. Set
\[
  D_t^{(m)}(x^{(m)}) := \frac{\log p_t^{(m)}(x^{(m)})-\log p^{(m)}(x^{(m)})}{t}.
\]
Since $t\mapsto\log p_t^{(m)}(x^{(m)})$ is concave, $D_t^{(m)} \uparrow \frac{q^{(m)}-p^{(m)}}{p^{(m)}}$ as $t\downarrow 0$ 
for $x^{(m)} \in A_m$.
For $t\le\tfrac12$, $D_t^{(m)}\ge D_{1/2}^{(m)}\ge-2\log2$ holds as $p_{1/2}^{(m)}\ge\tfrac12p^{(m)}$.
By monotone convergence applied to the nonnegative family $D_t^{(m)}+2\log 2$, $t\in(0,\tfrac12]$ and Tonelli,
\begin{align*}
  0\;\ge\;\lim_{t\downarrow0}\frac{L(p_t)-L(p)}{t}
  &=\lim_{t\downarrow0}\sum_m\int D_t^{(m)}(x^{(m)})\,\mathbb{p}_0\bigl(x^{(m)},m\bigr)\,dx^{(m)}\\
  &=\sum_m\int\frac{q^{(m)}(x^{(m)})-p^{(m)}(x^{(m)})}{p^{(m)}(x^{(m)})}\,
     \mathbb{p}_0\bigl(x^{(m)},m\bigr)\,dx^{(m)}\\
  &=\int q(x)\,\score{p}(x)\,dx-\sum_m\int\mathbb{p}_0\bigl(x^{(m)},m\bigr)\,dx^{(m)}\\
   &=\;\int q(x)\,\score{p}(x)\,dx-1,
\end{align*}
where we note that again $$\sum_{m} \int q^{(m)}(x^{(m)}) \frac{\mathbb{p}(x^{(m)}, m)}{p(x^{(m)})}\,dx=\int q(x)\,\score{p}(x)\,dx.$$
\noindent
Consider now $\densityspacemodel=\densityspace$. Assume now $P_0(\score{p}(X)>1) >0$ and let $\lambda$ be the Lebesgue measure on $\R^d$. Then, we may write
\[
\{x:\score{p}(x)>1\}=\bigcup_{k} \{x:\score{p}(x) \geq 1+1/k \} \cap \{x: \|x \|_2 \leq k\}=\lim_{k \to \infty} B_k,
\]
thus there exists $K$ such that $P_0(B_K) >0$ and thus also $\lambda(B_K) >0$, while boundedness ensures that $\lambda(B_K) < \infty$. Thus,
\[
q(x)=\frac{\one\{x \in B_K\}}{\lambda(B_K)},
\]
is a valid density and
\[
\int q(x)\score{p}(x) dx = \frac{1}{\lambda(B_K)} \int_{B_K} \score{p}(x) dx > 1,
\]
a contradiction. Hence $\score{p}\le1$ a.e.
Its $m=\mathbf 0$ summand gives
\[
  p(x)\;\ge\;\mathbb{p}_0(x,\mathbf 0)\;\ge\;c_0\,p_0(x)\;>\;0\quad\text{a.e.},
\]
or $p(x) > 0$ a.e. Thus Lemma~\ref{lem:propofPhi}(ii) gives $\int p(x)\,(1-\score{p}(x))\,dx=0$ with a nonnegative integrand, so $p(x)\,(1-\score{p}(x))=0$ a.e. Since $p(x)>0$ a.e., $\score{p}(x)=1$ a.e., which is \eqref{eq:stat}.
 
\smallskip
\emph{Proof of \textup{(iii)}.}
Let $p_1,p_2\in\Solutionspace$, so that $L(p_1)=L(p_2)=V:=\max_{P\in\mathscr{P}}L(p)$.
Since $L(p_j)>-\infty$, $p_j^{(m)}(x^{(m)})>0$ for $x^{(m)} \in A_m$\ for $j=1,2$ and every $m$.
By convexity, $\bar p:=\tfrac12(p_1+p_2)\in\mathscr{P}$, so maximality of $V$ and \textup{(i)}
give $V\ge L(\bar p)\ge\tfrac12L(p_1)+\tfrac12L(p_2)=V$. Hence
\begin{align*}
0&=L(\bar{p}) - (\tfrac12L(p_1)+\tfrac12L(p_2))\\
&=\sum_m\int\log\frac{\tfrac12\bigl(p_1^{(m)}(x^{(m)})+p_2^{(m)}(x^{(m)})\bigr)}{\sqrt{p_1^{(m)}(x^{(m)})p_2^{(m)}(x^{(m)})}}
  \;\mathbb{p}_0\bigl(x^{(m)},m\bigr)\,dx^{(m)} .
\end{align*}
By the AM--GM inequality every integrand is nonnegative, so the $m=\mathbf 0$ integrand
vanishes on $A_{\mathbf{0}}$, i.e.\ $p_1(x)=p_2(x)$ a.e.

\smallskip
\emph{Proof of \textup{(iv)}.}
By Lemma~\ref{lem:propofPhi}(i), $\score{p_0}=W$. If $W=1$ a.e., then $p_0\in\Solutionspace$ by \ref{ii-a}, and $\Solutionspace=\{p_0\}$ by \textup{(iii)}. Conversely, if $\Solutionspace=\{p_0\}$, then \ref{ii-b} gives $W=\score{p_0}=1$ a.e.

\smallskip
\emph{Proof of \textup{(v)}.}
Denote as before, 
\[\E_{\Pjoint}\bigl[f(X^{(M)},M)\bigr]=\sum_m\int f\bigl(x^{(m)},m\bigr)\,
\mathbb{p}_0\bigl(x^{(m)},m\bigr)\,dx^{(m)}.\] 
Let $p\in\mathscr{P}$. Since $L(p_\star) > -\infty$, $p_\star^{(m)}(x^{(m)})>0$
for $x^{(m)} \in A_m$, so $\rho_{p,p_\star}^{(m)}(x^{(m)})$ is well defined for $x^{(m)} \in A_m$ up to null sets. By Tonelli and \textup{(ii-b)},
\begin{equation}\label{eq:rho-mean}
  \E_{\Pjoint}\bigl[\rho_{p,p_\star}^{(M)}(X^{(M)})\bigr]
  =\sum_m\int\frac{\mathbb{p}_0\bigl(x^{(m)},m\bigr)}{p_\star^{(m)}\bigl(x^{(m)}\bigr)}\,
    p^{(m)}\bigl(x^{(m)}\bigr)\,dx^{(m)}
  =\int p(x)\,\score{p_\star}(x)\,dx\;\le\;1 .
\end{equation}
If $L(p)=-\infty$ again there is nothing to prove, so assume $L(p)>-\infty$.
Applying $-\log t\ge1-t$ at $t=\sqrt{\rho_{p,p_\star}^{(M)}}$,
\begin{align*}
  L(p_\star)-L(p)
  &=\E_{\Pjoint}\bigl[-\log\rho_{p,p_\star}^{(M)}(X^{(M)})\bigr] \\
  &\ge2\,\E_{\Pjoint}\Bigl[1-\sqrt{\rho_{p,p_\star}^{(M)}(X^{(M)})}\Bigr]\\
  &=\E_{\Pjoint}\Bigl[\bigl(\sqrt{\rho_{p,p_\star}^{(M)}(X^{(M)})}-1\bigr)^2\Bigr]
    +\Bigl(1-\E_{\Pjoint}\bigl[\rho_{p,p_\star}^{(M)}(X^{(M)})\bigr]\Bigr)\\
  &\ge\E_{\Pjoint}\Bigl[\bigl(\sqrt{\rho_{p,p_\star}^{(M)}(X^{(M)})}-1\bigr)^2\Bigr]\\
  &=\sum_m\int\Bigl(\sqrt{p^{(m)}(x^{(m)})}-\sqrt{p_\star^{(m)} (x^{(m)})}\Bigr)^2
    \frac{\mathbb{p}_0\bigl(x^{(m)},m\bigr)}{p_\star^{(m)}\bigl(x^{(m)}\bigr)}\,dx^{(m)}\\
  &=\sum_m\int\Bigl(\sqrt{p^{(m)} (x^{(m)})}-\sqrt{p_\star^{(m)}(x^{(m)})}\Bigr)^2
    \rho_{p_0, p_{\star}}^{(m)}\bigl(x^{(m)}\bigr)\,\PMmxm\,dx^{(m)}\\
      &\ge\int\Bigl(\sqrt{p(x)}-\sqrt{p_\star(x)}\Bigr)^2
    \rho_{p_0, p_{\star}}\bigl(x\bigr)\,\PMzerox\,dx\\
  &\ge\kappa\,\HeMAR^2\bigl(p_\star,p\bigr).
\end{align*}

The second inequality uses \eqref{eq:rho-mean}; and the last inequality uses $\rho_{p_0,p_\star}(x)\ge\kappa$ for $x \in A_\mathbf{0}$ and thus a.e.

\smallskip
\emph{Proof of \textup{(vi)}.}
For every $p\in\densityspacemodel$, the expectation $\E_{\Pjoint}\bigl[-\log\rho_{p,p_\star}^{(M)}(X^{(M)})\bigr]$
is well defined, since $(\log t)^+\le t$ and \eqref{eq:rho-mean} make its negative part integrable. 
The chain in the proof of \textup{(v)} gives
\begin{equation}\label{eq:curv-ratio}
  c_0\kappa\,\He^2\bigl(p_\star,p\bigr)\le\E_{\Pjoint}\bigl[-\log\rho_{p,p_\star}^{(M)}(X^{(M)})\bigr],
  \qquad p\in\mathscr{P}.
\end{equation}
Take $p=p_0\in \densityspacemodel$. By symmetry of $\He$, Jensen's inequality, and Tonelli,
\begin{align*}
  \kappa c_0\,\He^2\bigl(p_0,p_\star\bigr)
  &\le\E_{\Pjoint}\Bigl[\log\frac{p_\star^{(M)} (X^{(M)})}{p_0^{(M)}(X^{(M)})}\Bigr]
   \le\log\E_{\Pjoint}\Bigl[\frac{p_\star^{(M)}(X^{(M)})}{p_0^{(M)}(X^{(M)})}\Bigr]\\
   &=\log\left(\sum_m\int p_\star^{(m)}\bigl(x^{(m)}\bigr)\,\PMmxm\,dx^{(m)}\right)\\
  &=\log\left(\sum_m\int p_\star^{(m)}\bigl(x^{(m)}\bigr)\,\PMmxm\,dx^{(m)}\right)\\
  &=\log\left(\int p_\star(x)\,W(x)\,dx\right)\\
  & \leq \int p_\star(x)\,(W(x)-1)\,dx\\
  &= \int (p_\star(x) - p_0(x))\,(W(x)-1)\,dx,
\end{align*}
where we used Tonelli and $\int p_0(x)W(x)\,dx=\sum_m\int\mathbb p_0(x^{(m)},m)\,dx^{(m)}=1$, in the last step. Using that fact that $p_\star(x)\le p_0(x)/\kappa$ a.e., and $(\sqrt{p_\star}(x) - \sqrt{p_0}(x))= (p_\star(x) - p_0(x))/(\sqrt{p_\star}(x) + \sqrt{p_0}(x))\geq (p_\star(x) - p_0(x))/\sqrt{p_0}(x)(1+\kappa^{-1/2}) $ and Cauchy-Schwarz, it follows that
\begin{align*}
    c_0 \kappa \He^2(p_{0}, p_{\star})\\
    &\leq \int (p_\star(x) - p_0(x))\,(W(x)-1)\,dx\\
    &\leq \int \sqrt{p_0}(x)(1+\kappa^{-1/2}) (\sqrt{p_\star}(x) - \sqrt{p_0}(x))\,(W(x)-1)\,dx\\
       &\leq (1+\kappa^{-1/2})\left(\int p_0(x)\,(W(x)-1)^2\,dx\right)^{1/2} \He(p_{0}, p_{\star}),\\
\end{align*}
showing that 
\[
c_0 \kappa \He(p_{0}, p_{\star}) \leq (1+\kappa^{-1/2})\Var_{P_0}(W(X))^{1/2}.
\]
\end{proof}

\MARRepresentationCharacterization*

\begin{proof}
We show below that in both directions $p(x)\ge\mathbb p_0(x,\mathbf 0)$
for a.e.\ $x\in\mathcal X$. Given this, Assumption~\ref{asm_positive_MDM} yields
\[
  p(x)\ \ge\ \mathbb p_0(x,\mathbf 0)\ =\ p_0(x)\,\PMzerox\ \ge\ c_0\,p_0(x)\ >\ 0
  \text{ a.e.},
\]
and hence by Lemma \ref{lemma_aeproperty} for all $m$, $p^{(m)}(x^{(m)})>0$ a.e.

\emph{($\Rightarrow$) Suppose $(p,\nu)$ is a MAR representation of $\Pjoint$.}
By assumption, this means that $\sum_m v_{m}(x_m)=1$, a.e. In particular, this implies $\nu_{\mathbf 0}(x)\le1$ a.e., or $\mathbb p_0(x,\mathbf 0)=\nu_{\mathbf 0}(x)\,p_0(x)\le p(x)$; so the preliminary step applies and $p^{(m)}(x^{(m)})>0$ a.e. This determines $\nu$ from $p$: For all $m$ and $a.e.$,
\begin{equation}\label{eq:nuforced}
  \nu_m\bigl(x^{(m)}\bigr)\;=\;\frac{\mathbb{p}_0\bigl(x^{(m)},m\bigr)}{p^{(m)}\bigl(x^{(m)}\bigr)} . 
\end{equation}
Summing \eqref{eq:nuforced} over $m$ and using $\sum_m\nu_m(x^{(m)})=1$ gives $\score{p}(x)=\sum_m\nu_m(x^{(m)})=1$ a.e.

\emph{($\Leftarrow$) Suppose $\score{p}=1$ a.e.}
The $m=\mathbf 0$ summand gives $\mathbb p_0(x,\mathbf 0)/p(x)\le\score{p}(x)=1$ (with
$c/0:=+\infty$ for $c>0$), that is, $\mathbb p_0(x,\mathbf 0)\le p(x)$; so the preliminary
step applies and $p^{(m)}(x^{(m)})>0$ a.e. Define $\nu_m$ by \eqref{eq:nuforced}. Each $\nu_m$ is then a measurable function of $x^{(m)}$ alone and is nonnegative; it is bounded by $1$ a.e.\ because it is one of the summands of $\score{p}$, which are nonnegative and sum to $1$ a.e.\ (we redefine $\nu_m$ on the remaining null set). That same identity gives $\sum_m\nu_m(x^{(m)})=\score{p}(x)=1$ a.e.,
so $\nu$ is admissible. Finally, factoring $\nu_m$ out of the integral as above,
\[
  \int p(x)\,\nu_m\bigl(x^{(m)}\bigr)\,dx^{(\bar m)}
  \;=\;\nu_m\bigl(x^{(m)}\bigr)\,p^{(m)}\bigl(x^{(m)}\bigr)
  \;=\;\mathbb{p}_0\bigl(x^{(m)},m\bigr) ,
\]
so $(p,\nu)$ is a MAR representation of $\mathbb{p}_0(x^{(m)}, m)$. In either direction $\nu$ is the one given by \eqref{eq:nuforced}, so it is unique up to null sets.
\end{proof}   

\Marlikeresult*

\begin{proof}
    By Lemma~\ref{lem:propofPhi}(i), $W=\score{p_0}$, so the assumption states $\score{p_0}(x)=1$ a.e. 
    and by Proposition~\ref{prop:MARRepresentationCharacterization} with $p=p_0$, $(p_0,\nu)$ is a MAR representation, where
    \[
    \nu_m(x^{(m)})= \mathbb{p}_0(x^{(m)}, m)/p_0^{(m)}(x^{(m)})
     = \mathbb{P}(M=m\mid X^{(m)} = x^{(m)})
    \] on $\{x^{(m)}: p_0^{(m)}(x^{(m)}) > 0\}$ and 0 elsewhere.  So $\nu_m$ equals the observed-data propensity of the true law.
    These are measurable functions of $X^{(m)}$ with $0\le\nu_m\le1$ and $\sum_m \nu_m = 1$.
    We use $\nu$ to define a joint density of $(X,M)$ as 
    \[
    \check{f}(x,m) := p_0(x)\nu_m(x^{(m)}).
    \]
    Two properties hold for $\check{f}$:
    (1) Its $X$-marginal density is $p_0(x)$, as 
    \[
    \check{p}(x) := \sum_m \check{f}(x,m) = p_0(x) \sum_m \nu_m(x^{(m)}) = p_0(x).
    \]
    (2) Its observed data density is $\mathbb{p}_0(x^{(m)}, m)$, as 
    \[
    \check{\mathbb{p}}_0(x^{(m)}, m) := 
    \int \check{f}(x,m) dx^{(\bar{m})}
    = \nu_m(x^{(m)}) \int p_0(x) dx^{(\bar{m})}
    = \nu_m(x^{(m)}) p_0^{(m)}(x^{(m)})
    = \mathbb{p}_0(x^{(m)}, m),
    \]
    where the last equality holds on $\{x^{(m)}: p_0^{(m)}(x^{(m)}) > 0\}$ by definition of $\nu_m$, and off it both sides vanish since $\mathbb p_0(\cdot,m)\le p_0^{(m)}$.
    Let $\check{\mathbb{P}}_{(X,M)}$ denote the probability measure with density $\check{f}$. 
    First, it satisfies MAR (Assumption~\ref{asm_true_MDM}). Indeed, let $E_m:=\{x:p_0^{(m)}(x^{(m)})=0\}$. \Cref{lemma_aeproperty} implies that $E:=\bigcup_m E_m$ has Lebesgue measure zero. For $x\notin E$ both conditionals below are ratios with positive denominators, and for any $m$,
    \[
      \check{\mathbb P}(M=m\mid X=x)
      =\frac{\check f(x,m)}{p_0(x)}=\nu_m\bigl(x^{(m)}\bigr)
      =\check{\mathbb P}\bigl(M=m\mid X^{(m)}=x^{(m)}\bigr).
    \]
    Second, it satisfies Assumption \ref{asm_positive_MDM}. Indeed, at $m=\mathbf 0$ nothing is integrated out, so $\mathbb{P}(M = \mathbf{0} \mid X = x) \geq c_0 > 0$ implies 
    $\check{\mathbb{P}}(M= \mathbf{0} \mid X=x)=p_0(x)\mathbb{P}(M = \mathbf{0} \mid X = x)/p_0(x)=\mathbb{P}(M = \mathbf{0} \mid X = x) \ge c_0$ for a.e.\ $x$ with the same $c_0$.
    Third, it satisfies Assumption \ref{asm_no_empty_MDM}. Indeed, at $m=\mathbbm{1}$ nothing is observed, so $\nu_{\mathbbm{1}}$ is a constant
    and $p_0^{(\mathbbm{1})}=\int p_0(x)\,dx=1$. Hence
    \[
      \check{\mathbb P}(M=\mathbbm{1})
      =\int p_0(x)\,\nu_{\mathbbm{1}}\,dx
      =\nu_{\mathbbm{1}}
      =\int p_0(x)\,\mathbb P(M=\mathbbm{1}\mid X=x)\,dx
      =\mathbb P(M=\mathbbm{1})
      =0
    \]
    as $\mathbb{P}(M = \mathbbm{1}) = 0$.
    Hence $\check{\mathbb P}_{(X,M)}$ satisfies Assumptions~\ref{asm_true_MDM}--\ref{asm_positive_MDM}.

    It remains to verify the other hypotheses of Theorems~\ref{thm:8mar} and~\ref{Thm:DensityResult} for $\check{\mathbb P}_{(X,M)}$, and to transfer their conclusions back to $\mathbb P_{(X,M)}$. 
    By (1) and (2), the two laws have the same $X$-marginal $p_0$ and the same observed-data law $\mathbb p_0$, and every hypothesis and conclusion of Theorems~\ref{thm:8mar} and~\ref{Thm:DensityResult} other than Assumptions~\ref{asm_true_MDM}--\ref{asm_positive_MDM} is determined by this pair. 
    The observed sample is i.i.d.\ $\mathbb p_0$ under both laws, so $\hat p_n$ and \eqref{eq:approxmax} are the same statements. 
    $\HeMAR$ is weighted by the observed-block propensities, which coincide under the two laws (First above), so \eqref{eq:approx} and the $\HeMAR$-conclusion coincide. 
    $\KLMAR$ and the expectation in \eqref{eq:tail} are expectations of functions of $(X^{(M)},M)\sim\mathbb p_0$. 
    The constant $c_0$ is the same (Second above), and $\He$, \eqref{eq:entropy} and Assumption~\ref{asm_Holder} involve only $p_0$ and the sieve. 
    Thus each of these hypotheses and conclusions is the same statement whether it is read under $\mathbb P_{(X,M)}$ or under $\check{\mathbb P}_{(X,M)}$.
\end{proof}

\beyondW*

\begin{proof}
For a density $p$, write $\nu_{p,m}$ for \eqref{eq:nuforced} with its dependence on $p$ made explicit, and $\HeMAR_p$ for the adapted Hellinger functional with weights $\nu_{p,m}$ in place of $\PMmxm$.
Then $\nu_{p_0,m}=\PMmxm$, $\sum_m\nu_{p_0,m}=W$ and $\HeMAR_{p_0}=\HeMAR$. By the triangle inequality,
$\He(p_0,\hat p_n)\le\He(p_\star,\hat p_n)+\He(p_\star,p_0)$, so it suffices to bound the estimation term (Step~1) and the bias term (Step~2), and then to specialize the bias bound (Steps~3--4). Step~0 collects two facts used throughout.

\emph{Step 0: a lower bound on $p_\star$.}
Each summand of $\score{p_\star}$ is nonnegative, hence $\mathbb{p}_0(x^{(m)},m)/p_{\star}(x^{(m)}) \leq 1$ a.e. At $m=\mathbf 0$ nothing is integrated out, so by Assumption~\ref{asm_positive_MDM},
\begin{equation}\label{eq:pstarlower}
  p_\star(x)\ \ge\ \mathbb p_0(x,\mathbf 0)\ =\ p_0(x)\,\PMzerox\ \ge\ c_0\,p_0(x)\ 
  \qquad\text{a.e.\ } .
\end{equation}
Thus $\{x: p_{0}(x) > 0\} \subset \{x: p_{\star}(x) > 0\}$ up to null sets and by the definition of $\kappa$, 
\begin{equation}\label{eq:nu0lower}
  \nu_{p_\star,\mathbf 0}(x)=\frac{\mathbb p_0(x,\mathbf 0)}{p_\star(x)}
  =\PMzerox\,\frac{p_0(x)}{p_\star(x)}\ \ge\ c_0\kappa
  \qquad\text{a.e.\ } 
\end{equation}

\emph{Step 1: estimation term,
$\He(p_\star,\hat p_n)=O_P\bigl((c_0\kappa)^{-1/2}(\delta_n\vee n^{-1/2})\bigr)$.}
The argument is that of Corollary~\ref{cor:MARlike}, with the MAR representation $(p_0,\nu_{p_0})$ replaced by $(p_\star,\nu_{p_\star})$.

\emph{The MAR counterpart.} By
Proposition~\ref{prop:MARRepresentationCharacterization}, $0\le\nu_{p_\star,m}\le1$ and $\sum_m\nu_{p_\star,m}=\score{p_\star}=1$ a.e. 
Define the joint density
\[
  \check f_\star(x,m):=p_\star(x)\,\nu_{p_\star,m}\bigl(x^{(m)}\bigr),
\]
and let $\check{\Prob}_\star$ be the corresponding law of $(X,M)$. Two
properties hold.
(1) Its $X$-marginal density is $\sum_m\check f_\star(x,m)=p_\star(x)\score{p_\star}(x)=p_\star(x)$ a.e.
(2) Its observed-data density is $\mathbb p_0$: since $\nu_{p_\star,m}$ depends on $x^{(m)}$ only,
\[
  \check{\mathbb p}_\star\bigl(x^{(m)},m\bigr):=\int\check f_\star(x,m)\,dx^{(\bar m)}
  =\nu_{p_\star,m}\bigl(x^{(m)}\bigr)\,p_\star^{(m)}\bigl(x^{(m)}\bigr)
  =\mathbb p_0\bigl(x^{(m)},m\bigr),
\]
where the last equality holds on $\{x^{(m)}: p_\star^{(m)}(x^{(m)})>0\}$ by definition of
$\nu_{p_\star,m}$; off it both sides vanish, because
$\mathbb p_0(x^{(m)},m)\le p_\star^{(m)}(x^{(m)})$ a.e.\ (as each summand of
$\score{p_\star}$ is at most $1$).

\emph{The standing assumptions hold for $\check{\Prob}_\star$.}
First, MAR holds. By \eqref{eq:pstarlower}, $p_\star(x)>0$ a.e., so it follows from \Cref{lemma_aeproperty} that 
the set $E_\star:=\bigcup_m\{x:p_\star^{(m)}(x^{(m)})=0\}$ has Lebesgue measure zero. For $x\notin E_\star$ and every $m$,
\[
  \check{\Prob}_\star(M=m\mid X=x)=\frac{\check f_\star(x,m)}{p_\star(x)}
  =\nu_{p_\star,m}\bigl(x^{(m)}\bigr)
  =\frac{\check{\mathbb p}_\star(x^{(m)},m)}{p_\star^{(m)}(x^{(m)})}
  =\check{\Prob}_\star\bigl(M=m\mid X^{(m)}=x^{(m)}\bigr),
\]
using (2) and, by (1), that the $m$-marginal of the $X$-density of
$\check{\Prob}_\star$ is $p_\star^{(m)}$. So $\check{\Prob}_\star$ is MAR, its observed-block propensities are $\nu_{p_\star,m}$, and its adapted Hellinger functional is $\HeMAR_{p_\star}$.
Second, positivity: by \eqref{eq:nu0lower},
$\check{\Prob}_\star(M=\mathbf 0\mid X=x)=\nu_{p_\star,\mathbf 0}(x)\ge c_0\kappa$ for a.e.\ $x$, so Assumption~\ref{asm_positive_MDM} holds for $\check{\Prob}_\star$ with constant $c_0\kappa$.
Third, the empty pattern: $x^{(\mathbbm 1)}$ is empty and
$p_\star^{(\mathbbm 1)}=\int p_\star=1$, so
\[
  \check{\Prob}_\star(M=\mathbbm 1)=\int p_\star(x)\,\nu_{p_\star,\mathbbm 1}\,dx
  =\nu_{p_\star,\mathbbm 1}=\mathbb p_0(\cdot,\mathbbm 1)=\Prob(M=\mathbbm 1)=0
\]
by Assumption~\ref{asm_no_empty_MDM}. Hence $\check{\Prob}_\star$ satisfies Assumptions~\ref{asm_true_MDM}--\ref{asm_positive_MDM}, with $X$-marginal $p_\star$ and constant $c_0\kappa$.

\emph{Transfer.} The laws $\check{\Prob}_\star$ and $\Pjoint$ have the same observed-data law $\mathbb p_0$, so the observed sample is i.i.d.\ $\mathbb p_0$ under both, and $\hat p_n$ and \eqref{eq:approxmax} are the same statements.
Applied to $\check{\Prob}_\star$, Theorem~\ref{thm:8mar} has target $p_\star$, adapted Hellinger functional $\HeMAR_{p_\star}$ and constant $c_0\kappa$, 
all determined by the pair $(p_\star,\mathbb p_0)$. 
$\KLMAR$ and the expectation in \eqref{eq:tail} are expectations under $\mathbb p_0$ of functions of $p_\star^{(m)}$ and the sieve, since $p_\star^{(m)}(x^{(m)})\nu_{p_\star,m}(x^{(m)})=\mathbb p_0(x^{(m)},m)$; $\He$, \eqref{eq:entropy} and Assumption~\ref{asm_Holder} involve only $p_\star$ and the sieve. 
Thus every
hypothesis and conclusion of Theorems~\ref{thm:8mar} and~\ref{Thm:DensityResult}, stated with $p_0$, $\HeMAR$ and $c_0$ replaced by $p_\star$, $\HeMAR_{p_\star}$ and $c_0\kappa$, is the same statement under $\check{\Prob}_\star$ and under
$\Pjoint$. 
Both theorems apply to $\check{\Prob}_\star$, so their conclusions hold under $\Pjoint$: $\He(p_\star,\hat p_n)=O_P\bigl((c_0\kappa)^{-1/2}(\delta_n\vee
n^{-1/2})\bigr)$, and, if Assumption~\ref{asm_Holder} holds for $p_\star$, $\He(p_\star,\hat p_n)$ attains the rate of Theorem~\ref{Thm:DensityResult} with
$c_0$ replaced by $c_0\kappa$.

\emph{Step 2: Bias term.}

Similar to Theorem~\ref{thm:new}, by \eqref{eq:nu0lower},
 \begin{align*}
 c_0\kappa\,\He^2\bigl(p_\star,p_0\bigr)  &\leq \int\Bigl(\sqrt{p_0(x)}-\sqrt{p_\star(x)}\Bigr)^2  \frac{\mathbb{p}_0\bigl(x,\mathbf{0}\bigr)}{p_\star\bigl(x\bigr)}\,dx \\
 & \leq \sum_m\int\Bigl(\sqrt{p^{(m)}(x^{(m)})}-\sqrt{p_\star^{(m)} (x^{(m)})}\Bigr)^2
    \frac{\mathbb{p}_0\bigl(x^{(m)},m\bigr)}{p_\star^{(m)}\bigl(x^{(m)}\bigr)}\,dx^{(m)}\\
    &=\E_{\Pjoint}\Bigl[\bigl(\sqrt{\rho_{p_0,p_\star}^{(M)}(X^{(M)})}-1\bigr)^2\Bigr]\\
    &=\E_{\Pjoint}\Bigl[\bigl(\sqrt{\rho_{p_0,p_\star}^{(M)}(X^{(M)})}-1\bigr)^2\Bigr]
    +\Bigl(1-\E_{\Pjoint}\bigl[\rho_{p_0,p_\star}^{(M)}(X^{(M)})\bigr]\Bigr)\\
    &=2\,\E_{\Pjoint}\Bigl[1-\sqrt{\rho_{p_0,p_\star}^{(M)}(X^{(M)})}\Bigr]\\
    &\leq \E_{\Pjoint}\bigl[-\log\rho_{p_0,p_\star}^{(M)}(X^{(M)})\bigr],
\end{align*}
where we used $\E_{\Pjoint}\bigl[\rho_{p_0,p_\star}^{(M)}(X^{(M)})\bigr]=\int p_0(x)\,\score{p_\star}(x)\,dx=1$. Now, just as in Theorem \ref{thm:new}, we can bound,
\begin{align}\label{eq_Variancebound2}
     c_0\kappa\,\He^2\bigl(p_\star,p_0\bigr) \leq (1+\kappa^{-1/2})\Var_{P_0}(W(X))^{1/2}.
\end{align}
This argument does not use Theorem~\ref{thm:new}\,(vi), so it needs neither convexity of $\densityspacemodel$ nor $P_0\in\Modelspace$. Together with Step~1, it proves (i), and the H\"older statement.
Under \eqref{eq_Contamination},
$\E_{P_0}[W^{\scriptscriptstyle\mathrm{MNAR}}]=\sum_m\int\PMNAR(M=m\mid x)\,p_0(x)\,dx=1$, so Step~3 gives $\Var_{P_0}(W)=\varepsilon^2\Var_{P_0}(W^{\scriptscriptstyle\mathrm{MNAR}})$.

\emph{Step 3: contamination.} By definition of $W$,
$\PMmxm=\int\Prob(M=m\mid x)\,p_0(x^{(\bar m)}\mid x^{(m)})\,dx^{(\bar m)}$.
Inserting \eqref{eq_Contamination}, the MAR term is unchanged, since
$\PMAR(M=m\mid X=x)$ depends on $x^{(m)}$ only. Summing over $m$ and using $\sum_m\PMAR(M=m\mid X=x)=1$ gives $W(x)=(1-\varepsilon)+\varepsilon\,W^{\scriptscriptstyle\mathrm{MNAR}}(x)$ a.e., that is, $W-1=\varepsilon(W^{\scriptscriptstyle\mathrm{MNAR}}-1)$. Hence
$\Var(W)=\varepsilon\Var(W^{\scriptscriptstyle\mathrm{MNAR}}-1)$, and inserting this into \eqref{eq_Variancebound2} and Step~1 proves (ii).

\emph{Step 4: positivity of the MAR component.} By \eqref{eq_Contamination} and
$\varepsilon<1$, $\Prob(M=\mathbf 0\mid x)\ge(1-\varepsilon)\PMAR(M=\mathbf 0\mid x)\ge(1-\varepsilon)c_0>0$
a.e., so Assumption~\ref{asm_positive_MDM} holds with constant
$(1-\varepsilon)c_0$. Steps~0--3 used Assumption~\ref{asm_positive_MDM} only through this constant, and $\kappa$ does not involve it, which proves (iii).

\end{proof}

\clearpage
{\small
\bibliographystyle{apalike}
\bibliography{reference}
}


\end{document}